\documentclass[letterpaper,11pt]{article}

\usepackage[algo2e,linesnumbered,ruled,vlined] {algorithm2e}
\usepackage{stix}
\usepackage[T1]{fontenc}
\usepackage{amsmath,amsfonts,amsthm}
\usepackage{mathtools}
\usepackage{thmtools}
\usepackage{enumitem}
\usepackage[normalem]{ulem}
\usepackage[margin=1in]{geometry}
\usepackage{multirow}
\usepackage{comment}
\usepackage{complexity}
\usepackage{subcaption}
\usepackage{float}
\usepackage{thm-restate}
\usepackage{authblk}

\SetCommentSty{mycommfont}

\usepackage[dvipsnames,usenames]{xcolor}
\usepackage[colorlinks=true,pdfpagemode=UseNone,urlcolor=RoyalBlue,linkcolor=RoyalBlue,citecolor=OliveGreen,pdfstartview=FitH,linktocpage]{hyperref}

\usepackage{dsfont}
\usepackage{xspace}
\usepackage{tikz}
\usepackage{tcolorbox}
\usepackage[nameinlink, capitalise]{cleveref}
\crefname{section}{Section}{Sections}
\crefname{subsection}{Subsection}{Subsections}
\crefname{figure}{Figure}{Figures}
\crefname{claim}{Claim}{Claims}
\crefname{fact}{Fact}{Facts}

\usepackage{csquotes}
\usepackage[backend=biber, style=alphabetic, maxbibnames=20, minalphanames=3, maxalphanames=4]{biblatex}
\declaretheorem{theorem}
\declaretheorem[sibling=theorem]{lemma}
\declaretheorem[sibling=theorem]{claim}
\declaretheorem[sibling=theorem]{corollary}
\declaretheorem[sibling=theorem]{observation}
\declaretheorem[sibling=theorem]{fact}

\declaretheorem[sibling=theorem]{remark}
\declaretheorem[sibling=theorem]{definition}
\declaretheorem[sibling=theorem]{property}

\theoremstyle{definition}

\DeclareMathOperator*{\argmin}{arg\,min}

\DeclareMathOperator*{\arginf}{arg\,inf}
\DeclareMathOperator*{\nn}{\boldsymbol{nn}}
\DeclareMathOperator*{\nna}{\boldsymbol{nn}^{\alpha}}
\DeclareMathOperator*{\rel}{rel^{\boldsymbol{nn}}}
\DeclareMathOperator*{\rela}{rel^{\boldsymbol{nn},\alpha}}
\DeclareMathOperator*{\simrel}{\sim_{\rm rel}}

\renewcommand{\vec}{\boldsymbol}

\usepackage{figchild}

\newclass{\CLS}{CLS\xspace}
\newclass{\UEOPL}{UEOPL\xspace}
\newclass{\EOPL}{EOPL\xspace}
\newclass{\pUEOPL}{PromiseUEOPL\xspace}
\newclass{\pEOPL}{PromiseEOPL\xspace}

\renewcommand{\epsilon}{\varepsilon}

\newcommand{\Sperner}{{\sc Sperner}\xspace}
\newcommand{\Brouwer}{{\sc Brouwer}\xspace}
\newcommand{\EoL}{{\sc End-of-Line}\xspace}
\newcommand{\eps}{\epsilon}

\newcommand*{\defeq}{:=}
\newcommand{\ind}{\mathds{1}}

\newcommand{\normtxt}[2]{\| {#1} \|_{#2}}
\newcommand{\norm}[2]{\left \| {#1} \right\|_{#2}}
\newcommand{\floortxt}[1]{\lfloor {#1} \rfloor}
\newcommand{\floor}[1]{\left \lfloor {#1} \right \rfloor}
\newcommand{\settxt}[1]{\{{#1} \}}
\newcommand{\set}[1]{\left \{{#1} \right \}}
\newcommand{\abstxt}[1]{|{#1} |}
\newcommand{\abs}[1]{\left |{#1} \right |}
\newcommand{\aset}[1]{\{ {#1} \} }

\newcommand{\idx}{\mathsf{index}}
\newcommand{\rect}{{\sf rect}\xspace}
\newcommand{\Csym}[1]{C_{\rm sym}^{(#1)}\xspace}

\newcommand{\eqClass}[1]{[#1]}

\usepackage{url}

\newcommand{\mccc}[0]{($3$-out-of-$k$, $\varepsilon$)-Mostly Approximately Envy-Free Cake Cut}

\newcommand{\mccs}[0]{($3$-out-of-$k$, $\varepsilon$)-mostly approximately envy-free cake cut}

\title{Hardness of Approximate Sperner and Applications to Envy-Free Cake Cutting}
\author[1]{Ruiquan Gao}
\author[1]{Mohammad Roghani}
\author[1]{Aviad Rubinstein}
\author[1]{Amin Saberi}
\affil[1]{Stanford University, \url{{ruiquan, roghani, aviad, saberi}@stanford.edu}}
\date{}

\begin{document}

 \begin{titlepage}
    \maketitle
     \thispagestyle{empty}
     \begin{abstract}
Given a so called ``Sperner coloring'' of a triangulation of the $D$-dimensional simplex, Sperner's lemma guarantees the existence of a {\em rainbow simplex}, i.e.~a simplex colored by all $D+1$ colors. 
However, finding a rainbow simplex was the first problem to be proven \PPAD-complete in Papadimitriou's classical paper introducing the class \PPAD~\cite{DBLP:journals/jcss/Papadimitriou94}.
In this paper, we prove that the problem does not become easier if we relax ``all $D+1$ colors'' to allow some fraction of missing colors: in fact, for any constant $D$, finding even a simplex with just three colors remains \PPAD-complete!

Our result has an interesting application for the envy-free cake cutting from fair division. 
It is known that if agents value pieces of cake using general continuous functions satisfying a simple boundary condition (``a non-empty piece is better than an empty piece of cake''),  there exists an envy-free allocation with connected pieces. We show that for any constant number of agents it is \PPAD-complete to find an allocation --even using any constant number of possibly disconnected pieces-- that makes just three agents envy-free. 

Our results extend to super-constant dimension, number of agents, and number of pieces, as long as they are asymptotically bounded by any $\log^{1-\Omega(1)}(\eps)$, where $\eps$ is the precision parameter (side length for Sperner and approximate envy-free for cake cutting).

     \end{abstract}
 \end{titlepage}

 \newpage

\section{Introduction}

Sperner's lemma is a fundamental result in combinatorics with important applications including topology (Brouwer's fixed point theorem), game theory (Nash equilibrium), economics (market equilibrium), and fair division (envy-free cake cutting). The lemma is stated in terms of a {\em large} $D$-dimensional simplex, a partition of the simplex into disjoint {\em small} $D$-dimensional simplices, and a $D+1$-{\em Sperner coloring} of the vertices of the small simplices which has to satisfy appropriate boundary conditions%
\footnote{Namely that each vertex of the large simplex has a different color, and any small-simplex-vertex on a large-simplex-facet shares a color with one of the large-simplex-vertices of that facet.}
. The lemma says that there always exists a small {\em rainbow simplex}, i.e.~a small simplex whose vertices are colored by all $D+1$ colors. 

Sperner's existence result does not come with an efficient algorithm. Indeed, when the coloring is given by a circuit, finding a rainbow simplex is \PPAD-complete, even for the case of $D+1=3$ colors~\cite{DBLP:journals/tcs/ChenD09} (for $D+1=2$ colors, a bichromatic simplex (edge) can be found easily by a binary search). The \PPAD-completeness of Sperner's lemma has been used to show analogous hardness results for all its aforementioned applications (Brouwer~\cite{DBLP:journals/jcss/Papadimitriou94}, Nash~\cite{DBLP:journals/siamcomp/DaskalakisGP09}, market equilibrium~\cite{DBLP:journals/jacm/VaziraniY11,DBLP:journals/jacm/ChenPY17}, cake-cutting~\cite{DBLP:journals/ior/DengQS12,DBLP:conf/focs/HollenderR23}, and more~\cite{DBLP:conf/focs/KintaliPRST09,DBLP:journals/teco/OthmanPR16,DBLP:journals/jcss/GoldbergH21}). 

Given the hardness of finding a small rainbow simplex, in this work we ask whether we the problem becomes easier when relaxing the problem to only requiring the small simplex to have most of the colors. Our main technical contribution is a strong impossibility result for this relaxation of Sperner's lemma:
\begin{theorem}[informal version of \cref{thm:main}]
For any $D>0$, given circuit (resp.~oracle) access to a $D$-dimensional Sperner coloring with $\eps^D$-side-length small simplices, finding even a trichromatic small simplex is \PPAD-complete (resp.~requires $\poly(1/\eps)$ oracle queries).
\end{theorem}
Note that this result is tight in the sense that as mentioned earlier, finding a bichromatic simplex is easy (using binary search). It is also interesting that even though the number of colors in the existence result grows linearly with $D$, the number of colors we can find algorithmically does not grow at all.

\subsubsection*{Cake cutting}

It is interesting to understand what consequences we can get from \cref{thm:main} for \PPAD-complete applications of Sperner's lemma. In this work, we focus on the envy-free cake cutting problem from the study of fair division: how can we fairly partition and allocate a heterogeneous ``cake'' over the unit-interval between $n$ agents?

Briefly (details deferred to \Cref{sec:cake-cutting}), in the cake cutting  model each agent has a value function mapping a partition%
\footnote{We use the following terminology: (i) a configuration of {\em cuts}, is just a finite subset of $[0,1]$; (ii) the {\em partition} of the cake into disjoint, connected pieces as determined the cuts; (iii) the {\em bundling} of the pieces into disjoint subsets; and finally (iv) the {\em allocation} which determines which agent receives which subset.} of the cake to a value in $[0,1]$ for each of the pieces. An allocation of cake is {\em envy-free} if no agent envies another agent's allocation. A classical result (e.g.~\cite{Su99-rental-harmony} using Sperner's lemma) shows that an envy-free allocation always exists even if we insist on giving each agent a single contiguous (or {\em connected}) piece of cake.
This result holds for a general class of value functions, that are only required to (i) be continuous in the location of the cuts and (ii) assign value $0$ to a piece if and only if it has length $0$. In particular, the value for a piece of cake may be non-monotone (i.e.~less cake is sometimes better), and more generally depend on the entire partition, not just the cuts at its endpoints. 

While an envy-free allocation with connected pieces always exists, it is not clear how to algorithmically find it. In order to cast this problem in a standard computational model, \cite{DBLP:journals/ior/DengQS12} propose studying the complexity of {\em $\varepsilon$-EF} cake cutting, i.e.~we want an allocation where no agent envies another agent's allocation by more than an additive $\varepsilon$. (To make the additive guarantee meaningful we normalize the values to $[0,1]$ and insist that the value functions are Lipschitz.)
They prove that $\varepsilon$-EF cake cutting is indeed \PPAD-complete, like the problem associated with Sperner's lemma.

In this work we use our main technical theorem to show that  $\varepsilon$-EF cake cutting continues to be \PPAD-complete even under two natural relaxations of the problem:
\begin{description}
    \item[Can almost everyone be almost envy-free?] Given that it is intractable to find an allocation where no agent has $\varepsilon$-envy, it would be desirable to find an allocation where at least most agents are envy-free. Unfortunately, we show that for any constant number of agents, it is \PPAD-complete to find an allocation where even $3$ agents do not envy any other agents. 
    \item[Beyond connected pieces\footnotemark]\footnotetext{We thank Alexandros Hollender for suggesting to us the connection between our notion of approximate Sperner's Lemma and cake cutting with disconnected pieces.} In most applications it is desirable to have as few cuts as possible (e.g.~when the cake models a resource like a computing cluster or vacation home shared across time), but connected pieces may not be a strict requirement. Indeed, in the (important) special case of additive value functions, finding efficient algorithms with {\em general}, i.e.~possibly disconnected, pieces has been a famous open problem for several decades. A celebrated breakthrough of~\cite{DBLP:journals/cacm/AzizM20} shows that for any constant number of agents, there is an efficient algorithm that finds an envy-free allocation with a constant number of cuts --- this is in contrast to connected pieces where a $\polylog(1/\varepsilon)$-time algorithm was only recently given for $n=4$ agents~\cite{DBLP:conf/focs/HollenderR23}, and $n\ge 5$ agents remain an open problem. Here we show that for more general value functions, the problem becomes \PPAD-complete, even with $3$ agents and any constant number of cuts. 
\end{description}

\begin{theorem}[Informal version of \cref{thm:cake-main}]
For any number $k \ge 3$ of agents and cuts, it is \PPAD-hard to find an $\varepsilon^k$-EF allocation (of possibly disconnected pieces) that makes even $3$ agents envy-free.    
\end{theorem}

Our \PPAD-completeness result imply $\poly(1/\eps)$ lower bounds on the number of value queries (respectively, color queries for Sperner's lemma). We remark also that our lower bounds for cake cutting hold for the special case where all agents have the same value function.

\subsection{Closely related work}

Our result on approximate Sperner's lemma is most closely related to works on approximate Brouwer fixed point~\cite{DBLP:journals/jc/HirschPV89, DBLP:journals/siamcomp/Rubinstein18, DBLP:conf/innovations/BabichenkoPR16, DBLP:conf/focs/Rubinstein16, DBLP:journals/siamcomp/FilosRatsikasHSZ23}. In particular, \cite{DBLP:conf/focs/Rubinstein16} proves that in the regime of asymptotically large dimension, it is \PPAD-complete to find a small simplex with $1-\delta$ fraction of the colors (for a small, unspecified constant%
\footnote{An earlier result of~\cite{DBLP:journals/siamcomp/Rubinstein18} proved a weaker result where only the side length is a small unspecified constant; the constants in that paper were recently dramatically improved in~\cite{DBLP:journals/siamcomp/FilosRatsikasHSZ23}. However, neither of those works has any non-trivial guarantee on the fraction of colors in the small simplex; this fraction is the focus of our paper.}
$\delta > 0$) even when the small simplex has constant side length%
\footnote{The domain in~\cite{DBLP:conf/focs/Rubinstein16} and other related works on approximate Brouwer fixed point is the hypercube rather than the simplex. While in constant dimension this makes little difference, in $n$-dimensions the simplex is exponentially smaller. Our comparison here is informal, possibly erring on the side of giving earlier work too much credit.}
$\varepsilon >0$.  
This result is incomparable to our result which holds in constant dimension but requires the small simplex to be asymptotically small. In terms of the fraction of colors, we show that $3$-vs-$\omega(1)$ is already hard, compared to $(1-\delta)n$-vs-$n$ in~\cite{DBLP:conf/focs/Rubinstein16}. Our recursive construction is technically completely different from the error correcting code technique of~\cite{DBLP:conf/focs/Rubinstein16}; it is an interesting direction for future work to explore whether our result on approximate Sperner coloring can have interesting implications for approximate equilibria in games (main goal of~\cite{DBLP:conf/focs/Rubinstein16}), equilibria in other domains (e.g.~markets),  or even the PCP Conjecture for \PPAD~\cite{DBLP:conf/innovations/BabichenkoPR16}.

For cake-cutting with additive valuations, the problem of envy-free cake cutting has been studied in the Robertson-Webb oracle model, where in addition to value queries, the algorithm can also ask the oracle for a piece that agent $i$ values at $\alpha$. 
For additive $\varepsilon$-EF allocations, \cite{DBLP:conf/nips/BranzeiN22} show that cut queries can be simulated with $\Theta(\log(1/\varepsilon))$ value queries, so the models are equivalent up to these log factors. Cut queries have a natural extension to monotone value functions (see e.g.~\cite{DBLP:conf/focs/HollenderR23}), but it is not clear how to extend them to our model of general value functions.

In the Robertson-Webb oracle model, one can ask for exact EF allocations.
For connected pieces, even in the Robertson-Webb oracle model no algorithm can find an EF allocation in finite time~\cite{DBLP:journals/combinatorics/Stromquist08} (demonstrating a strong separation between the connected and general pieces). 
For general pieces, the celebrated breakthrough of~\cite{DBLP:journals/cacm/AzizM20} gives an efficient algorithm that runs in constant time (and hence a constant number of cuts) for a constant number of agents. However, the dependence on the number of agents $n$ is poor: a six-power tower; in contrast, the best lower bound is only $\Omega(n^2)$~\cite{DBLP:conf/ijcai/Procaccia09}, and closing the gap is a famous open problem. Our result suggests that it would be impossible to find better algorithms by only exploiting the topology of the problem. 

Our hardness on making most agents envy-free is related to a result of~\cite{DBLP:conf/innovations/BabichenkoPR16} who showed, assuming the PCP Conjecture for \PPAD, a similar-flavor notion of hardness for approximately fair allocation of courses to students in the A-CEEI framework. (The techniques are completely different, and they focus on the regime of asymptotically large number of agents/students.)

\cite{DBLP:conf/focs/HollenderR23} study the complexity of $\varepsilon$-EF cake cutting with $4$ agents. They show that an efficient protocol with monotone value functions, and \PPAD, query, and communication complexity hardness for non-monotone functions. Their notion of non-monotone functions is less general than our notion of value functions in the sense that in~\cite{DBLP:conf/focs/HollenderR23} the value of a piece only depends on the locations of the cuts that define it. It is an interesting open problem whether our results can be extended to this notion of value functions, and/or to communication complexity (see also~\cite{DBLP:conf/ec/BranzeiN19} on communication complexity in cake cutting).

Many other problems in fair division are known to be complete for \PPAD~(e.g.~\cite{DBLP:journals/teco/OthmanPR16,DBLP:journals/corr/abs-2008-00285}), or related classes such as \PPA~\cite{DBLP:conf/stoc/Filos-RatsikasG18, DBLP:conf/stoc/DeligkasFHM22, DBLP:journals/ai/DeligkasFH22}, \PLS~\cite{DBLP:conf/sagt/GoldbergHH23}, \PPA-$k$~\cite{DBLP:conf/soda/Filos-RatsikasH21}, and likely also  for $\PPAD \cap \PLS$~\cite{DBLP:conf/soda/Arunachaleswaran19}; see also the recent survey of~\cite{DBLP:journals/ai/AmanatidisABFLMVW23}. In particular, \cite{DBLP:conf/stoc/DeligkasFHM22} extend a \PPA-completeness result of fair division (``consensus halving'') with $n$ cuts to allow a small fraction of redundant cuts, namely $n+n^{1-\delta}$; in comparison, we show hardness of $3$-vs-$\omega(1)$ cuts. 

\section{Overview of techniques}
To prove hardness of approximate Sperner (\cref{thm:main}) we use a combination of a recursive construction and the hard $2$-dimensional instance of~\cite{DBLP:journals/tcs/ChenD09}.
At the base of our recursion we consider a $D=1$-dimensional instance, which is simply a line with two colors.  We assume wlog that the line is colored with blue for all the point left of the midpoint, and red on the right. 

The recursive step builds on a construction of~\cite{DBLP:journals/tcs/ChenD09} who showed that given a Sperner coloring in $D=2$ dimensions, it is \PPAD-complete to find a small trichromatic simplex. We can think of their construction as extending a $1$-dimensional base case to a $2$-dimensional instance. We use their construction almost black box, except again for the assumption (which can be easily seen to hold in their construction) that the base case is colored with blue up to the midpoint, and red thereafter. 

\subsubsection*{The key idea}

At a high level, for going from a generic $D$ to $D+1$ dimensions, we consider all the points where two colors meet. Locally, consider the line segment perpendicular to the manifold between the two colors, and observe that it looks just like a shrunk version of our $1$-dimensional base case (one color up to the midpoint, then a different color). So we can extend it in the $(D+1)$-th dimension from a line segment to a $2$-dimensional simplex by pasting a shrunk copy of \cite{DBLP:journals/tcs/ChenD09}'s construction. By default, all other points with a non-trivial component in the $(D+1)$-th dimension are colored with the $(D+1)$-th color. 

Suppose by induction that it is \PPAD-hard to find a point simultaneously close to $3$ colors in the $D$-dimensional instance. The base case is trivial as clearly with $D+1 = 2$ colors there are no trichromatic regions at all. For the induction step, we need to argue that it is hard to find a point in the $(D+1)$-dimensional instance that is simultaneously close to the new $(D+1)$-th color and two of the other $D$ colors. However, any point close to two of the $D$ colors is covered by a shrunk copy of the $3$-color, $2$-dimensional \cite{DBLP:journals/tcs/ChenD09}-instance, where it is again \PPAD-complete to find a trichromatic point!

\subsubsection*{Challenge: maintaining consistency with empty pieces}
The most interesting challenge in implementing the above proof plan has to do with the inherent symmetries of the cake cutting applications. Consider for example the case where we only try to make some agents envy-free (a similar challenge arises for redundant cuts): in this case, it is acceptable to completely sacrifice the happiness of one agent and give them an empty piece. Now there are $D-1$ configurations of $D$ cuts that correspond to the same allocation of $D+1$ cake pieces with one of them being empty. 
For the cake cutting application, we would like to ensure that the agents' values for all those allocations to are identical.

Geometrically, ``$i$-th piece is empty'' corresponds to a tight constraint on the location of the cuts, or equivalently a facet of the simplex. So in order to correspond to a feasible value function, the Sperner coloring has to satisfy an unusual symmetry constraint: 
\begin{quote}
    For every $d<D$, the coloring on every $d$-dimensional the facets is identical, up to a permutation of the colors. 
\end{quote}

Notice that this issue comes up specifically because of the relaxed notion of envy-free with respect to some agents that we study in this paper: in contrast, for every-agent-envy-free cake cutting with connected pieces, we can assume wlog that every piece is nonempty (equivalently, there is no rainbow simplex on the boundary) --- for any  agent that receives it will envy other agents.

Recall that the $D-1$ dimensional construction is constructed from shrunk copies of the $2$-dimensional instance. When we recurse to add more dimensions, the copies have to shrink exponentially, resulting in different {\em shrinking factors}. Thus our Sperner  construction has different shrinking factors for different facets --- failing to satisfy the symmetry condition. 

To overcome this, we modify our $D$-dimensional construction to first have a copy of the $D-1$-dimensional construction on each facet. Then in the regions near the boundary we interpolate between this original copy and the copy with smaller shrinking factors. 

\section{Preliminaries}
\label{sec:prelim}

We use $[n]$ to denote the set $\{1,\dots, n\}$.
We use $Q_n$ to denote the integer set $\{0,\dots, 2^n-1\}$ that can be encoded by $n$ bits. 
We use $(x)_+$ to denote $\max\{0,x\}$.
We use $(x)_-$ to denote $\min\{1,x\}$.
Further, we use $(x)_{[0,1]}$ to denote $((x)_+)_-$.
For any vector $\vec{x}$ and any $l,r\in \mathbb{Z}$, we use $\vec{x}_{l:r}$ to denote $\vec{x}$ restricted to indices in $[l,r]$, i.e., $(x_l,\dots,x_r)$.
For any vector $\vec{x}$ and any $i\in \mathbb{Z}$, we use $\vec{x}_{-i}$ to denote $\vec{x}$ restricted to indices not equalling $i$, i.e., $(\vec{x}_{1:i-1}, \vec{x}_{i+1:r})$.

\paragraph{The Complexity Class: \TFNP.}
Search problems are defined via relations $R\subseteq \{0,1\}^*\times \{0,1\}^*$, where the goal is to find a string $\vec{y}$ for the input $\vec{x}$ such that $(\vec{x},\vec{y})\in R$, or to claim no such $\vec{y}$ exists. 
The class {\FNP\xspace} consists of all such search problems in which $R$ is polynomial-time computable, which means that whether any $(\vec{x},\vec{y})\in R$ can be decided in polynomial time, and polynomially balanced, which means that there exists a polynomial $p(n)$ such that $(\vec{x},\vec{y})\in R$ only if $|\vec{y}|\leq p(|\vec{x}|)$. 
Here, $|\vec{x}|$ is defined as the number of bits in string $\vec{x}$.
The class {\TFNP\xspace} consists of all {\FNP\xspace} problems whose relation $R$ is total, i.e., for any input $\vec{x}$, there always exists a string $\vec{y}$ such that $(\vec{x}, \vec{y})\in R$.

The polynomial-time reduction between two \TFNP~problems $R,R'$ is defined by two polynomial-time computable functions $f:\{0,1\}^*\to \{0,1\}^*$ and $g:\{0,1\}^*\times \{0,1\}^* \to \{0,1\}^*$ such that for any instance $\vec{x}$ of $R$, we can always reduce it to the instance $f(\vec{x})$ (of $R'$) and then recover a solution from any output $\vec{y}$ of the instance of $R'$ by $g(\vec{x}, \vec{y})$.
That is, for any input $\vec{x}$ of $R$ and any output $ \vec{y}$ of $R'$,
\begin{align*}
    (f(\vec{x}),\vec{y}) \in R' \Longrightarrow (\vec{x},g(\vec{x},\vec{y})) \in R~.
\end{align*}

\paragraph{The Complexity Class: \PPAD.}
The class 
\PPAD~consists of all \TFNP~problems that are polynomial-time reducible to the following \EoL problem.
\begin{definition}[\EoL]
    In \EoL, we are given a predecessor function $P:Q_n\to Q_n$ and a successor function $S:Q_n\to Q_n$ such that $P(0^n)=0^n\neq S(0^n)$, where both $P$ and $S$ are given by a polynomial-size circuit. The goal of \EoL is to find $x\in Q_n$ such that we have either $P(S(\vec{x}))\neq \vec{x}$ or $S(P(\vec{x}))\neq \vec{x} \neq 0^n$.
\end{definition}

This class, \PPAD, captures the complexity of the following $k$D-\Sperner problem, whose totality is guranteed by the famous Sperner's lemma \cite{sperner1928neuer}.

\begin{definition}
    \label{def:standard-simplex}
    The standard $k$-simplex is $\Delta^k = \{\vec{x}\in [0,1]^{k+1}: \sum_{i=1}^{k+1} x_i=1\}$.
\end{definition}

\begin{definition}
    \label{def:discrete-simplex}
    The discrete $k$-simplex with triangulation side-length of $2^{-n}$ is defined as the following subset of the standard $k$-simplex $\Delta^k$: $\Delta^k_n = \{\vec{x}\in \Delta^k: \forall i\in [k+1], ~ (2^n-1)\cdot x_i\in \mathbb{Z} \}$.
\end{definition}

\begin{definition}
    The array of non-trivial indices $\mathcal{I}_{>0}(\vec{x})$ of $\vec{x}$ is given by indices $i$ such that $x_i>0$ in an increasing order.
\end{definition}

\begin{definition}[$k${\rm D}-\Sperner]
    In $k${\rm D}-\Sperner, we are given function 
    $C:\Delta^k_n \to [k+1]$ satisfying $C(\vec{x}) \in \mathcal{I}_{>0}(\vec{x})$ for any $\vec{x}\in \Delta^k_n$.
    The goal of $k${\rm D}-\Sperner is to find $k+1$ points $\vec{x^{(1)}},\vec{x^{(2)}}, \dots, \vec{x^{(k+1)}}\in \Delta_n^k$ such that
    \begin{enumerate}[itemsep=0.2em,topsep=0.5em]
        \item {\bf the three points are close to each other}, i.e., $\forall i,j\in [k+1], \normtxt{\vec{x^{(i)}}-\vec{x^{(j)}}}{\infty}\leq 2^{-n}$; and 
        \item {\bf the induced simplex is rainbow}, i.e., $|\{C(\vec{x^{(i)}}): i\in [k+1]\}|=k+1$.
    \end{enumerate}
\end{definition}

\begin{theorem}[\cite{DBLP:journals/jc/HirschPV89,DBLP:journals/jcss/Papadimitriou94,DBLP:journals/tcs/ChenD09}]
    For any $k\geq 2$, $k${\rm D}-\Sperner is \PPAD-complete (if the coloring $C$ is given by a polynomial-size circuit) and requires a query complexity of $2^{\Omega(n)}$ (if the coloring $C$ is given by a black box).
\end{theorem}

\section{Approximate High-dimensional Sperner: Formulation and Statement of Result}
\label{sec:3color}
\newcommand{\outofk}{{$3$-out-of-$k\!+\!1$}\xspace}

In this section, we formally define the relaxed computational problem of Sperner's lemma, and its variants that we will use for the applications on envy-free cake cutting.  
In particular, we state our main technical result \cref{thm:main}, deferring the proof to \cref{sec:warmup-proof-for-approx-sperner,,sec:final-proof-for-approx-sperner}.

Here, we use the term {\em approximate Sperner's lemma} for the following relaxed version of the original Sperner's lemma~\cite{sperner1928neuer}: in any triangulation of a large $k$-dimensional simplex with a $k+1$-Sperner coloring that satisfies appropriate boundary condition, there always exists a small triangle (a.k.a., $2$-simplex) whose three vertices are colored differently (because Sperner's lemma states that there exists a small $k$-simplex whose $k+1$ vertices are colored by all $k+1$ colors).
Formally, we define the following \outofk Approximate \Sperner problem.

\begin{definition}[\outofk Approximate \Sperner]\label{def:approx-kdsperner}
    In \outofk Approximate \Sperner (with a triangulation side-length of $2^{-n}$), we are given a function 
    $C:\Delta^k_n \to [k+1]$ satisfying $C(\vec{x}) \in \mathcal{I}_{>0}(\vec{x})$ for any $\vec{x}\in \Delta^k_n$.
    The goal of \outofk Approximate \Sperner is to find three points $\vec{x_1},\vec{x_2}, \vec{x_3}\in \Delta_n^k$ such that
    \begin{enumerate}[itemsep=0.2em,topsep=0.5em]
        \item {\bf the three points are close to each other}, i.e., $\forall i,j\in [3], \normtxt{\vec{x_i}-\vec{x_j}}{\infty}\leq 2^{-n}$; and 
        \item {\bf the induced triangle is trichromatic}, i.e., $|\{C(\vec{x_i}): i\in [3]\}|=3$.
    \end{enumerate}
\end{definition} 

\begin{remark}
    \label{remark:equiv-at-2}
    When $k=2$, the \outofk Approximate \Sperner problem is simply equivalent to the $k${\rm D}-\Sperner problem.
\end{remark}
    
To apply the hardness result of \outofk Approximate \Sperner to the envy-free cake-cutting problem, we further consider the special case 
where the coloring defined by $C$ satisfies the following symmetry constraints on the boundary of the simplex. 
First, we define the symmetry that naturally arises from cake-cutting valuations.
For example, suppose we cut the cakes using two cuts in two different ways: with cut locations $(0.1, 0.1)$ and $(0.1, 1)$. Both partitions of the cake have the same two non-trivial pieces: $[0,0.1]$ and $[0.1,1]$. Thus we expect that the preference for each agent between those pieces is the same regardless of the representation on the simplex. 
This inspires us to define a symmetry constraint on cuts that have the same set of non-zero pieces.

\begin{definition}
    \label{def:indexing}
    Indexing $\idx(\vec{a}, v)$ of $v$ in an array $\vec{a}$ is defined as the index $i$ such that $a_i=v$. 
\end{definition}

\begin{restatable}[symmetric points in a coloring]{definition}{symmetrydef}
    \label{def:symmetry-in-coloring}
    We say $\vec{x}$ and $\vec{y}$ are symmetric in a coloring $C$ if 
    \begin{itemize}[topsep=0.5em, itemsep=0.1em]
        \item the number of non-trivial entries in $\vec{x}$ and $\vec{y}$ are the same, i.e., $|\mathcal{I}_{>0}(\vec{x})| = |\mathcal{I}_{>0}(\vec{y})|$.
        \item the indexing of $C$ are the same for the arrays of non-trivial indices of $\vec{x}$ and $\vec{y}$, i.e., $\idx(\mathcal{I}_{>0}(\vec{x}), C(\vec{x})) = \idx(\mathcal{I}_{>0}(\vec{y}), C(\vec{y}))$.
    \end{itemize}
    We use $\vec{x}\sim_C \vec{y}$ to denote this symmetry. 
\end{restatable}

Because the two conditions for symmetry are both equations, which satisfy reflexivity, symmetry, and transitivity, this symmetry is a valid equivalence relation. 
\begin{fact}
    \label{lem:symmetry-is-equiv-relation}
    For any coloring $C$, the symmetry of points in coloring $C$ is an equivalence relation, i.e., a binary relation satisfying reflexive ($\vec{x}\sim_{C} \vec{x}$), symmetric ($\vec{x}\sim_{C} \vec{y} \rightarrow \vec{y}\sim_{C} \vec{x}$) and transitive ($\vec{x}\sim_{C} \vec{y}$ and $\vec{y}\sim_{C} \vec{z} \rightarrow \vec{x}\sim_{C} \vec{z}$).
\end{fact}

In the symmetric version of the \outofk Approximate \Sperner problem, the input guarantees that symmetric points on the boundary are symmetric in the input coloring. 
We say two points on the boundary are symmetric to each other, if the resulting vectors are the same after we remove the zero entries in them.
\begin{restatable}[\outofk Approximate Symmetric \Sperner]{definition}{symsperner}
    \label{def:approx-symmetric-kd-sperner}
    \outofk Approximate Symmetric \Sperner is a special case of \outofk Approximate \Sperner, where the given circuit $C$ further satisfies the following property: for any $\vec{x}\in \Delta_n^{k}$ and any $i,j \in [k+1]$, $(\vec{x}_{1:i-1}, 0, \vec{x}_{i:k}) \sim_C (\vec{x}_{1:j-1},0, \vec{x}_{j:k})$. 
\end{restatable}

Our main technical result is the \PPAD-hardness of this \outofk Approximate Symmetric \Sperner problem. 
Because this problem is clearly a member of the class \PPAD, its complexity is \PPAD-complete. 

\begin{restatable}[]{theorem}{mainthm}
    \label{thm:main}
    For and any $k=\poly(n)$, \outofk Approximate Symmetric \Sperner
    with a triangulation side-length of $2^{-4kn}$ 
    is \PPAD-complete.
    Further, for any $k\geq 2$, it requires a query complexity of $2^{\Omega(n)} / \poly(n,k)$.
\end{restatable}

\section{Cake Cutting: Formal Model and Statement of Result}

\subsubsection*{Cuts, pieces and allocations}
\paragraph{$k$-cuts:} In line with the classic model of cake cutting, we use a line segment $[0, 1]$ to represent the ``cake''. We let $k$ be the number of the cuts. A {\em $k$-cut} can be represented as a vector $\vec{x} = (x_1, x_2, \ldots, x_k)$ with $k$ dimensions such that $0 \leq x_1 \leq x_2 \leq \ldots \leq x_k \leq 1$. Dimension $i$ shows where the $i$-th cut is located on $[0, 1]$. 

We only define preferences with respect to partitions that use exactly $k$ cuts, for some parameter $k$ (note that we allow $k$ to be larger than the number of agents). Partitions corresponding to fewer than $k$ cuts can easily%
\footnote{This is easy to {\em model}. But for the {\em proof}, the boundary conditions imposed by those empty pieces are actually one of our main sources of difficulty!} 
be incorporated in this model by adding trivial cuts (equivalently, empty pieces).  
We do not allow more than $k$ cuts in our model. Note that without an upper bound on the number of cuts, one could partition the cake into infinitesimal pieces and assign them to agents at random; by continuity of the valuations this would trivially be approximately envy-free, but eating infinitely many disconnected infinitesimal pieces of cake can be unsatisfactory in applications.

\paragraph{Pieces of cake:} A {\em piece of cake} is an open interval that is a subset of $[0, 1]$. So a $k$-cut partitions the cake into $k+1$ (possibly empty) pieces%
\footnote{... and their $k+2$ endpoints (since we treat pieces as open intervals).}. For a $k$-cut $\vec{x}$, we use $X = (X_0, X_1, \ldots, X_k)$ (capitalizing the character representing the $k$-cut) to denote the corresponding pieces of the cake where $X_0$ and $X_k$ are the leftmost and the rightmost pieces of the cake, respectively. 

\paragraph{Equivalent $k$-cuts:} We say that two $k$-cuts $\vec{x}$ and $\vec{y}$ are {\em equivalent} if they induce the same pieces of cake. (For example, $\vec{x} = (1/3,1/3)$ and $\vec{y} = (0,1/3)$ both induce the pieces $(0,1/3), (1/3,1), \emptyset$.) We use $\eqClass{\vec{x}}$ to denote the equivalence class of $\vec{x}$.

\paragraph{Allocation:}
An {\em allocation} is an assignment of pieces to a set of agents $\{1, 2, \ldots, p\}$ (possibly multiple pieces to the same agent). 

\subsubsection*{Utilities and preferences}

\paragraph{Utility functions:}
Each agent $d$ has a {\em utility function} $u_d$ that takes a piece of cake and the entire partition of the cake into an (unordered) set of pieces, and maps them to a real value in $[0,1]$. Equivalently, and consistent with the notation we will use, $u_d$ maps a piece $X_i$ and an equivalence class of $k$-cuts $\eqClass{\vec{x}}$ to $u_d(\eqClass{\vec{x}}, X_i) \in [0,1]$. We assume that our utility functions satisfy two conditions commonly found in cake-cutting literature. The first condition is Lipschitz continuity which lets us to discretize the problem. The second condition is Nonnegativity condition which is the requirement for the existence of a valid solution for cake-cutting. We formally define these two conditions as follows:

\begin{itemize}
    \item \textbf{Lipschitz condition:} Let $\vec{x}$ and $\vec{y}$ be two $k$-cuts. For any player $d$, we have $|u_d(\eqClass{\vec{x}}, X_i) - u_d(\eqClass{\vec{y}}, Y_i)| \leq L \cdot \norm{\vec{x} - \vec{y}}{1}$ for all $i \in \{0, \ldots, k\}$, where $L$ is the Lipschitz constant.

    \item \textbf{Nonnegativity condition:} Let $k$-cuts $\vec{x}$ be a $k$-cut. For any player $d$, we have $u_d(\eqClass{\vec{x}}, X_i) > 0$ if $X_i \neq \emptyset$, and $u_d(\eqClass{\vec{x}}, X_i) = 0$ otherwise.
\end{itemize}

\paragraph{Bundles of cake pieces:}
Given a partition of the cake into (connected) pieces, a {\em bundling} $\mathcal{B} = \{B_1, B_2, \ldots, B_{p}\}$ partitions the set of cake pieces into {\em bundles}, or subsets of pieces. The utility of a bundle of pieces for one agent is the sum of the utility of the pieces belonging to the bundle. Formally, if $B = \{X_{i_1}, X_{i_2}, \ldots, X_{i_r} \}$ is a bundle, we have $u_d(\eqClass{\vec{x}}, B) = \sum_{j=1}^r u_d(\eqClass{\vec{x}}, X_{i_j})$.

\paragraph{Preferences:}
Given a bundling $\mathcal{B} = \{B_1, B_2, \ldots, B_{p}\}$, we say agent $d$ (Weakly) {\em prefers} bundle $B_i$ if $u_d(\eqClass{\vec{x}}, B_i) \geq u_d(\eqClass{\vec{x}}, B_{j})$ for all $j$.

\subsection{Envy-free cake cuts}

\begin{definition}[$\eps$-Approximate Envy-Free Cake Cut]
    Given a $k$-cut $\vec{x}$, a bundling $\mathcal{B}$ of the induced cake pieces, an assignment $\pi: [p] \rightarrow [p]$ of the bundles to agents, and $\eps>0$ we say that  $(\vec{x}, \mathcal{B}, \pi)$ is  {\em $\eps$-approximate envy-free} if $u_d(\eqClass{\vec{x}}, B_{\pi(d)}) + \eps \geq u_d(\eqClass{\vec{x}}, B_i)$  for every agent $d$ and any bundle $B_i$. When $\eps = 0$, we simply say that $(\vec{x}, \mathcal{B}, \pi)$ is  envy-free.
\end{definition}

Put in this language, prior work showed that with connected pieces, an envy-free cake cut always exists, but finding one is \PPAD-complete:

\begin{theorem}[\textcite{stromquist1980cut,Su99-rental-harmony}]\label{thm:cake-cutting-exists}
    There exists an envy-free cake cut for $k + 1$ agents with $k$ cuts.
\end{theorem}

\begin{theorem}[\textcite{DBLP:journals/ior/DengQS12}]\label{thm:previous-result}
    Finding $\eps$-approximate envy-free cake cut for $k + 1$ agents using $k$ cuts is \PPAD-complete.
\end{theorem}

We consider two different settings that are less restrictive compared to the result of \Cref{thm:previous-result} and we prove that even if we allow both relaxations, the problem is still \PPAD-complete. We study the following relaxations of the cake-cutting:
\begin{itemize}

    \item \textbf{Making almost every agent almost envy-free:} similar to the setting in \cite{DBLP:journals/ior/DengQS12}, consider $k$-cut and $k + 1$ agents, with the objective of allocating each agent a single continuous piece obtained by $k$ cuts. Instead of seeking an envy-free $k$-cut, one might ask whether it is possible to find a $k$-cut where the majority of the agents do not envy anyone. Specifically, if $k + 1 > 3$, can we find a $k$-cut where at least three agents do not envy anyone? We provide a negative answer to this question and demonstrate that this relaxed version of the cake-cutting problem is also \PPAD-complete.

    \item \textbf{Three agents with multiple pieces:} given the current results, a natural question arises: if we allow agents to have a bundle of pieces of the cake, is the problem still hard? It is worth noting that this question is computationally easier compared to the scenario where we allocate only one piece to each agent, as we have the option to allocate empty pieces to agents. Suppose that we have three agents and $k + 1 > 3$ pieces of cake. Essentially, we cut the cake into $k + 1$ pieces and bundle these pieces into three bundles. The goal is to determine whether such a bundling exists where we can allocate to each agent a bundle of pieces such that agents do not envy each other. Similarly, we demonstrate that this problem is \PPAD-complete. 
\end{itemize}

More specifically, we combine both relaxations. We assume that we have $p$ agents and $k$ cuts such that $p \leq k + 1$. We bundle the $k + 1$ pieces to $p$ bundles and allocate them to the agents, and our objective is to find an approximate envy-free cake cut where at least three agents do not envy others.

\begin{restatable}[$\eps$-Approximate $p'$-out-of-$p$-Envy-Free Cake Cut]{definition}{poutofqcakecut}
    Given a $k$-cut $\vec{x}$, a bundling $\mathcal{B}$ of the induced cake pieces, an assignment $\pi: [p] \rightarrow [p]$ of the bundles to agents, $\eps>0$, and $p' < p$ we say that  $(\vec{x}, \mathcal{B}, \pi)$ is  $\eps$-approximate $p'$-out-of-$p$-envy-free if there exists a subset $S \subseteq [p]$ of $p'$ agents, such that for every $d \in S$, $u_d(\eqClass{\vec{x}}, B_{\pi(d)}) + \eps \geq u_d(\eqClass{\vec{x}}, B_i)$  for every bundle $B_i$. 
\end{restatable}

We can now formally state our main result for cake cutting.
\begin{restatable}{theorem}{cakecuttingtheorem}\label{thm:cake-main}
For any constants $3 \le p \le k+1$ and $\eps$ such that $k < \log^{1-\delta}(1/\eps)$ for some constant $\delta > 0$, the problem of finding an $\eps$-approximate $3$-out-of-$p$-envy-free cake cut with $k$ cuts is \PPAD-complete. If, instead, the algorithm has black-box access to a value oracle, it requires a query complexity of $(1/\eps)^{\Omega(1/k)} / \polylog(1/\eps)$.
\end{restatable}

We defer the proofs and technical details to \Cref{sec:cake-cutting}.

\section{Warmup: PPAD-Hardness of Approximate (Unconstrained) Sperner}
\label{sec:warmup-proof-for-approx-sperner}

In this section, we prove that \outofk Approximate \Sperner (without symmetry constraints) is \PPAD-complete. 
This is a warm-up proof for the \PPAD-completeness of the \outofk Approximate Symmetric \Sperner problem, whose hard instances and proof of \PPAD-completeness will be presented in the next section.
For convenience, we will not consider the query complexity in this warm-up section. 

\begin{theorem}
    For any constant $\delta>0$ and any $k=\poly(n)$, \outofk Approximate \Sperner is \PPAD-complete.
    In particular, there is a reduction in $\poly(n,k)$ time that reduces a 2{\rm D}-\Sperner instance with triangulation side-length of $2^n$ to a \outofk Approximate \Sperner instance with triangulation side-length of $2^{-3kn}$.
\end{theorem}

Because the \outofk Approximate \Sperner problem is a relaxation of $k${\rm D}-\Sperner problem, which is in \PPAD. 
The {\PPAD\xspace} membership of \outofk Approximate \Sperner problem is clear. 
Our focus of this section is to prove its \PPAD-hardness.

We will now begin describing our formal construction (still for the warm-up result). To help the reader keep track of the notation, both here and the main construction in the next section, we provide a summary in \cref{table:notation}.

\begin{table}[h] 
\caption{Table of notation for Sections~\ref{sec:warmup-proof-for-approx-sperner} and~\ref{sec:final-proof-for-approx-sperner}}
    \centering
    \begin{tabular}{c|c|c} 
        \hline
        \hline
        {\bf \large Notation} & {\bf \large Description} & {\bf \large References} \\
        \hline
        \hline
        $\vec{x}$ & vectors (also, as $\vec{x'}, \vec{x}^{(1)}, \vec{x}^{(2)}, \dots$) & all sections\\
        $x_i$ & coordinates in vectors (also, as $x'_i, x^{(1)}_i, x^{(2)}_i, \dots$) & all sections\\
        $\Delta^k$ & standard $k$-simplex & \cref{def:standard-simplex}\\
        $\Delta_n^k$ & discrete $k$-simplex with a triangulation side-length of $2^{-n}$ & \cref{def:discrete-simplex}\\
        $C_{\rect}(\cdot)$ & 2{\rm D}-\rect\Sperner instances & \cref{def:rect-sperner}\\
        $C(\cdot)$ & base 2{\rm D}-\Sperner instances & \cref{alg:base-instance}\\
        $C^{(k)}(\cdot)$ & \outofk Approximate \Sperner instances & \cref{alg:base-instance,alg:approx-sperner-kgeq3}\\
        $C^{(k)}_{\rm sym}(\cdot)$ & \outofk Approximate Symmetric \Sperner instances & \cref{alg:approx-symmetric-sperner-kgeq3} \\
        $\alpha(\cdot)$ & shrinking factor & \cref{eqn:shrinking-factor}\\
        $\rel(\cdot)$ & coordinate converter & \cref{eqn:coordinate-converter} \\
        $\rela(\cdot)$ & coordinate converter for symmetry & \cref{eqn:new-coordinate-converter}\\
        $d(\cdot,\cdot)$ & a metric to define $\rel$ & \cref{eqn:modified-distance}\\
        $d^{\alpha}(\cdot,\cdot)$ & an asymmetric quasimetric to define $\rela$ & \cref{eqn:shrinking-distance}\\
        $\nn(\cdot)$ & nearest neighbor & \cref{eqn:nn}\\
        $\nna(\cdot)$ &  nearest neighbor w.r.t. $\alpha$ & \cref{eqn:cnn-and-nn-for-symmetric}\\
        $\mathcal{N}(\cdot)$ & neighborhood of size $\eps\times \eps$ & \cref{def:neighbourhood} \\
        $C_{\sf nn}(\cdot)$ & neighboring color (i.e., the color of the nearest neighbor) & \cref{eqn:cnn}\\
        $\hat{C}_{\sf nn}(\cdot)$ & the modified neighboring color & \cref{eqn:modified-neighboring-color}\\
        $C_{\sf nn}^{\alpha}(\cdot)$ &  the neighboring color w.r.t. $\alpha$ & \cref{eqn:cnn-and-nn-for-symmetric}\\
        $\hat{C}_{\sf nn}^{\alpha}(\cdot)$ &  the modified neighboring color w.r.t. $\alpha$ & \cref{eqn:modified-neighboring-color-symmetry}\\
        $\sim_C$ & symmetry in colorings (binary relation) & \cref{def:symmetry-in-coloring} \\
        $\simrel$ & equivalence between converted coordinates (binary relation) & \cref{def:equiv-converted-coordinates}\\
        $\vec{P}(\cdot)$ & the projection step from $\Delta^{k}$ to $\Delta^{k-1}$ (for any $k$) & \cref{def:projection}\\
        $\vec{P}^{(\ell)}(\cdot)$ & performing $\ell$ projection steps & \cref{def:lth-projection}\\
        $i^*(\cdot)$ & the index of the first non-zero entry & \cref{eqn:first-nonzero-index}\\
        $\vec{y}^{(i)}(\cdot)$ & the $i$-th intermediate projection (\cref{subsec:symmetric-kd,,subsec:warmup-kd} only) & \cref{alg:approx-sperner-kgeq3,,alg:approx-symmetric-sperner-kgeq3} \\
        $\vec{c}^{(i)}(\cdot)$ & the $i$-th intermediate palette (\cref{subsec:symmetric-kd,,subsec:warmup-kd} only) & \cref{alg:approx-sperner-kgeq3,,alg:approx-symmetric-sperner-kgeq3} \\
        \hline
    \end{tabular} \label{table:notation}
\end{table}

We will construct a chain of \PPAD-hard instances $C^{(2)}, C^{(3)}, \dots$ to prove this warm-up technical result. 
We will present the first instance of the chain $C^{(2)}$ in \cref{subsec:warmup-2d}.
Later, this $C^{(2)}$ will also serve as the base \PPAD-hard 2{\rm D}-\Sperner instance $C$ (i.e., we will use the notation $C$ when referring to the base instance), and we will present in \cref{subsec:warmup-kd} a recursive way that constructs $C^{(k+1)}$ by combining $C^{(k)}$ with one copy of this base 2{\rm D}-\Sperner instance $C$. 
Towards this end, in \cref{subsec:warmup-convert} we introduce an intermediate function, which we call the {\em coordinate converter}; this function projects points in a standard $2$-simplex to real values in $[0,1]$ (the {\em converted coordinate}). When we move from $C^{(k)}$ to $C^{(k+1)}$ we have a new coordinate; we embed a copy of the base instance $C$ on the plane spanned by the new coordinate and the converted coordinate. 

Our instances will be constructed in continuous space instead of discrete space. 
That is, our \outofk Approximate \Sperner instance is a function $C^{(k)}: \Delta^{k} \to [k+1]$ satisfying the following hardness result.
For simplicity, we define $\eps\defeq 2^{-n}$ and will use $\eps$ and $2^{-n}$ interchangeably. 
\begin{theorem}
    \label{thm:continuous-asymmetry-main}
    There exists a chain of functions $\{C^{(k)}: \Delta^k \to [k+1]\}_{k\geq 2}$ that can be computed in $\poly(|\vec{x}|)$ time for each $\vec{x}\in \Delta^k$, where $|\vec{x}|$ is the bit complexity of $\vec{x}$ and satisfies the following property. 
    For any $k\geq 2$, it is \PPAD-hard to find three points $\vec{x}^{(1)}, \vec{x}^{(2)}, \vec{x}^{(3)}\in \Delta^{k}$ such that 
    \begin{itemize}[itemsep=0.2em, topsep=0.5em]
        \item {\bf they are close enough to each other:} for any $i,j\in [3]$, $\normtxt{\vec{x}^{(i)}-\vec{x}^{(j)}}{\infty}\leq 2^{-3kn}=\eps^{3k}$, 
        \item {\bf they induce a trichromatic triangle:} $|\settxt{C^{(k)}(\vec{x}^{(i)}):i\in [3]}|=3$. 
    \end{itemize}
\end{theorem}

\subsection{Hard instances with two dimensions}
\label{subsec:warmup-2d}
In this subsection, we will present a family of continuous $2${\rm D}-\Sperner instances that will serve both as the base \PPAD-complete instance we will embed into this chain of \outofk Approximate \Sperner instances and as the second instance $C^{(2)}$ in the chain.
Note that for $k=2$, the \outofk Approximate \Sperner problem is equivalent to the $2${\rm D}-\Sperner problem (\cref{remark:equiv-at-2}).

\paragraph{The instances.} This family of instances is constructed based on the hard instances of \textcite{DBLP:journals/tcs/ChenD09}.
They first construct a core rectangle which has a  boundary condition similar to $2$D-\Sperner, and show that it is \PPAD-hard to find a small trichromatic triangle in the core rectangle. 
In particular, their paper can easily imply that the following rectangular version of the 2D-\Sperner problem\footnote{In the original paper of \textcite{DBLP:journals/tcs/ChenD09}, they name this rectangular version as $2$D-\Brouwer.} is \PPAD-complete. 
\begin{definition}[2{\rm D}-\rect\Sperner]
    \label{def:rect-sperner}
    Let $Q^2_n = \{0,1,\dots,2^n-1\}^2$ denote the 2D grid with side-length of $2^n$ on each dimension. 
    In $2${\rm D}-\rect\Sperner, we are given coloring $C:Q^2_n \to [3]$ satisfying: (1) $C(x,0)\in \{1,2\}$ and $C(x,2^{n}-1)=3$ for any $x\in Q_n$; and (2) $C(0,0)=1, C(2^n-1,0)=2$, and $C(0,y)=C(2^{n}-1,y)=3$ for any $y\in Q_n\setminus \{0\}$. 
    The goal of $2${\rm D}-\rect\Sperner is to find $\vec{x^*}\in ( Q_{n}\setminus \{2^n-1\}  )^2$ such that $|\settxt{C(\vec{x}): x_i \in \settxt{x_{i}^*,x_{i}^*+1} }|=3$.
\end{definition}

\begin{theorem}[\cite{DBLP:journals/tcs/ChenD09}]
    $2${\rm D}-\rect\Sperner is \PPAD-complete.
\end{theorem}
The way~\textcite{DBLP:journals/tcs/ChenD09} use this core rectangle is embedding it into a larger (triangular) $2${\rm D}-\Sperner instance, maintaining the \PPAD-hardness by filling the colors outside this rectangle without creating spurious trichromatic triangles. 
In this paper, we will continue to use this idea but with some slight modifications.
The first modification is that we turn the core rectangle into a continuous one, by mapping each point on the discrete 2D grid to a small square on a continuous 2D grid.
Formally, we map each point $(x,y)$ to the square $[\eps x, \eps(x+1))\times [\eps y,\eps (y+1))$. 
The second modification 
is that we embed this core rectangle in a small region 
inside the $2$-simplex and color for the remainder of the $2${\rm D}-\Sperner instance, ensuring that (1) the bottom $1$-simplex (points $\vec{x}\in \Delta^2$ with the third coordinate $x_3=0$) is equally colored by $1$ and $2$; 
(2) the other two boundaries of the 2-simplex (other than the base $1$-simplex) are colored symmetrically, i.e. $C(\vec{x})=3$ {\em iff} $x_3>0.1$.
\cref{alg:base-instance} and \cref{fig:continuous-mapping} conclude our embedding of the core rectangle inside a $2$-simplex.
We will call this part of the $2$-simplex the {\em core region}. 
\cref{lem:base-instance-bottom-slice,,lem:base-instance-lr-boundaries} conclude the aforementioned boundary conditions.
\cref{lem:base-instance-ppad-hard} establishes the computational complexity of the base instances. 

\begin{figure}[t]
    \centering
    \input{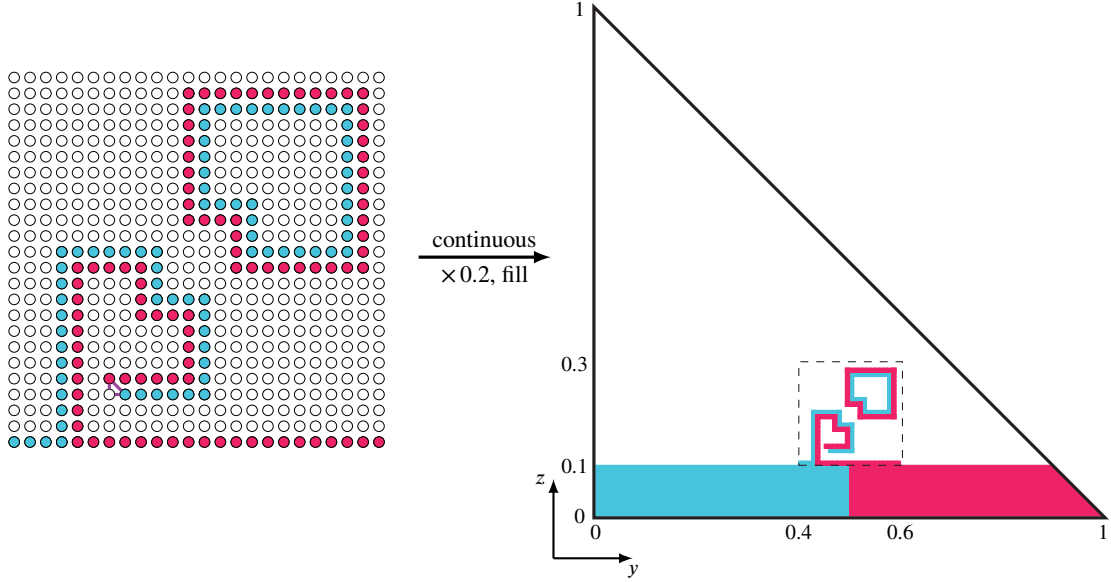}
    \caption{\small An example of the core rectangle in the constructions of \textcite{DBLP:journals/tcs/ChenD09} (left) and its corresponding triangular continuous variant (right) we will consider as the base instance $C$. 
    We present the triangular continuous instance by its projection onto the plane with the first coordinate $x_1=0$. 
    The dashed square presents the {\em core region}.     
    We turn the discrete instances into continuous instances by mapping each vertex in the discrete instances to a small square in the continuous instances. 
    Then, we scale the continuous instance, put it inside a triangle, and fill the remaining part of the triangle to ensure certain boundary conditions that avoid creating spurious solutions (\cref{lem:base-instance-lr-boundaries,,lem:base-instance-bottom-slice,,lem:base-instance-ppad-hard}). 
    }
    \label{fig:continuous-mapping}
\end{figure}

\begin{algorithm2e}[t]
    \caption{Base instance $C$}
    \label{alg:base-instance}

    \DontPrintSemicolon
    \SetKwInOut{Input}{Input\,}
    \SetKwInOut{Output}{Output\,}

    \Input{\,a vector $\vec{x}\in \Delta^2$ and a $2${\rm D}-\rect\Sperner instance $C_{\rect}$ with a side-length of $2^{-n+3}$}

    \Output{\,a color $c\in [3]$}

    \If{$x_3\leq 0.1$\label{line:bottom-begin}}{
        \eIf{$x_2\leq 0.5$}{
            \Return{$1$}
        }{
            \Return{$2$} \tcp*{coloring on the boundaries}\label{line:bottom-end}
        }
    }

    $\eps\gets 2^{-n}$

    \If{$0.4\leq x_2 < 0.6$ {\bf and} $x_3<0.3$}{
        \Return{$C_{\rect}(\floortxt{(1.6\eps)^{-1}\cdot (x_2-0.4)}, \floortxt{(1.6\eps)^{-1}\cdot (x_3-0.1)})$} \label{line:core-diamond} \tcp*{coloring in the core region}
    }

    \Return{$3$} \tcp*{default coloring} \label{line:default-coloring}
\end{algorithm2e}

\begin{fact} 
    \label{lem:base-instance-bottom-slice}
    For any $y\in [0,1]$, the base instance $C$ guarantees that $C(1-y,y,0)=1+\ind(y>0.5)$. 
\end{fact}

\begin{fact}
    \label{lem:base-instance-lr-boundaries}
    For any $z\in [0,1]$, the base instance $C$ guarantees that $C(1-z,0,z)=1+2\cdot \ind(z>0.1)$ and $C(0,1-z,z)=2+\ind(z> 0.1)$. 
\end{fact}

\begin{lemma}
    \label{lem:base-instance-ppad-hard}
    Let $\eps=2^{-n}$.
    It is \PPAD-hard (if $C$ is given by a polynomial-size circuit) 
    to find three points $\vec{x}^{(1)},\vec{x}^{(2)}, \vec{x}^{(3)}\in \Delta^{2}$ for the 2{\rm D}-\Sperner instances generated by \cref{alg:base-instance} such that
    \begin{itemize}[itemsep=0em]
        \item {\bf they are close enough to each other:} $\forall i,j\in [3], \normtxt{\vec{x}^{(i)}-\vec{x}^{(j)}}{\infty}\leq \eps$; and
        \item {\bf they induce a rainbow triangle:} $|\aset{C(\vec{x}^{(i)}):i\in [3]}|=3$.
    \end{itemize}
\end{lemma}

\begin{proof}
    Suppose that $\vec{x}^{(1)},\vec{x}^{(2)},\vec{x}^{(3)}$ form a solution of the 2{\rm D}-\Sperner the instance $C$ generated by \cref{alg:base-instance} using the 2{\rm D}-\rect\Sperner instance $C_{\rect}$. 
    Let $\mathcal{D}=\settxt{(x,y,z)\in \Delta^2: y\in 0.5\pm 0.1, z\in 0.2\pm 0.1}$ denote the core region. Because of the following \cref{lem:all-solutions-of-C-are-in-the-core}, we have $\vec{x}^{(1)}, \vec{x}^{(2)}, \vec{x}^{(3)}\in \mathcal{D}$.
    The technical proof is deferred to \cref{proof:all-solutions-of-C-are-in-the-core}.
    An intuitive proof can be found in \cref{fig:continuous-mapping}.

\begin{restatable}[]{fact}{repi}
    \label{lem:all-solutions-of-C-are-in-the-core}
    Suppose $\vec{x}^{(1)}, \vec{x}^{(2)}, \vec{x}^{(3)}$ form a solution of $C$, then we have $x^{(j)}_2\in [0.4, 0.6)$ and $x^{(j)}_3\in [0.1, 0.3)$ for any $j\in [3]$. 
\end{restatable}

    To complete the proof of \cref{lem:base-instance-ppad-hard}, w.l.o.g., we assume that $C(\vec{x}^{(i)})=i$ for each $i\in [3]$.
    For each $i,j\in [3]$, because $\normtxt{\vec{x}^{(i)}-\vec{x}^{(j)}}{\infty}\leq \eps$, we have 
    \begin{align*}
        \abs{\floor{(1.6\eps)^{-1}\cdot (x_{2}^{(i)}-0.4)} - \floor{(1.6\eps)^{-1}\cdot (x_{2}^{(j)}-0.4)}} &\leq 1~, \quad \text{and }\\
        \abs{\floor{(1.6\eps)^{-1}\cdot (x_{3}^{(i)}-0.4)} - \floor{(1.6\eps)^{-1}\cdot (x_{3}^{(j)}-0.4)}} &\leq 1~.
    \end{align*} 
    Let $\vec{\hat{x}}^{(i)} = (\floortxt{(1.6\eps)^{-1}\cdot (x_{2}^{(i)}-0.4)}, \floortxt{(1.6\eps)^{-1} \cdot (x_{3}^{(i)}-0.1)}) \in Q_{n-3}^2$ for each $i\in [3]$.
    Because \cref{alg:base-instance} gives $C(\vec{x}^{(i)})=C_{\rect}(\vec{\hat{x}}^{(i)})$, we have $C_{\rect}(\vec{\hat{x}}^{(i)}) = i$ for each $i\in [3]$. 
    Therefore, we can obtain a solution, $(\vec{\hat{x}}^{(1)}, \vec{\hat{x}}^{(2)}, \vec{\hat{x}}^{(3)})$, for the $2${\rm}-\rect\Sperner instance $C_{\rect}$ in polynomial-time.
\end{proof}

\subsection{The coordinate converter} 
\label{subsec:warmup-convert}
In this subsection, we will define the coordinate converter for this family of instances.
From now on, because we have a clear context about the base instance $C$, we will assume that $C$ is any fixed $2${\rm D}-\Sperner instance constructed by \cref{alg:base-instance}.
Further, we assume that any function we will define has oracle access to this fixed instance $C$.
Then, the coordinate converter will map any point in the 2-simplex to a point in the 1-simplex, according to $C$. 
This coordinate converter will be used in the next subsection to recursively generate a $(k+1)$-dimensional instance from a $k$-dimensional instance.

Before giving the explicit definition, we first give some motivations for why we define this coordinate converter. 
During our construction, we will repeatedly use the following projection step that maps a point in $\vec{x}\in \Delta^{k}$ to a point $\vec{P}(\vec{x}) \in \Delta^{k-1}$ that is proportional to the first $k$ coordinates of $\vec{x}$.
\begin{definition} [Projection step]
    \label{def:projection}
    For any $k\geq 2$ and any point $\vec{x}\in \Delta^k\setminus \{(\vec{0^k},1)\}$
    , its {\em projection step} $\vec{P}: \Delta^k \rightarrow \Delta^{k-1}$ is defined as $\vec{P}(\vec{x}) \defeq \vec{x}_{1:k}/(1-x_{k+1})$. 
    In particular, we define $\vec{P}(\vec{0^k},1)=(\vec{0^{k-1}},1)$.
\end{definition}
Consider we are now constructing the $3$-dimensional instance $C^{(3)}$.
We can rewrite any point $\vec{x}\in \Delta^3$ in the following form: $\vec{x}=((1-x_4)\cdot \vec{P}(\vec{x}),x_4)$. 
A trivial way to define $C^{(3)}$ is that we use $C^{(2)}$ at ``the bottom'' of the $3$-simplex, and use the next color uniformly ``on top'' of it.
That is, 
we let $C^{(3)}(\vec{x})=C^{(2)}(\vec{P}(\vec{x}))$ if $x_4<0.1$ and let $C^{(3)}(\vec{x})= 4$ otherwise. 
In this trivial definition, it is \PPAD-hard to find a small {\em rainbow} (i.e.~{\em tetra}chromatic) $3$-simplex.
However, because $C^{(2)}$ has a lot of regions that are near the ``color switches'' between two colors and can be found trivially, we have created a lot of spurious easy-to-find trichromatic regions in $C^{(3)}$, i.e.~at the intersection of $C^{(2)}$ color switches and the boundary ($x_4=0.1$) of the new $4$-th color. 

In our actual construction, we again start with a $C^{(2)}$ instance at the ``bottom'' facet of the $3$-simplex. But now, we place another, ``standing'' copy of the base instance $C$ on top of each color switch in the bottom copy of $C^{(2)}$.
Note that, locally, each color switch region in the bottom  $C^{(2)}$ instance is a 1-dimensional manifold.
If we look at a short segment perpendicular to this manifold, it's left half is colored by one color, and it's right half by the other: exactly like the base of $C$! 
So given the projection in the bottom $C^{(2)}$ instance that is close to a color switch, we want to map it to a point $\vec{\tilde{x}}=(1-\tilde{x}_2, \tilde{x}_2)\in \Delta^1$ at the base of the ``top/standing'' $C$ instance. 
Then, we can combine the ``converted'' coordinate $\tilde{x}_2$ and next original coordinate $x_4$ as $((1-x_4)\cdot (1-\tilde{x}_2), (1-x_4)\cdot \tilde{x}_2, x_4)$. Now we can use $C$'s color at $((1-x_4)\cdot (1-\tilde{x}_2), (1-x_4)\cdot \tilde{x}_2, x_4)$ (properly shifted as will be explained later) to color the corresponding point in $C^{(3)}$.

We provide a two-step demonstration of how to define this coordinate converter in \cref{fig:convert,,fig:convert2}.
The details of how it is used for $C^{(3)}$ are provided in \cref{alg:approx-sperner-kgeq3} (next subsection).

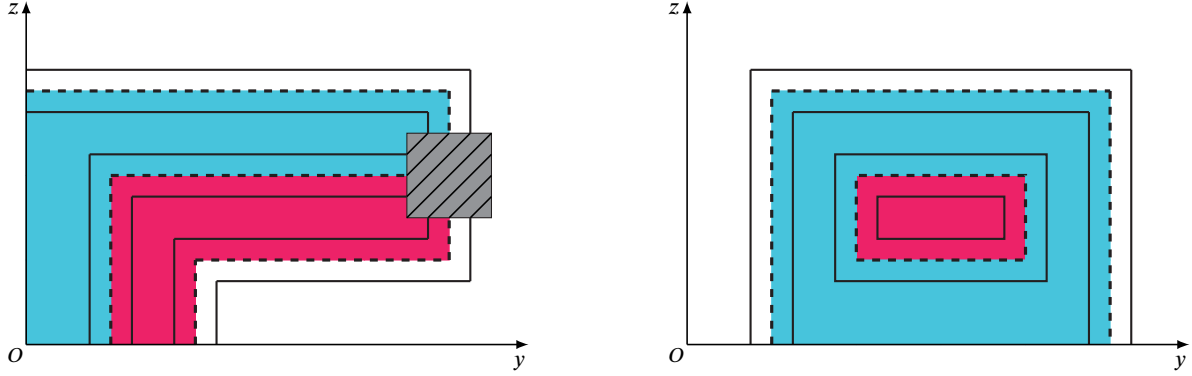
\begin{figure}[t]
    \centering
    \tikzset{%
    every neuron/.style={
        circle,
        draw,
        minimum size=5pt,
        inner sep=0pt,
        very thin
    },
    neuron/.style={
        fill
    },
}

\begin{tikzpicture}[x=0.8cm, y=0.8cm, >=latex, every text node part/.style={align=center}]
    \foreach \x in {1,...,1}
        \foreach \y in {0,...,0} {
            \filldraw [fill=WildStrawberry,draw=WildStrawberry] (32pt*\x,32pt*\y) rectangle (32pt*\x+32pt, 32pt*\y+32pt);
        } 
    \foreach \x in {1,...,4}
        \foreach \y in {1,...,1} {
            \filldraw [fill=WildStrawberry,draw=WildStrawberry] (32pt*\x,32pt*\y) rectangle (32pt*\x+32pt, 32pt*\y+32pt);
        } 
    \foreach \x in {0,...,0}
        \foreach \y in {0,...,1} {
            \filldraw [fill=SkyBlue,draw=SkyBlue] (32pt*\x,32pt*\y) rectangle (32pt*\x+32pt, 32pt*\y+32pt);
        } 
    \foreach \x in {0,...,4}
        \foreach \y in {2,...,2} {
            \filldraw [fill=SkyBlue,draw=SkyBlue] (32pt*\x,32pt*\y) rectangle (32pt*\x+32pt, 32pt*\y+32pt);
        } 

    \draw[draw=Black, fill=Gray, nearly opaque] (144pt, 48pt) rectangle (176pt, 80pt);
    \path[semithick,nearly opaque] (144pt,48pt) edge (176pt,80pt);
    \path[semithick,nearly opaque] (144pt,56pt) edge (168pt,80pt);
    \path[semithick,nearly opaque] (144pt,64pt) edge (160pt,80pt);
    \path[semithick,nearly opaque] (144pt,72pt) edge (152pt,80pt);
    \path[semithick,nearly opaque] (152pt,48pt) edge (176pt,72pt);
    \path[semithick,nearly opaque] (160pt,48pt) edge (176pt,64pt);
    \path[semithick,nearly opaque] (168pt,48pt) edge (176pt,56pt);

    \path[Black,very thick,dashed] (32pt, 0pt) edge (32pt, 64pt);
    \path[Black,very thick,dashed] (32pt, 64pt) edge (144pt, 64pt);
    \path[Black,thick] (24pt, 0pt) edge (24pt, 72pt);
    \path[Black,thick] (24pt, 72pt) edge (144pt, 72pt);
    \path[Black,thick] (40pt, 0pt) edge (40pt, 56pt);
    \path[Black,thick] (40pt, 56pt) edge (144pt, 56pt);
    \path[Black,very thick,dashed] (0pt, 96pt) edge (160pt, 96pt);
    \path[Black,very thick,dashed] (160pt, 96pt) edge (160pt, 80pt);
    \path[Black,thick] (0pt, 104pt) edge (168pt, 104pt);
    \path[Black,thick] (168pt, 104pt) edge (168pt, 80pt);
    \path[Black,thick] (0pt, 88pt) edge (152pt, 88pt);
    \path[Black,thick] (152pt, 88pt) edge (152pt, 80pt);
    \path[Black,very thick,dashed] (64pt, 0pt) edge (64pt, 32pt);
    \path[Black,very thick,dashed] (64pt, 32pt) edge (160pt, 32pt);
    \path[Black,very thick,dashed] (160pt, 32pt) edge (160pt, 48pt);
    \path[Black,thick] (56pt, 0pt) edge (56pt, 40pt);
    \path[Black,thick] (56pt, 40pt) edge (152pt, 40pt);
    \path[Black,thick] (152pt, 40pt) edge (152pt, 48pt);
    \path[Black,thick] (72pt, 0pt) edge (72pt, 24pt);
    \path[Black,thick] (72pt, 24pt) edge (168pt, 24pt);
    \path[Black,thick] (168pt, 24pt) edge (168pt, 48pt);

    \path[->,semithick] (0pt, 0pt) edge (0pt, 130pt);
    \path[->,semithick] (0pt, 0pt) edge (190pt, 0pt);
    \node [] () at (-4pt, -4pt) {\scriptsize $O$};
    \node [] () at (186.3pt, -7pt) {\footnotesize $y$};
    \node [] () at (-5pt, 127pt) {\footnotesize $z$};

    \foreach \x in {2,...,3}
        \foreach \y in {1,...,1} {
            \filldraw [fill=WildStrawberry,draw=WildStrawberry] (32pt*\x+250pt,32pt*\y) rectangle (32pt*\x+282pt, 32pt*\y+32pt);
        } 
    \foreach \x in {1,...,1}
        \foreach \y in {0,...,2} {
            \filldraw [fill=SkyBlue,draw=SkyBlue] (32pt*\x+250pt,32pt*\y) rectangle (32pt*\x+282pt, 32pt*\y+32pt);
        } 
    \foreach \x in {4,...,4}
        \foreach \y in {0,...,2} {
            \filldraw [fill=SkyBlue,draw=SkyBlue] (32pt*\x+250pt,32pt*\y) rectangle (32pt*\x+282pt, 32pt*\y+32pt);
        } 
    \foreach \x in {2,...,3}
        \foreach \y in {0,...,0} {
            \filldraw [fill=SkyBlue,draw=SkyBlue] (32pt*\x+250pt,32pt*\y) rectangle (32pt*\x+282pt, 32pt*\y+32pt);
        } 
    \foreach \x in {2,...,3}
        \foreach \y in {2,...,2} {
            \filldraw [fill=SkyBlue,draw=SkyBlue] (32pt*\x+250pt,32pt*\y) rectangle (32pt*\x+282pt, 32pt*\y+32pt);
        } 
    
    \draw [draw=Black,very thick,dashed] (314pt, 32pt) rectangle (378pt, 64pt);
    \path [Black,very thick,dashed] (282pt, 0pt) edge (282pt, 96pt);
    \path [Black,very thick,dashed] (282pt, 96pt) edge (410pt, 96pt);
    \path [Black,very thick,dashed] (410pt, 96pt) edge (410pt, 0pt);

    \draw [draw=Black,thick] (306pt, 24pt) rectangle (386pt, 72pt);
    \draw [draw=Black,thick] (322pt, 40pt) rectangle (370pt, 56pt);
    \path [Black,thick] (274pt, 0pt) edge (274pt, 104pt);
    \path [Black,thick] (274pt, 104pt) edge (418pt, 104pt);
    \path [Black,thick] (418pt, 104pt) edge (418pt, 0pt);
    \path [Black,thick] (290pt, 0pt) edge (290pt, 88pt);
    \path [Black,thick] (290pt, 88pt) edge (402pt, 88pt);
    \path [Black,thick] (402pt, 88pt) edge (402pt, 0pt);

    \path[->,semithick] (250pt, 0pt) edge (250pt, 130pt);
    \path[->,semithick] (250pt, 0pt) edge (440pt, 0pt);
    \node [] () at (246pt, -4pt) {\scriptsize $O$};
    \node [] () at (436.3pt, -7pt) {\footnotesize $y$};
    \node [] () at (245pt, 127pt) {\footnotesize $z$};
\end{tikzpicture}
    \caption{\small A demonstration of the first steps for our definition of the coordinate converter. 
    Here, we only focus on the region of the coloring that is at most bichromatic because trichromatic regions (e.g., the shadowed part on the left) are \PPAD-hard to find.
    The dashed 1-manifolds represent the ``color switches'' in each instance. 
    We will find the $\ell_{\infty}$ distance (only in terms of the last two coordinates) of each point to these color switches. The solid 1-manifolds represent the sets of points at distance $\eps^2$ to color switches.
    }
    \label{fig:convert}
\end{figure}

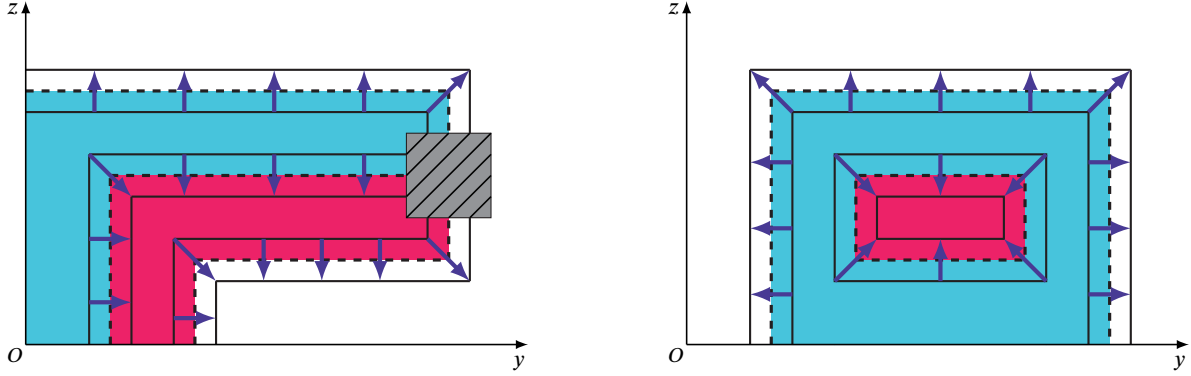
\begin{figure}[t]
    \tikzset{%
    every neuron/.style={
        circle,
        draw,
        minimum size=5pt,
        inner sep=0pt,
        very thin
    },
    neuron/.style={
        fill
    },
}

\begin{tikzpicture}[x=0.8cm, y=0.8cm, >=latex, every text node part/.style={align=center}]

    \filldraw [fill=SkyBlue,draw=SkyBlue] (0pt,0pt) rectangle (24pt, 88pt);
    \filldraw [fill=SkyBlue,draw=SkyBlue] (24pt,72pt) rectangle (144pt, 88pt);
    \filldraw [fill=SkyBlue,draw=SkyBlue] (144pt,80pt) rectangle (152pt, 88pt);

    \fill[fill=SkyBlue,nearly transparent] (24pt,0pt) rectangle (32pt, 72pt);
    \fill[fill=SkyBlue,nearly transparent] (32pt,64pt) rectangle (144pt, 72pt);
    \fill[fill=SkyBlue,nearly transparent] (0pt,88pt) rectangle (152pt, 96pt);
    \fill[fill=SkyBlue,nearly transparent] (152pt,80pt) rectangle (160pt, 96pt);

    \filldraw [fill=WildStrawberry,draw=WildStrawberry] (40pt,0pt) rectangle (56pt, 56pt);
    \filldraw [fill=WildStrawberry,draw=WildStrawberry] (56pt,40pt) rectangle (144pt, 56pt);
    \filldraw [fill=WildStrawberry,draw=WildStrawberry] (144pt,40pt) rectangle (152pt, 48pt);

    \fill[fill=WildStrawberry,nearly transparent] (32pt,0pt) rectangle (40pt, 64pt);
    \fill[fill=WildStrawberry,nearly transparent] (40pt,56pt) rectangle (144pt, 64pt);
    \fill[fill=WildStrawberry,nearly transparent] (56pt,0pt) rectangle (64pt, 40pt);
    \fill[fill=WildStrawberry,nearly transparent] (64pt,32pt) rectangle (152pt, 40pt);
    \fill[fill=WildStrawberry,nearly transparent] (152pt,32pt) rectangle (160pt, 48pt);

    \draw[draw=Black, fill=Gray, nearly opaque] (144pt, 48pt) rectangle (176pt, 80pt);
    \path[semithick,nearly opaque] (144pt,48pt) edge (176pt,80pt);
    \path[semithick,nearly opaque] (144pt,56pt) edge (168pt,80pt);
    \path[semithick,nearly opaque] (144pt,64pt) edge (160pt,80pt);
    \path[semithick,nearly opaque] (144pt,72pt) edge (152pt,80pt);
    \path[semithick,nearly opaque] (152pt,48pt) edge (176pt,72pt);
    \path[semithick,nearly opaque] (160pt,48pt) edge (176pt,64pt);
    \path[semithick,nearly opaque] (168pt,48pt) edge (176pt,56pt);

    \path[Black,very thick,dashed] (32pt, 0pt) edge (32pt, 64pt);
    \path[Black,very thick,dashed] (32pt, 64pt) edge (144pt, 64pt);
    \path[Black,thick] (24pt, 0pt) edge (24pt, 72pt);
    \path[Black,thick] (24pt, 72pt) edge (144pt, 72pt);
    \path[Black,thick] (40pt, 0pt) edge (40pt, 56pt);
    \path[Black,thick] (40pt, 56pt) edge (144pt, 56pt);
    \path[Black,very thick,dashed] (0pt, 96pt) edge (160pt, 96pt);
    \path[Black,very thick,dashed] (160pt, 96pt) edge (160pt, 80pt);
    \path[Black,thick] (0pt, 104pt) edge (168pt, 104pt);
    \path[Black,thick] (168pt, 104pt) edge (168pt, 80pt);
    \path[Black,thick] (0pt, 88pt) edge (152pt, 88pt);
    \path[Black,thick] (152pt, 88pt) edge (152pt, 80pt);
    \path[Black,very thick,dashed] (64pt, 0pt) edge (64pt, 32pt);
    \path[Black,very thick,dashed] (64pt, 32pt) edge (160pt, 32pt);
    \path[Black,very thick,dashed] (160pt, 32pt) edge (160pt, 48pt);
    \path[Black,thick] (56pt, 0pt) edge (56pt, 40pt);
    \path[Black,thick] (56pt, 40pt) edge (152pt, 40pt);
    \path[Black,thick] (152pt, 40pt) edge (152pt, 48pt);
    \path[Black,thick] (72pt, 0pt) edge (72pt, 24pt);
    \path[Black,thick] (72pt, 24pt) edge (168pt, 24pt);
    \path[Black,thick] (168pt, 24pt) edge (168pt, 48pt);

    \path[->,semithick] (0pt, 0pt) edge (0pt, 130pt);
    \path[->,semithick] (0pt, 0pt) edge (190pt, 0pt);
    \node [] () at (-4pt, -4pt) {\scriptsize $O$};
    \node [] () at (186.3pt, -7pt) {\footnotesize $y$};
    \node [] () at (-5pt, 127pt) {\footnotesize $z$};

    \path[BlueViolet,->,ultra thick] (24pt, 72pt) edge (40pt, 56pt);

    \path[BlueViolet,->,ultra thick] (24pt, 40pt) edge (40pt, 40pt);
    \path[BlueViolet,->,ultra thick] (24pt, 16pt) edge (40pt, 16pt);
    \path[BlueViolet,->,ultra thick] (94pt, 72pt) edge (94pt, 56pt);
    \path[BlueViolet,->,ultra thick] (60pt, 72pt) edge (60pt, 56pt);
    \path[BlueViolet,->,ultra thick] (128pt, 72pt) edge (128pt, 56pt);

    \path[BlueViolet,->,ultra thick] (152pt, 88pt) edge (168pt, 104pt);

    \path[BlueViolet,->,ultra thick] (128pt, 88pt) edge (128pt, 104pt);
    \path[BlueViolet,->,ultra thick] (94pt, 88pt) edge (94pt, 104pt);
    \path[BlueViolet,->,ultra thick] (60pt, 88pt) edge (60pt, 104pt);
    \path[BlueViolet,->,ultra thick] (26pt, 88pt) edge (26pt, 104pt);

    \path[BlueViolet,->,ultra thick] (56pt, 40pt) edge (72pt, 24pt);
    \path[BlueViolet,->,ultra thick] (152pt, 40pt) edge (168pt, 24pt);

    \path[BlueViolet,->,ultra thick] (134pt, 40pt) edge (134pt, 24pt);
    \path[BlueViolet,->,ultra thick] (112pt, 40pt) edge (112pt, 24pt);
    \path[BlueViolet,->,ultra thick] (90pt, 40pt) edge (90pt, 24pt);
    \path[BlueViolet,->,ultra thick] (56pt, 10pt) edge (72pt, 10pt);

    \fill[fill=SkyBlue,nearly transparent] (306pt,24pt) rectangle (314pt, 72pt);
    \fill[fill=SkyBlue,nearly transparent] (314pt,24pt) rectangle (386pt, 32pt);
    \fill[fill=SkyBlue,nearly transparent] (314pt,64pt) rectangle (386pt, 72pt);
    \fill[fill=SkyBlue,nearly transparent] (378pt,32pt) rectangle (386pt, 64pt);
    \fill[fill=SkyBlue,nearly transparent] (282pt,0pt) rectangle (290pt, 96pt);
    \fill[fill=SkyBlue,nearly transparent] (402pt,0pt) rectangle (410pt, 96pt);
    \fill[fill=SkyBlue,nearly transparent] (290pt,88pt) rectangle (402pt, 96pt);

    \filldraw [fill=SkyBlue,draw=SkyBlue] (290pt,0pt) rectangle (306pt, 88pt);
    \filldraw [fill=SkyBlue,draw=SkyBlue] (386pt,0pt) rectangle (402pt, 88pt);
    \filldraw [fill=SkyBlue,draw=SkyBlue] (306pt,0pt) rectangle (386pt, 24pt);
    \filldraw [fill=SkyBlue,draw=SkyBlue] (306pt,72pt) rectangle (386pt, 88pt);

    \fill[fill=WildStrawberry,nearly transparent] (314pt,32pt) rectangle (378pt, 64pt);
    \filldraw [fill=WildStrawberry,draw=WildStrawberry] (322pt,40pt) rectangle (370pt, 56pt);

    \draw [draw=Black,very thick,dashed] (314pt, 32pt) rectangle (378pt, 64pt);
    \draw [draw=Black,thick] (306pt, 24pt) rectangle (386pt, 72pt);
    \draw [draw=Black,thick] (322pt, 40pt) rectangle (370pt, 56pt);

    \path [Black,very thick,dashed] (282pt, 0pt) edge (282pt, 96pt);
    \path [Black,very thick,dashed] (282pt, 96pt) edge (410pt, 96pt);
    \path [Black,very thick,dashed] (410pt, 96pt) edge (410pt, 0pt);
    \path [Black,thick] (274pt, 0pt) edge (274pt, 104pt);
    \path [Black,thick] (274pt, 104pt) edge (418pt, 104pt);
    \path [Black,thick] (418pt, 104pt) edge (418pt, 0pt);
    \path [Black,thick] (290pt, 0pt) edge (290pt, 88pt);
    \path [Black,thick] (290pt, 88pt) edge (402pt, 88pt);
    \path [Black,thick] (402pt, 88pt) edge (402pt, 0pt);

    \path[->,semithick] (250pt, 0pt) edge (250pt, 130pt);
    \path[->,semithick] (250pt, 0pt) edge (440pt, 0pt);
    \node [] () at (246pt, -4pt) {\scriptsize $O$};
    \node [] () at (436.3pt, -7pt) {\footnotesize $y$};
    \node [] () at (245pt, 127pt) {\footnotesize $z$};

    \path[BlueViolet,->,ultra thick] (306pt, 24pt) edge (322pt, 40pt);
    \path[BlueViolet,->,ultra thick] (306pt, 72pt) edge (322pt, 56pt);
    \path[BlueViolet,->,ultra thick] (386pt, 24pt) edge (370pt, 40pt);
    \path[BlueViolet,->,ultra thick] (386pt, 72pt) edge (370pt, 56pt);

    \path[BlueViolet,->,ultra thick] (346pt, 24pt) edge (346pt, 40pt);
    \path[BlueViolet,->,ultra thick] (346pt, 72pt) edge (346pt, 56pt);

    \path[BlueViolet,->,ultra thick] (290pt, 88pt) edge (274pt, 104pt);
    \path[BlueViolet,->,ultra thick] (402pt, 88pt) edge (418pt, 104pt);

    \path[BlueViolet,->,ultra thick] (346pt, 88pt) edge (346pt, 104pt);
    \path[BlueViolet,->,ultra thick] (312pt, 88pt) edge (312pt, 104pt);
    \path[BlueViolet,->,ultra thick] (380pt, 88pt) edge (380pt, 104pt);
    \path[BlueViolet,->,ultra thick] (290pt, 69pt) edge (274pt, 69pt);
    \path[BlueViolet,->,ultra thick] (290pt, 44pt) edge (274pt, 44pt);
    \path[BlueViolet,->,ultra thick] (290pt, 19pt) edge (274pt, 19pt);
    \path[BlueViolet,->,ultra thick] (402pt, 69pt) edge (418pt, 69pt);
    \path[BlueViolet,->,ultra thick] (402pt, 44pt) edge (418pt, 44pt);
    \path[BlueViolet,->,ultra thick] (402pt, 19pt) edge (418pt, 19pt);
\end{tikzpicture}
    \caption{\small A demonstration of the second steps for our definition of the coordinate converter. 
    Here, we focus on the region of the coloring that is at most bichromatic.
    Suppose the colors of the blue\,/\,red\,/\,white regions are respectively indexed by 1\,/\,2\,/\,3. 
    For each ``color switch'', we draw arrows from the 1-manifold that is $\eps^2$ away from it and has the lower index to the 1-manifold that is $\eps^2$ away from it and has the higher index.
    The region covered by the arrows are {\em hot regions}.
    Points on the sources of the arrows will have a converted coordinate of $0$, while those on the sinks will have a converted coordinate of $1$.
    We give converted coordinates for other points on the arrows according to their distances to the sources. 
    Points not on any arrow will be marked as either {\em warm} or {\em cold}. 
    }
    \label{fig:convert2}
\end{figure}

\paragraph{Definition.} 
Mathematically, the coordinate converter is defined via the nearest differently colored points. 
More specifically, we consider the $\ell_\infty$ distance on the second and the third coordinate:
\begin{align}
    \label{eqn:modified-distance}
    d(\vec{x}, \vec{y}) := \max\set{|x_2-y_2|, |x_3-y_3|}
\end{align} 
For each point $\vec{x}\in \Delta^2$, we find the infimum distance among all points with a fixed color $c\neq C(\vec{x})$: 
\begin{align}\label{eq:d_min}
    d_{\min}(\vec{x}, c) = \inf_{\vec{y}\in \Delta^2: C(\vec{y})=c} 
    d(\vec{x}, \vec{y})
    ~,
\end{align}
and we define its {\em neighboring color} by 
\begin{align}
    \label{eqn:cnn}
    C_{\sf nn}(\vec{x}) = \argmin_{c\in [3]:c\neq C(\vec{x})} \left(d_{\min}(\vec{x}, c),\ c\right)~.
\end{align}
That is, its neighboring color is defined as the color $c$ that takes the minimum infimum distance $d_{\min}(\vec{x}, c)$, breaking ties by choosing the color with the smallest index.
Or equivalently, the neighboring color of a point equals the different color on the nearest ``color switch'' from the point. 
With this definition of neighboring colors, we define the {\rm nearest neighbor} as the point that attains the minimum in~\cref{eq:d_min}, i.e.:
\begin{align}
    \label{eqn:nn}
    \nn(\vec{x}) = \arginf_{\vec{y}\in \Delta^2: C(\vec{y})=C_{\sf nn}(\vec{x},C)} d(\vec{x}, \vec{y})~.
\end{align}
Note that the color on the nearest neighbor of $\vec{x}$ may be the same as the color on $\vec{x}$, but it must correspond to the limit of an infinite sequence of points that have different colors than $\vec{x}$.
Finally, we define the coordinate converter $\rel$ 
as follows. (The coordinate converter is defined {\em relative to the nearest neighbor}, hence its name.)
We divide the entire $2$-simplex into different regions based on each point's {\em distance to the nearest neighbor}, i.e., $d(\vec{x}, \nn(\vec{x}))$.
\begin{itemize}[itemsep=0.2em,topsep=0.5em]
    \item The {\em hot} region consists of points having distance to the nearest neighbor strictly less than $\eps^2$, i.e., points $\vec{x}$ with $d(\vec{x}, \nn(\vec{x}))<\eps^2$; we say that points in this region are {\em hot}.
    \item The {\em warm} region consists of points having distance to the nearest neighbor between $\eps^2$ (inclusive) and $2\eps^2$ (exclusive), i.e., points $\vec{x}$ with $d(\vec{x}, \nn(\vec{x}))\in [\eps^2,2\eps^2)$; we say that points in this region are {\em warm}.
    \item The {\em cold} region consists of points having distance to the nearest neighbor no less than $2\eps^2$, i.e., points $\vec{x}$ with $d(\vec{x}, \nn(\vec{x}))\geq 2\eps^2$; we say that points in this region are {\em cold}.
\end{itemize}
For points in the {\em hot} region induced by each ``color switch'', 
we map it continuously to a converted coordinate in $(0,1)$, from the boundary with a smaller-indexed color to the boundary with a larger-indexed color.
For points in the {\em warm} region induced by each ``color switch'', we integrally map it to $0$ or $1$ depending on whether its color has a smaller or larger index than that of its neighboring color.
For points in the {\em cold} region, we enforce its converted coordinate to $0$ or $1$ only depending on its own color. 
We summarize the definition in \cref{eqn:coordinate-converter}, where the first two cases are for the {\em hot} and {\em warm} regions and the last two cases are for the {\em cold} region.
\begin{align}
    \label{eqn:coordinate-converter}
    \rel(\vec{x}) = \begin{cases}
        \hfill \left(0.5 - 0.5\eps^{-2} \cdot d(\vec{x}, \vec{nn}(\vec{x}))
        \right)_+ \hfill & \text{if $\vec{x}$ is {\em hot\,/\,warm} and $C_{\sf nn}(\vec{x})>C(\vec{x})$,}\\
        \left(0.5 + 0.5\eps^{-2} \cdot d(\vec{x}, \vec{nn}(\vec{x}))\right)_- & \text{if $\vec{x}$ is {\em hot\,/\,warm} and $C_{\sf nn}(\vec{x})<C(\vec{x})$,}\\
        \hfill 0 \hfill & \text{if $\vec{x}$ is {\em cold} and $C(\vec{x})=1$,}\\
        \hfill 1 \hfill & \text{if $\vec{x}$ is {\em cold} and $C(\vec{x})\in \{2,3\}$.}
    \end{cases}
\end{align}
Because warm points satisfy $d(\vec{x},\nn(\vec{x}))\geq \eps^2$ and hot points are those with $d(\vec{x}, \nn(\vec{x}))<\eps^2$, we have the following fact.
\begin{fact}
    \label{fact:hot-versus-value}
    For any point $\vec{x}\in \Delta^2$, we have $\rel(\vec{x})\in (0,1)$ if and only if $\vec{x}$ is hot.
\end{fact}
As discussed above, our use of this coordinate converter is to convert each point $\vec{x}\in \Delta^3$ to $((1-x_4)\cdot (1-\rel(\vec{P}(\vec{x}))), (1-x_4)\cdot \rel(\vec{P}(\vec{x})), x_4)\in \Delta^2$.
If $\vec{P}(\vec{x})$ is warm or cold, $\rel(\vec{P}(\vec{x}))=0$ and $\vec{x}$ will be converted to either $(1-z,0,z)$ or $(0,1-z,z)$ for $z=x_4$. 
Recall \cref{lem:base-instance-lr-boundaries}, where we have shown for $C$ that
\begin{align*}
 \forall z\in [0,1], \quad \quad & \hfill C(1-z,0,z) = \begin{cases}
    1 & z\leq 0.1~, \\
    3 & z>0.1~;
 \end{cases}
 \hfill
 &
 \hfill
 C(0,1-z,z) = \begin{cases}
    2 & z\leq 0.1~, \\
    3 & z>0.1~.
 \end{cases}
 \hfill
\end{align*} 
One unified way to interpret this coloring is to consider the colors on the line segments $(1-z,0,z)$ and $(0,1-z,z)$ as functions of $z \in [0,1]$: when $z > 0.1$ they take their color from $(0,0,1)$, and when $z \le 0.1$ they take the color of the other endpoint. 
This symmetry will be useful later when we that there is no spurious solutions are created where two warm\,/\,cold points are close to each other but have a different converted coordinate (i.e., one equals $0$, but the other equals $1$). 
Therefore, we can consider the following equivalence between the converted coordinates.

\begin{definition}[Equivalence between converted coordinates]
    \label{def:equiv-converted-coordinates}
    For any converted coordinates $\tilde{x}, \tilde{x}'\in [0,1]$, we say they are equivalent (denoted as $\tilde{x}\simrel \tilde{x}'$) if we have either $\tilde{x}=\tilde{x}'$ or $\tilde{x},\tilde{x}'\in \{0,1\}$. 
\end{definition}

\begin{claim}
    \label{lem:color-equiv-for-converted}
    For any converted coordinates $\tilde{x}, \tilde{x}'\in [0,1]$ such that $\tilde{x}\sim \tilde{x}'$, and any $z\in [0,1]$, we have 
    \begin{align*}
        C((1-z)\cdot (1-\tilde{x}), (1-z)\cdot \tilde{x}, z) = C((1-z)\cdot (1-\tilde{x}'), (1-z)\cdot \tilde{x}', z)~.
    \end{align*}
\end{claim}

\subsubsection{Key properties of the coordinate converter}
Next, we establish the key properties about the coordinate converter on this family of instances.

\paragraph{Property I: Polynomial time computation}
First, we show that all previous functions can be 
``roughly'' 
implemented in polynomial time and we can thus use them freely in our reduction.
That is, we can always output the true converted coordinate $\rel$ and, whenever $\vec{x}$ is hot or warm (i.e., $d(\vec{x}, \nn(\vec{x}))<2\eps^2$), we can also output the neighboring color $C_{\sf nn}(\vec{x})$.

\begin{lemma}
    \label{lem:poly-time-oracle-and-converter}
    Given oracle access to $C$, there is a polynomial-time algorithm that takes any $\vec{x}\in \Delta^2$ as input and that outputs 
    $\rel(\vec{x})$.
    Furthermore, if $d(\vec{x}, \nn(\vec{x}))<2\eps^2$, the algorithm can also compute $C_{\sf nn}(\vec{x})$ in polynomial time.
\end{lemma}
\begin{proof}
    According to our definition of the coordinate converter~\cref{eqn:coordinate-converter}, it suffices to only compute the exact value of $d(\vec{x}, \nn(\vec{x}))$ for $\vec{x}$ such that $d(\vec{x}, \nn(\vec{x}))\leq 2\eps^2$.
    The value of $d(\vec{x}, \nn(\vec{x}))$ can be computed by its distance to each ``color switch''. 
    Note that we only have a constant number of all the color switches outside the core  region  are defined by $O(1)$ line segments (\cref{alg:base-instance}, or see \cref{fig:continuous-mapping}), thus we can enumerate over them. 
    For those ``color switches'' strictly inside the core region, because each point in $C_{\rect}$ is transformed into a $1.6\eps\times 1.6\eps$ square (\cref{alg:base-instance}) and $0.8\eps \gg 2\eps^2$, we only need to check at most $8$ points in the 2{\rm D}-\rect\Sperner instance $C_{\rect}$, which can be done by querying the oracle of $C$ according to \cref{alg:base-instance}.
\end{proof}

\paragraph{Property II: Symmetry}
As mentioned earlier, we will use the converted coordinate with a new coordinate (in a higher dimension) to embed a new copy of the 2{\rm D}-\Sperner instance inside our construction.
That is, supposing the converted coordinate is $\tilde{x}$ and the new coordinate is $z$, we will define a vector, $((1-z)\cdot (1-\tilde{x}), (1-z)\cdot \tilde{x}, z)$, 
for the new copy. 
Later, to argue that there is no issue arising from the fact that we can have very close points with a converted coordinate of $\tilde{x}=0$ and of $\tilde{x}=1$, we also need the following property that the converted coordinates of $(1-z,0,z)$ and $(0,1-z,z)$ are equivalent and their neighboring colors are symmetric. 
This stronger symmetry also implies a characterization of the temperature on the boudnary:
\begin{itemize}[itemsep=0.2em,topsep=0.5em]
    \item If $z\in 0.1\pm \eps^2$, points $(1-z,0,z)$ and $(0,1-z,z)$ are {\em hot} and the converted coordinate equals $0.5+0.5\eps^{-2}\cdot(z-0.1)$.
    \item If $z\in 0.1\pm 2\eps^2$ but $z\notin 0.1\pm \eps^2$, points $(1-z,0,z)$ and $(0,1-z,z)$ are {\em warm} and the converted coordinate is the rounding of $0.5+0.5\eps^{-2}\cdot(z-0.1)$ to its nearest integer in $\{0,1\}$.
    \item If $z\notin 0.1\pm 2\eps^2$, points $(1-z,0,z)$ and $(0,1-z,z)$ are {\em cold} and the converted coordinate is in $\{0,1\}$ (where we don't care its exact value).
\end{itemize}
Here, the value of $|z-0.1|$ gives the {\em hot} and {\em warm} points' distances to their nearest neighbor. 

\begin{lemma}[Symmetry on converted coordinate]
    \label{lem:symmetry-on-converted-coordinates}
    Warm and hot points on the line segments $(1-x_3, 0,x_3)$ and $(0, 1-x_3,x_3)$ have $x_3\in 0.1\pm 2\eps^2$. For any such $x_3$, we have 
    \begin{align}
        \label{eqn:stronger-symmetry-on-converted-coordinates}
        \rel(1-x_3, 0,x_3) = \left(1/2 + \frac{x_3-0.1}{2 \eps^2}\right)_{[0,1]} =\rel(0, 1-x_3, x_3)~,
    \end{align}
    and the neighboring color is characterized as follows
    \begin{align*}
        C_{\sf nn}(1-x_3,0,x_3) = \begin{cases}
            3 & \text{if $x_3\leq 0.1$,}\\
            1 & \text{if $x_3>0.1$,}
        \end{cases}
        \quad
        \text{and}
        \quad
        C_{\sf nn}(0,1-x_3,x_3) = \begin{cases}
            3 & \text{if $x_3\leq 0.1$,}\\
            2 & \text{if $x_3>0.1$.}
        \end{cases}
    \end{align*}
\end{lemma}
\begin{proof}
    Note that we guarantee our coloring to satisfy all the following equations.
    \begin{align*}
        \forall \vec{y}\in \Delta^2, \quad
        \begin{cases}
        C(\vec{y}) = 1 & \quad \text{if } y_3\leq 0.1, y_2\leq  0.5~,\\
        C(\vec{y}) = 2 & \quad \text{if } y_3\leq 0.1, y_2> 0.5~,\\
        C(\vec{y}) = 3 & \quad \text{if } y_3> 0.1, y_2\notin 0.5\pm 0.1~,\\
        C(\vec{y}) = 3 & \quad \text{if } y_3>0.3 ~.\\
        \end{cases}
    \end{align*}

    \paragraph{Consider any $x_3\leq 0.1$.} 
    We have $C(1-x_3,0,x_3)=1$ and $C(0,1-x_3,x_3)=2$. 
    Note that any point $\vec{y}$ that has color other than $1$ has either $y_3\geq 0.1$ or $y_2>0.5$. 
    Hence, we have $\vec{nn}(1-x_3,0,x_3)=(0.9,0,0.1)$ for any $x_3\leq 0.1$. 
    We then have $d((1-x_3,0,x_3), \vec{nn}(1-x_3,0,x_3))=0.1-x_3$ and $C_{\sf nn}(1-x_3,0,x_3)=3$.
    According to the definition of the coordinate converter \cref{eqn:coordinate-converter}, if $0.1-x_3< 2\eps^2$ (i.e., $x_3> 0.1-2\eps^2$), the point $(1-x_3,0,x_3)$ is hot or warm, and
    \begin{align*}
        \rel(1-x_3,0,x_3) = \left(0.5-0.5\eps^{-2}\cdot (0.1-x_3)\right)_+ = \left(0.5 + \frac{x_3-0.1}{2\eps^2}\right)_{[0,1]}~,
    \end{align*}
    which matches \cref{eqn:stronger-symmetry-on-converted-coordinates}.
    Otherwise, we have $d((1-x_3,0,x_3), \vec{nn}(1-x_3,0,x_3))\geq 2\eps^2$ and thus the point $(1-x_3,0,x_3)$ is cold.
    Similarly, for any $x_3\leq 0.1$, we have $\vec{nn}(0,1-x_3,x_3)=(0,0.9,0.1)$, $d((0,1-x_3,x_3), \vec{nn}(0,1-x_3,x_3))=0.1-x_3$, and $C_{\sf nn}(0,1-x_3,x_3)=3$.
    According to the definition of the temperature, if $x_3>0.1-2\eps^2$, the point $(0,1-x_3,x_3)$ is hot or warm, and we can obtain the right-handed-side equation of \cref{eqn:stronger-symmetry-on-converted-coordinates} according to \cref{eqn:coordinate-converter}. 
    Otherwise, we have $d((0,1-x_3,x_3), \vec{nn}(0,1-x_3,x_3))\geq 2\eps^2$, which implies that the point $(0,1-x_3,x_3)$ is cold.

    \paragraph{Consider any $x_3>0.1$.}
    We have $C(1-x_3,0,x_3)=C(0,1-x_3,x_3)=3$. 
    Note that any point $\vec{y}$ that has color other than $3$ has either $y_3\leq 0.1$ or $y_2\in [0.4, 0.6]$.
    Hence, if $x_3\geq 0.1+2\eps^2$, we have $d((1-x_3,0,x_3), \nn(1-x_3,0,x_3))\geq 2\eps^2$ and thus the point $(1-x_3,0,x_3)$ is cold.
    Otherwise, if $x_3<0.1+2\eps^2$, we have $\nn(1-x_3,0,x_3)=(0.9,0,0.1)$. 
    We then have $d((1-x_3,0,x_3), \vec{nn}(1-x_3,0,x_3))=z-0.1<2\eps^2$ and $C_{\sf nn}(1-x_3,0,x_3)=1$.
    The point $(1-x_3,0,x_3)$ is hot or warm in this case. 
    Further, according to the definition of the coordinate converter \cref{eqn:coordinate-converter}, 
    \begin{align*}
        \rel(1-x_3,0,x_3) = \left(0.5+0.5\eps^{-2}\cdot (x_3-0.1)\right)_- = \left(0.5+ \frac{x_3-0.1}{2\eps^{2}} \right)_{[0,1]},
    \end{align*}
    matching \cref{eqn:stronger-symmetry-on-converted-coordinates}.
    On the other hand, note that any point $\vec{y}$ has color other than $3$ and $y_3\geq 0.1$ must have $y_2+y_3\leq 0.9$. 
    The distance $d$ between $(0,1-x_3,x_3)$ and any such point is at least $0.05$.
    Note that $0.05\gg 2\eps^2$.
    Hence, when $x_3\geq 0.1+2\eps^2$, we have $d((0,1-x_3,x_3), \nn(0,1-x_3,x_3))\geq 2\eps^2$ and thus according to 
    the definition of the temperature, 
    the point $(0,1-x_3,x_3)$ is cold. 
    Otherwise, if $x_3<0.1+2\eps^2$, we have $\nn(0,1-x_3,x_3) = (x_3-0.1,1-x_3,0.1)$. 
    We then have $d((0,1-x_3,x_3), \vec{nn}(0,1-x_3,x_3))=x_3-0.1$ and $C_{\sf nn}(0,1-x_3,x_3)=2$.
    According to the definition of the coordinate converter \cref{eqn:coordinate-converter}, 
    \begin{align*}
        \rel(0,1-x_3,x_3) = \left(0.5+0.5\eps^{-2}\cdot(x_3-0.1)\right)_-  = \left(0.5+ \frac{x_3-0.1}{2\eps^2}\right)_{[0,1]},
    \end{align*}
    matching \cref{eqn:stronger-symmetry-on-converted-coordinates}.
\end{proof}

\begin{corollary}
    \label{lem:converted-equiv-for-converted}
    For any converted coordinates $\tilde{x}, \tilde{x}'\in [0,1]$ such that $\tilde{x}\sim \tilde{x}'$, and any $z\in [0,1]$, we have 
    \begin{align*}
        \rel((1-z)\cdot (1-\tilde{x}), (1-z)\cdot \tilde{x}, z) \,\simrel\, \rel((1-z)\cdot (1-\tilde{x}'), (1-z)\cdot \tilde{x}', z)~.
    \end{align*}
\end{corollary}

\paragraph{Property III: Lipshitzness}  
We show that the coordinate converter is Lipschitz in bichromatic regions in the following sense: consider we view the interval $[0,1]$ as a loop, where $0,1$ corresponds to the same point on the loop, and the distance between any $\rel(\vec{x})$ and $\rel(\vec{x'})$ on the loop can be upper bounded by a $\poly(1/\eps)$ factor of the distance between $\vec{x}$ and $\vec{x'}$. 
We say a point is in a bichromatic region if its neighbourhood, defined in \cref{def:neighbourhood}, is bichromatic.
And, the Lipschitzness property is formalized by \cref{lem:lipschitz-of-the-converter} with a slightly stronger argument.

\begin{definition}
    \label{def:neighbourhood}
    Let $\eps=2^{-n}$.
    For any $\vec{x} \in \Delta^{2}$, let $\mathcal{N}(\vec{x})=\settxt{\vec{y}\in \Delta^2: |x_2-y_2|, |x_3-y_3|\leq \eps/2}$ denote the neighbourhood of $\vec{x}$ with a side length of $\eps$. 
\end{definition}

\begin{lemma}
    \label{lem:lipschitz-of-the-converter}
    For any coloring $C$ and any pair of points $\vec{x}, \vec{x'}$, at least one of the following properties is satisfied:
    \begin{itemize}[itemsep=0em]
        \item {\bf one is in a trichromatic region:} there are three different colors in $\mathcal{N}(\vec{x})$ (or, $\mathcal{N}(\vec{x'})$);
        \item {\bf Lipschitz in the hot\&warm regions:} $|\rel(\vec{x})-\rel(\vec{x'})|\leq \eps^{-2}\cdot \normtxt{\vec{x}-\vec{x'}}{\infty}$; 
        \item {\bf both are in the warm\&cold regions:} $\rel(\vec{x}),\rel(\vec{x'})\in \{0,1\}$. 
    \end{itemize}
\end{lemma}
\begin{proof}
    It suffices to show that the Lipschitz property, $|\rel(\vec{x})-\rel(\vec{x'})|\leq \eps^{-2} \cdot \normtxt{\vec{x}-\vec{x'}}{\infty}$, under the following three assumptions: (1) $\mathcal{N}(\vec{x})$ and $\mathcal{N}(\vec{x'})$ are both at most bichromatic; (2) $\normtxt{\vec{x}-\vec{x'}}{\infty} < \eps^2$; and (3) $\rel(\vec{x})\in (0,1)$ (equivalently by \cref{fact:hot-versus-value}, $d(\vec{x}, \nn(\vec{x}))<\eps^2$). 
    The second assumption makes sense because $\rel(\vec{x}), \rel(\vec{x'})\in [0,1]$ and the second property can trivially hold if the assumption is violated.
    In addition, the second and the third assumption ensures that $d(\vec{x}, \nn(\vec{x})), d(\vec{x'}, \nn(\vec{x'})) < 2\eps^2$, which
    can be obtained by the triangle inequality when $\vec{x}, \vec{x'}$ have the same color and otherwise is trivial by the second assumption.
    This ensures that we must go to the first two cases of \cref{eqn:coordinate-converter} when computing the converted coordinates of $\vec{x}, \vec{x'}$. 
    Next, we discuss two cases on the colors of $\vec{x}$ and $\vec{x'}$.

    \paragraph{Case 1: $\vec{x}$ and $\vec{x'}$ have different colors.}
    Since $d(\vec{x},\vec{x'})\leq \normtxt{\vec{x}-\vec{x'}}{\infty}$, we have 
    $d(\vec{x}, \vec{nn}(\vec{x}))\leq \normtxt{\vec{x}-\vec{x'}}{\infty}$ and $d(\vec{x'}, \vec{nn}(\vec{x'}))\leq \normtxt{\vec{x}-\vec{x'}}{\infty}$.
    According to \cref{eqn:coordinate-converter}, we have
    \begin{align*}
        \rel(\vec{x}) &\in 0.5 \pm 0.5\eps^{-2} \cdot d(\vec{x}, \vec{nn}(\vec{x})) \subseteq 0.5 \pm 0.5\eps^{-2} \cdot \norm{\vec{x}-\vec{x'}}{\infty}~, \quad \text{and}\\
        \rel(\vec{x'}) &\in 0.5 \pm 0.5\eps^{-2} \cdot d(\vec{x'}, \vec{nn}(\vec{x'})) \subseteq 0.5 \pm 0.5\eps^{-2} \cdot \norm{\vec{x}-\vec{x'}}{\infty}~.
    \end{align*}
    Therefore, we have $\abstxt{\rel(\vec{x})-\rel(\vec{x'})}\leq \eps^{-2} \cdot \normtxt{\vec{x}-\vec{x'}}{\infty}$.

    \paragraph{Case 2: $\vec{x}$ and $\vec{x'}$ have the same color.}
    Suppose that $c:=C(\vec{x})=C(\vec{x'})$.
    Because of the definition of the coordinate converter (\cref{eqn:coordinate-converter}), the third assumption implies $d(\vec{x}, \vec{nn}(\vec{x}))<\eps^2$.
    Since $C(\vec{x})\neq C_{\sf nn}(\vec{x})$, according to the definition of the neighbourhood $\mathcal{N}(\vec{x})$ (\cref{def:neighbourhood}), this further implies that $\mathcal{N}(\vec{x})$ is bichromatic combined with our first assumption.
    Suppose that $c'$ is the other color in $\mathcal{N}(\vec{x})$. 
    Or say, $c':=C_{\sf nn}(\vec{x})$. 
    Recall that we have showed that $d(\vec{x'}, \nn(\vec{x'}))<2\eps^{2}$.
    Note that the following set is strictly contained in $\mathcal{N}(\vec{x})$:
    \begin{align*}
        \set{\vec{y}\in \Delta^2: y_2\in x'_2\pm 2\eps^2, y_3\in x'_3\pm 2\eps^2}~.
    \end{align*}
    We don't have any color other than $c,c'$ within a distance of $\eps^2$ from $\vec{x'}$, and thus $C_{\sf nn}(\vec{x'}) =  c'$.
    This means 
    \begin{align*}
        d(\vec{x}, \vec{nn}(\vec{x})) &= d_{\min}(\vec{x}, c') = \inf_{\vec{y}\in \Delta^2: C(\vec{y})=c'} d(\vec{x}, \vec{y})~,\\
        d(\vec{x'}, \vec{nn}(\vec{x'})) &= d_{\min}(\vec{x'}, c') = \inf_{\vec{y}\in \Delta^2: C(\vec{y})=c'} d(\vec{x'}, \vec{y})~.
    \end{align*}
    Therefore, we have
    \begin{align*}
        \abs{d(\vec{x}, \vec{nn}(\vec{x}))-d(\vec{x'}, \vec{nn}(\vec{x'}))} &= \abs{\inf_{\vec{y}\in \Delta^2: C(\vec{y})=c'} d(\vec{x}, \vec{y}) - \inf_{\vec{y}\in \Delta^2: C(\vec{y})=c'} d(\vec{x'}, \vec{y})}
        \\
        &\leq \inf_{\vec{y}\in \Delta^2: C(\vec{y})=c'} \abs{d(\vec{x}, \vec{y})-d(\vec{x'},\vec{y})}
        \\
        &\leq 2\norm{\vec{x}-\vec{x'}}{\infty}~.
    \end{align*}
    Because $C(\vec{x})=c=C(\vec{x'})$ and $C_{\sf nn}(\vec{x}) = c' = C_{\sf nn}(\vec{x'})$, when we compute $\rel(\vec{x})$ and $\rel(\vec{x'})$, we follow the same case of \cref{eqn:coordinate-converter}. 
    Hence, 
    \begin{align*}
        \abs{\rel(\vec{x}) - \rel(\vec{x'})} \leq 0.5\eps^{-2} \cdot \abs{d(\vec{x}, \vec{nn}(\vec{x}))-d(\vec{x'}, \vec{nn}(\vec{x'}))} \leq \eps^{-2} \cdot \norm{\vec{x}-\vec{x'}}{\infty}~,
    \end{align*}
    and we have obtained the lemma for this case.
\end{proof}

Furthermore, a very similar (and slightly simpler) Lipschitzness can be established for points on the two boundaries other than the base $1$-simplex. 
The proof is much simpler -- we only need to do simple calculations based on our previous characterization of the coordinate converter (\cref{lem:symmetry-on-converted-coordinates}).
The technical proof is deferred to \cref{proof:lipschitz-of-the-converter-on-boundaries}.
\begin{restatable}[]{lemma}{repii}
    \label{lem:lipschitz-of-the-converter-on-boundaries}
    For any $z,z'\in [0,1]$ and any pair of points $\vec{x}, \vec{x'}$ such that $\vec{x}\in \{(0,1-z,z),\, (1-z,0,z)\}$ and $\vec{x'}\in \{(0,1-z',z'),\, (1-z',0,z')\}$, at least one of the following properties is satisfied:
    \begin{itemize}[itemsep=0em]
        \item {\bf Lipschitz in the hot\&warm regions:} $|\rel(\vec{x})-\rel(\vec{x'})|\leq \eps^{-2}\cdot |z-z'|$; 
        \item {\bf both are in the warm\&cold regions:} $\rel(\vec{x}),\rel(\vec{x'})\in \{0,1\}$. 
    \end{itemize}
\end{restatable}

\paragraph{Property IV: Cold on the top}
Finally, we show that the points that have reasonably high values on the third dimension ($x_3>0.5$) are all cold. 
This lemma will help ensure that there are no spurious solutions in which any of the three points has a coordinate that has a very high value (e.g., $>1-2^{-\Omega(n)}$) in any of its $\ell$-th projection step (see \cref{def:lth-projection} and next subsection for details). 

\begin{fact}
    \label{lem:trivial-converter}
    For any $\vec{x}\in \Delta^{2}$ such that $x_3\geq 0.5$, we have $\rel(\vec{x})=1$.
\end{fact}
\begin{proof}
    According to \cref{alg:base-instance}, we have $C(\vec{x})=3$ for any $\vec{x}\in \Delta^2$ such that $x_3>0.3$. 
    Therefore, if $x_3\geq 0.5$, $d(\vec{x}, \vec{nn}(\vec{x}))\geq 0.2\gg 2\eps^2$ and thus $\vec{x}$ is cold. 
    According to \cref{eqn:coordinate-converter}, $\rel(\vec{x})=1$. 
\end{proof}

\subsection{Hard instances with three or more dimensions}
\label{subsec:warmup-kd}
In this subsection, we give the constructions for $C^{(k)}$ with $k\geq 3$.
Before giving the construction, for convenience, we introduce the $\ell$-th projection steps we will recursively use to project a point.

\begin{definition}[$\ell$-th projection steps]
    \label{def:lth-projection}
    For any $0\leq \ell\leq k-1$, we use $\vec{P}^{(\ell)}(\vec{x})\in \Delta^{k-\ell}$ to denote the vector we obtain after applying a total number of $\ell$ projection steps on $\vec{x}$, i.e., we let $P^{(0)}(\vec{x})=\vec{x}$ and let $\vec{P}^{(\ell)}(\vec{x}) = \vec{P}(\vec{P}^{(\ell-1)}(\vec{x}))$ for any $ \ell\geq 1$.
\end{definition}

We also define a {\em modified neighbor color} that is consistent with our definition of the coordinate converter for those cold points $\vec{x}$ having $d(\vec{x}, \vec{nn}(\vec{x}))\geq 2\eps^2$.
\begin{align}
    \label{eqn:modified-neighboring-color}
    \hat{C}_{\sf nn}(\vec{x}) = \begin{cases}
        C_{\sf nn}(\vec{x}) & \text{if $d(\vec{x}, \vec{nn}(\vec{x}))< 2\eps^2$,}\\
        \hfill 2 \hfill & \text{if $d(\vec{x}, \vec{nn}(\vec{x})) \geq 2\eps^2$ and $C(\vec{x})=1$,}\\
        \hfill 1 \hfill & \text{if $d(\vec{x}, \vec{nn}(\vec{x})) \geq 2\eps^2$ and $C(\vec{x})\in \{2,3\}$.}
    \end{cases}
\end{align}
The efficient computation of this function follows from \cref{lem:poly-time-oracle-and-converter}.
This will be an auxillary function for our algorithm, which allows us to use the following fact.
\begin{fact}
\label{fact:modified-neighboring-color}
    For any $\vec{x}\in \Delta^2$, we have $C(\vec{x})<\hat{C}_{\sf nn}(\vec{x})$ if $\rel(\vec{x})=0$, and $C(\vec{x})>\hat{C}_{\sf nn}(\vec{x})$ if $\rel(\vec{x})=1$.
\end{fact}

Our construction (formalized in \cref{alg:approx-sperner-kgeq3}) for three or more dimensions is recursive, i.e., we lift $C^{(k)}$ from $C^{(k-1)}$ by the base instance $C$. 
During this process, we recursively compute a sequence of points $\vec{y}^{(2)}, \dots, \vec{y}^{(k)}\in \Delta^2$,  that can each be thought of as a different projection of $\vec{x}$ to a $2$-simplex; additionally, we compute for each $i\in \{2,\dots,k\}$ a palette of three candidate colors $\vec{c}^{(i)}\in \binom{[i+1]}{3}$ that can be used to color the $i$-th copy of $C$; 
in particular, we can use $c^{(k)}_j$ for $j=C(\vec{y}^{(k)})$ as our output color.

In more detail,
our construction begins with projecting $\vec{x}$ to the base $2$-simplex (i.e., letting $\vec{y}^{(2)}=\vec{P}^{(k-2)}(\vec{x})$) and a vector of the only three colors $\vec{c}^{(2)}=(1,2,3)$. 
This initialization is consistent with our $2$-dimensional instance $C^{(2)}$ and the base instance $C$ constructed in \cref{subsec:warmup-2d}.
Then, we recursively lift the $(i-1)$-dimensional instance to the $i$-dimensional instance for each $i\in \{3,\cdots, k\}$ by the following steps:
\begin{enumerate}[itemsep=0.2em, topsep=0.5em]
    \item We locally project the previous projection $\vec{y}^{(i-1)}$ to a point $(1-\tilde{y},\tilde{y})$ on $1$-simplex, using our converted coordinate $\tilde{y}=\rel(\vec{y}^{(i-1)})$. 
    \item We let the new coordinate $z$ (i.e., the next coordinate in $C^{(i)}$) be the last coordinate of $\vec{x}$'s projection to the base $i$-simplex (i.e., $\vec{P}^{(k-i)}(\vec{x})$). 
    Note that this coordinate $z$ was not used in the construction of $C^{(i-1)}$.
    \item With $\tilde{y}$ and $z$, we define our new projection of $\vec{x}$ to a $2$-simplex by $\vec{y}^{(i)}=((1-z)\cdot (1-\tilde{y}), (1-z)\cdot \tilde{y}, z)$, which can be viewed as an weighted average of our previous converted point $(1-\tilde{y}, \tilde{y}, 0)$ on the base $1$-simplex and topmost point $(0,0,1)$ via the new coordinate $z$.
    \item To define the $i$-th palette of three new candidate colors, we look at the previous color $c^{(i-1)}_j$ for $j=C(\vec{y}^{(i-1)})$ and the previous neighboring color $c^{(i-1)}_{j'}$ for $j'=C_{\sf nn}(\vec{y}^{(i-1)})$, and use $c^{(i-1)}_{j}, c^{(i-1)}_{j'}$ and the new color $i+1$ as the three candidate colors. 
\end{enumerate}

\begin{algorithm2e}[t]
    \caption{\outofk Approximate \Sperner Instance $C^{(k)}(\vec{x})$}
    \label{alg:approx-sperner-kgeq3}

    \DontPrintSemicolon
    \SetKwInOut{Input}{Input\,}
    \SetKwInOut{Output}{Output\,}

    \Input{~vector $\vec{x}\in \Delta^k$}

    \Output{~color $c\in [k+1]$}

    $\vec{y}^{(2)} \gets \vec{P}^{(k-2)}(\vec{x})$ \tcp*{initiate $\vec{y}$ by $\vec{x}$'s projection to the base $2$-simplex}

    $\vec{c}^{(2)}\gets (1,2,3)$ \tcp*{initiate the set of 3 colors}

    \For{$i\in \{3,\dots, k\}$}{

        $y_3^{(i)} \gets P_{i+1}^{(k-i)}(\vec{x})$
        \tcp*{find the new $z$-coordinate from the projection of $\vec{x}$}
        \label{line:y_3}

        $y_2^{(i)} \gets (1-y_3^{(i)}) \cdot \rel ( \vec{y}^{(i-1)} )$
        \label{line:y_2}

        $y_1^{(i)} \gets (1-y_3^{(i)})\cdot \left(1-\rel ( \vec{y}^{(i-1)} )\right) $ 
        \tcp*{convert $\vec{y}$ to a point on the $2$-simplex}
        \label{line:y_1}

        $j_1\gets \min \left\{C(\vec{y}^{(i-1)}),\ \hat{C}_{\sf nn}(\vec{y}^{(i-1)})\right\}$

        $j_2\gets \max \left\{C(\vec{y}^{(i-1)}),\ \hat{C}_{\sf nn}(\vec{y}^{(i-1)})\right\}$

        $\vec{c}^{(i)} \gets \left(c_{j_1}^{(i-1)}, c_{j_2}^{(i-1)}, i+1\right)$
        \label{line:obtain-new-set-of-colors}
        \tcp*{obtain the new set of 3 colors}
    }

    $j \gets C \left(\vec{y}^{(k)}\right)$
            
    \Return{$c_j^{(k)}$}
\end{algorithm2e}

Because \cref{alg:approx-sperner-kgeq3} is clearly polynomial-time, in the rest of this subsection, we will prove that we can recover a solution of the $2${\rm D}-\Sperner instance $C$ in polynomial time from any solution of our instances $C^{(k)}$ to finish the analysis. 

\subsubsection{Recovering a 2D solution} 
\label{subsubsec:recovery}

From now on, we assume that the tuple $(\vec{x}^{(1)}, \vec{x}^{(2)}, \vec{x}^{(3)})$ is any good solution for the \outofk Approximate \Sperner instance $C^{(k)}$ satisfying that $\normtxt{\vec{x}^{(i)}-\vec{x}^{(j)}}{\infty}\leq 2^{-3kn}=\eps^3$ for any $i,j\in [3]$. 
We will prove that we can recover a $2${\rm D}-\Sperner solution from it in polynomial time.
\begin{lemma}
    \label{lem:poly-time-recovery}
    Given oracle access to a $2${\rm D}-\Sperner instance $C$ and a tuple of points $(\vec{x}^{(1)}, \vec{x}^{(2)}, \vec{x}^{(3)})$ which satisfy $\normtxt{\vec{x}^{(i)}-\vec{x}^{(j)}}{\infty}\leq 2^{-3kn}$ for any $i,j\in [3]$, and which are trichromatic in the \outofk Approximate \Sperner instance $C^{(k)}$ constructed by \cref{alg:approx-sperner-kgeq3},
    then there is a polynomial-time algorithm that finds a tuple of points $(\vec{\hat{x}}^{(1)}, \vec{\hat{x}}^{(2)}, \vec{\hat{x}}^{(3)})$ which satsify $\normtxt{\vec{\hat{x}}^{(i)}-\vec{\hat{x}}^{(j)}}{\infty}\leq 2^{-n}$ for any $i,j\in [3]$ and which are trichromatic in the $2${\rm D}-\Sperner instance $C$.
\end{lemma}

In the analysis, 
we use $\vec{y}^{(i)}(\vec{x}) = (y^{(i)}_1(\vec{x}), y^{(i)}_2(\vec{x}), y^{(i)}_3(\vec{x}))$
to denote the {\em intermediate projections} we use in \cref{alg:approx-sperner-kgeq3} when the inputs are $\vec{x}$ and $C$. 
Similarly, we use
$\vec{c}^{(i)}(\vec{x}) = (c^{(i)}_1(\vec{x}), c^{(i)}_2(\vec{x}), c^{(i)}_3(\vec{x}))$ 
to denote the {\em intermediate palettes} we use in \cref{alg:approx-sperner-kgeq3} for $\vec{x}$.
In the rest of this section, because the subscripts we will use for $\vec{c}^{(i)}$ can be very complicated, we will use $c_j^{(i)}(\vec{x})$ and $c^{(i)}(\vec{x}, j)$ interchangeably for better presentation.
For any vector $\vec{y}\in \Delta^2$, we use $i^*(\vec{y})$ to denote the first non-zero index of $\vec{y}$, i.e., 
\begin{align}
    \label{eqn:first-nonzero-index}
    i^*(\vec{y}) = \begin{cases}
        1 & \text{if $y_1>0$,}\\
        2 & \text{if $y_1=0$ and $y_2>0$,}\\
        3 & \text{otherwise.}
    \end{cases}
\end{align}

\begin{algorithm2e}[t]
    \caption{Recover a $2${\rm D}-\Sperner solution from \outofk Approximate \Sperner solutions}
    \label{alg:recover-2d-sol}

    \DontPrintSemicolon
    \SetKwInOut{Input}{Input\,}
    \SetKwInOut{Output}{Output\,}

    \Input{~vectors $\vec{x}^{(1)}, \vec{x}^{(2)}, \vec{x}^{(3)}\in \Delta^k$}

    \Output{~vectors $\vec{\hat{x}}^{(1)}, \vec{\hat{x}}^{(2)}, \vec{\hat{x}}^{(3)}\in \Delta^2$}

    \For{$i\in \{2,3,\dots,k-1\}$}{
        \For{$j\in [3]$}{

            \If{$C$ has $3$ different colors in $\mathcal{N}(\vec{y}^{(i)}(\vec{x}^{(j)}))$}{
                $\vec{\hat{x}}^{(1)}, \vec{\hat{x}}^{(2)}, \vec{\hat{x}}^{(3)} \gets $ any 3 points colored differently by $C$ in $\mathcal{N}(\vec{y}^{(i)}(\vec{x}^{(j)}))$\tcp*{the definition of $\vec{y}^{(i)}(\vec{x}^{(j)})$ follows \cref{alg:approx-sperner-kgeq3}}

                \Return{$\vec{\hat{x}}^{(1)}, \vec{\hat{x}}^{(2)}, \vec{\hat{x}}^{(3)}$} \tcp*{trichromatic triangle found while simulation}
            }
        }

    }

    \Return{$\vec{y}^{(k)}(\vec{x}^{(1)}), \vec{y}^{(k)}(\vec{x}^{(2)}), \vec{y}^{(k)}(\vec{x}^{(3)})$}
    \tcp*{trichromatic triangle found when lifting $C^{(k-1)}$ to $C^{(k)}$}
\end{algorithm2e}

Our recovery algorithm (formalized in \cref{alg:recover-2d-sol}) simulates \cref{alg:approx-sperner-kgeq3} for each $\vec{x}^{(j)}$.
At time $i\in [2,k-1]$ when we have computed $\vec{y}^{(i)}(\vec{x}^{(j)})$ for each $j\in [3]$, we examine whether one of them is inside a trichromatic region, i.e., whether $C$ has three different colors in $\mathcal{N}(\vec{y}^{(i)}(\vec{x}^{(j)}))$ for some $j\in [3]$.
Because we turn each point in the 2D-\rect\Sperner instance to a $1.6\eps\times 1.6\eps$ square in the core region, and outside the core region all color switches are defined by $O(1)$ line segments, we can compute all the connected parts of each color within any $\mathcal{N}(\vec{y})$ in polynomial time. 
After we finish the simulation, we simply output $\vec{y}^{(k)}(\vec{x}^{(1)}), \vec{y}^{(k)}(\vec{x}^{(2)}), \vec{y}^{(k)}(\vec{x}^{(3)})$ as a solution for $C$. 

It is clear that any output during the simulation phase of this recovery algorithm gives a valid solution for the 2{\rm D}-\Sperner instance $C$. 
To prove \cref{lem:poly-time-recovery}, we only need to show that the three intermediate projections after the simulation form a valid solution for $C$ if we output after the simulation phase. 
Or equivalently, if each $\vec{y}^{(i)}(\vec{x}^{(j)})$ does not lie in a trichromatic region, the final converted coordinates, $\vec{y}^{(k)}(\vec{x}^{(1)}), \vec{y}^{(k)}(\vec{x}^{(2)}), \vec{y}^{(k)}(\vec{x}^{(3)})$, give a solution for $C$.

Our proof will be considers two cases; the first is where
the given solution $\vec{x}^{(1)}, \vec{x}^{(2)}, \vec{x}^{(3)}$  for $C^{(k)}$ further enjoys the following property: 
\begin{align}
\label{eqn:assumption-for-simple-proof-of-recovery}
\forall 2\leq i\leq k, \forall j\in [3], \quad P_{i+1}^{(k-i)}(\vec{x}^{(j)})\leq 0.9~.
\end{align}
One benefit of first considering this case is that we don't have too crazy projections in \cref{alg:recover-2d-sol}, which can help us significantly simplify the proof.
The formal intermediate technical result is presented in \cref{cor:no-sols-in-sim-implies-final-ones}.
Later, in the second step, we will explain how to prove \cref{lem:poly-time-recovery} in the  case where this property is not satisfied.

\paragraph{Case 1: the output satisfies \cref{eqn:assumption-for-simple-proof-of-recovery}.}

The following lemma presents us a Lipschitz property of the projection step for this case. 

\begin{lemma}
\label{lem:lipschitz-of-projection}
Consider any $k\geq 2$. 
If $\vec{x}, \vec{x'}\in \Delta^k$ satisfy $x_{k+1},x'_{k+1}\leq 0.9$, then we have 
\begin{align*}
    \normtxt{\vec{P}(\vec{x})-\vec{P}(\vec{x'})}{\infty} \leq 110\cdot \normtxt{\vec{x}-\vec{x'}}{\infty}~.
\end{align*}
\end{lemma}
\begin{proof}
For any $i\in [k]$, we have
\begin{align*}
    \abs{P_i(\vec{x}) - P_i(\vec{x'})} 
    &= 
    \abs{\frac{x_i}{1-x_{k+1}} - \frac{x'_i}{1-x'_{k+1}}}
    \\
    &\leq
    \abs{\frac{x_i}{1-x_{k+1}} - \frac{x'_i}{1-x_{k+1}}} + \abs{\frac{x'_i}{1-x_{k+1}} - \frac{x'_i}{1-x'_{k+1}}}
    \\
    &=
    \abs{\frac{x_i-x'_i}{1-x_{k+1}}} + |x'_i|\cdot \abs{\frac{x_{k+1}-x'_{k+1}}{(1-x_{k+1})(1-x'_{k+1})}}
    \\
    &\leq
    10 \cdot |x_i-x'_i| + 100\cdot |x_{k+1}-x'_{k+1}|
    \tag{$x_{k+1},x'_{k+1}\leq 0.9$}
    \\
    &\leq
    110\cdot \norm{\vec{x}-\vec{x'}}{\infty}~. &
\end{align*}
According to the definition of $\ell_\infty$-norm, we complete the proof.
\end{proof}

We can further use this Lipschitzness of the projection step to obtain Lipschitzness for the intermediate projections we use in \cref{alg:approx-sperner-kgeq3}. 

\begin{lemma}
\label{lem:intermediate-projections-are-close-to-each-other}
    Consider any $k\geq 2$ and any $\vec{x}^{(1)}, \vec{x}^{(2)}\in \Delta^k$.
    Suppose that $P_{i+1}^{(k-i)}(\vec{x}^{(j)})\leq 0.9$ for any $2\leq i\leq k$ and any $j\in [2]$.
    Also, suppose that $\mathcal{N}(\vec{y}^{(i)}(\vec{x}^{(j)}))$ is at most bichromatic for any $2\leq i<k$ and any $j\in [2]$.
    Then, for any $2\leq i\leq k$, we have Lipschitzness for the third coordinate of the intermediate projections:
    \begin{align*}
        \abstxt{y^{(i)}_3(\vec{x}^{(1)})-y^{(i)}_3(\vec{x}^{(2)})}\leq 2^{O(k)}\cdot \normtxt{\vec{x}^{(1)}-\vec{x}^{(2)}}\infty~.
    \end{align*}
    Furthermore, at least one of the following is satisfied for the second coordinates of the intermediate projections:
    \begin{itemize}[itemsep=0.1em, topsep=0.4em]
        \item {\bf Lipschitz:} $\abstxt{y_2^{(i)}(\vec{x}^{(1)})-y_2^{(i)}(\vec{x}^{(2)})}\leq 2^{2in+O(k)}\cdot \normtxt{\vec{x}^{(1)}-\vec{x}^{(2)}}\infty$, or
        \item {\bf both are on the left/right boundaries:} for any $j\in [2]$, we have either $y_1^{(i)}(\vec{x}^{(j)})=0$ or $y_2^{(i)}(\vec{x}^{(j)})=0$. 
    \end{itemize}
\end{lemma}
\begin{proof}
    We prove this lemma by induction on $i$. The base case is when $i=2$. Its proof is straightforward since we always have $\vec{y}^{(2)}(\vec{x})=\vec{P}^{(k-2)}(\vec{x})$ and the first bullet is satisfied by the Lipschitzness of projection \cref{lem:lipschitz-of-projection}. 

    Consider any $i_0\geq 3$. 
    Suppose that we have proved this lemma for $i<i_0$. 
    Next, we consider when $i=i_0$.
    Note that we have $y_3^{(i)}(\vec{x})=P_{i+1}^{(k-i)}(\vec{x})$ for any $\vec{x}$.
    By the premise that $P_{i+1}^{(k-i)}(\vec{x}^{(j)})\leq 0.9$, we can use \cref{lem:lipschitz-of-projection} to get:
    \begin{align*}
        \abs{y_3^{(i)}(\vec{x}^{(1)})-y_3^{(i)}(\vec{x}^{(2)})}
        &\leq 
        \normtxt{\vec{P}^{(k-i)}(\vec{x}^{(1)})-\vec{P}^{(k-i)}(\vec{x}^{(2)})}{\infty}
        \\
        &\leq 
        110\cdot \normtxt{\vec{P}^{(k-i-1)}(\vec{x}^{(1)})-\vec{P}^{(k-i-1)}(\vec{x}^{(2)})}{\infty}
        \\
        &\leq
        \cdots
        \\
        &\leq 
        110^{k-i}\cdot \normtxt{\vec{x}^{(1)}-\vec{x}^{(2)}}{\infty}
        \leq 2^{7k}\cdot \normtxt{\vec{x}^{(1)}-\vec{x}^{(2)}}{\infty}~,
    \end{align*}
    which establish the first statement of this lemma. 

    Next, we establish the second statement of this lemma, in which we need to prove at least one of the bullets is satisfied.
    For the induction hypothesis, we will use the following more specific version of the first bullet 
    \[
        \abstxt{y_2^{(i)}(\vec{x}^{(1)})-y_2^{(i)}(\vec{x}^{(2)})}\leq 2^{2in+7k+\log  2i}\cdot \normtxt{\vec{x}^{(1)}-\vec{x}^{(2)}}\infty~.
    \]
    Note that $y_2^{(i)}(\vec{x})=(1-y_3^{(i)}(\vec{x}))\cdot \rel(\vec{y}^{(i-1)}(\vec{x}))$. 
    Then, it is easy to obtain that
    \begin{align*}
        \abs{y_2^{(i)}(\vec{x}^{(1)})-y_2^{(i)}(\vec{x}^{(2)})} 
        &= 
        \abs{(1-y_3^{(i)}(\vec{x}^{(1)}))\cdot \rel(\vec{y}^{(i-1)}(\vec{x}^{(1)})) - (1-y_3^{(i)}(\vec{x}^{(2)}))\cdot \rel(\vec{y}^{(i-1)}(\vec{x}^{(2)}))}
        \\
        &\leq
        \rel(\vec{y}^{(i-1)}(\vec{x}^{(1)}))\cdot \abs{y_3^{(i)}(\vec{x}^{(1)})-y_3^{(i)}(\vec{x}^{(2)})} 
        \\
        & \qquad \quad + (1-y_3^{(i)}(\vec{x}^{(2)})) \cdot \abs{\rel(\vec{y}^{(i-1)}(\vec{x}^{(1)})) - \rel(\vec{y}^{(i-1)}(\vec{x}^{(2)})) }
        \\
        &\leq
        \abs{y_3^{(i)}(\vec{x}^{(1)})-y_3^{(i)}(\vec{x}^{(2)})} + \abs{\rel(\vec{y}^{(i-1)}(\vec{x}^{(1)})) - \rel(\vec{y}^{(i-1)}(\vec{x}^{(2)})) }
        \\
        &\leq 
        2^{7k}\cdot \normtxt{\vec{x}^{(1)}-\vec{x}^{(2)}}\infty + \abs{\rel(\vec{y}^{(i-1)}(\vec{x}^{(1)})) - \rel(\vec{y}^{(i-1)}(\vec{x}^{(2)})) }~.
    \end{align*}

    Suppose that our induction hypothesis gives the first bullet for $i-1$.
    Combining the Lipschitzness on the third coordinate, we have 
    \begin{align*}
        \normtxt{\vec{y}^{(i-1)}(\vec{x}^{(1)})-\vec{y}^{(i-1)}(\vec{x}^{(2)})}\infty &\leq (2^{2(i-1)n+7 k + \log 2(i-1)}+2^{7k}) \cdot \normtxt{\vec{x}^{(1)}-\vec{x}^{(2)}}\infty~. 
    \end{align*}
    If $\rel(\vec{y}^{(i-1)}(\vec{x}^{(1)})), \rel(\vec{y}^{(i-1)}(\vec{x}^{(2)})) \in \{0,1\}$, we have the second bullet for $i=i_0$.
    Otherwise, according to \cref{lem:lipschitz-of-the-converter}, we have 
    \begin{align*}
        \abs{y_2^{(i)}(\vec{x}^{(1)})-y_2^{(i)}(\vec{x}^{(2)})} 
        &\leq 
        2^{7k}\cdot \normtxt{\vec{x}^{(1)}-\vec{x}^{(2)}}\infty + 2^{2n} \cdot \abs{\vec{y}^{(i-1)}(\vec{x}^{(1)}) - \vec{y}^{(i-1)}(\vec{x}^{(2)})}
        \\
        &\leq
        2^{7k}\cdot \normtxt{\vec{x}^{(1)}-\vec{x}^{(2)}}{\infty} + 2^{2n}\cdot (2^{2(i-1)n+7k+\log 2(i-1)}+2^{7k}) \cdot \normtxt{\vec{x}^{(1)}-\vec{x}^{(2)}}{\infty}
        \\
        &\leq
        2^{7k}\cdot \normtxt{\vec{x}^{(1)}-\vec{x}^{(2)}}{\infty} + (2i-1)\cdot 2^{2in+7k}\cdot \normtxt{\vec{x}^{(1)}-\vec{x}^{(2)}}{\infty}
        \\
        &\leq 
        2^{2in+7k + \log 2i}\cdot \normtxt{\vec{x}^{(1)}-\vec{x}^{(2)}}{\infty} ~, 
    \end{align*}
    which gives the first bullet for $i=i_0$.

    On the other hand, suppose that the induction hypothesis further gives us the second bullet for $i-1$. 
    That is, for any $j\in [2]$, we have 
    \[
        \vec{y}^{(i-1)}(\vec{x}^{(j)}) \in \set{\left(0,\ 1-y_3^{(i-1)}(\vec{x}^{(j)}),\ y_3^{(i-1)}(\vec{x}^{(j)})\right), ~~ \left(1-y_3^{(i-1)}(\vec{x}^{(j)}),\ 0,\ y_3^{(i-1)}(\vec{x}^{(j)})\right)}~.
    \]
    If $\rel(\vec{y}^{(i-1)}(\vec{x}^{(1)})), \rel(\vec{y}^{(i-1)}(\vec{x}^{(2)})) \in \{0,1\}$, we have the second bullet for $i=i_0$.
    Otherwise, according to \cref{lem:lipschitz-of-the-converter-on-boundaries}, we have 
    \begin{align*}
        \abs{y_2^{(i)}(\vec{x}^{(1)})-y_2^{(i)}(\vec{x}^{(2)})} 
        &\leq 
        2^{7k}\cdot \normtxt{\vec{x}^{(1)}-\vec{x}^{(2)}}\infty + 2^{2n} \cdot \abs{y_3^{(i-1)}(\vec{x}^{(1)}) - y_3^{(i-1)}(\vec{x}^{(2)})}
        \\
        &\leq
        2^{7k}\cdot \normtxt{\vec{x}^{(1)}-\vec{x}^{(2)}}{\infty} + 2^{2n+7k} \cdot \normtxt{\vec{x}^{(1)}-\vec{x}^{(2)}}{\infty}
        \\
        &\leq 
        2^{2in+7k+1}\cdot \normtxt{\vec{x}^{(1)}-\vec{x}^{(2)}}{\infty} ~, 
    \end{align*}
    which gives the first bullet for $i=i_0$.
\end{proof}

Suppose that \cref{alg:recover-2d-sol} fails to give us any solution during the simulation phase, i.e., for any $2\leq i\leq k-1$ and $j\in [3]$, $\mathcal{N}(\vec{y}^{(i)}(\vec{x}^{(j)}))$ is at most bichromatic. 
Since the solution of the \outofk Approximate \Sperner instance satisfies $\normtxt{\vec{x}^{(j_1)}-\vec{x}^{(j_2)}}{\infty}\leq 2^{-3kn}$, we have $\abstxt{y_3^{(i)}(\vec{x}^{(j_1)})-y_3^{(i)}(\vec{x}^{(j_2)})} < 2^{-2n}=\eps^2$ for any $2\leq i\leq k$ and any $j_1,j_2\in [3]$, and further that
\begin{itemize}[itemsep=0.2em, topsep=0.5em]
    \item {\bf the intermediate projections are close to each other:} $\abstxt{y_2^{(i)}(\vec{x}^{(j_1)})-y_2^{(i)}(\vec{x}^{(j_2)})}<2^{-2n}=\eps^2$ for any $j_1,j_2\in [3]$, or
    \item {\bf all intermediate projections are on the left/right boundaries:} for any $j\in [3]$, we have either $y_1^{(i)}(\vec{x}^{(j)})=0$ or $y_2^{(i)}(\vec{x}^{(j)})=0$.
\end{itemize}

Next, we can characterize the intermediate palettes used in \cref{alg:approx-sperner-kgeq3} by \cref{lem:characterization-of-intermediate-colors}.
The characterization gives equivalence between the set of relevant colors in the palette.
In the first case, where the intermediate projections of the three input vectors are close to each other (and at least one of them is not on the left/right boundaries), the palettes are exactly the same.
In the second case, where the intermediate projections of the input vectors are all on the left/right boundaries, the only two relevant colors for the points, $c^{(i)}(\vec{x}, i^*(\vec{y}^{(i)}(\vec{x})))$ and $c^{(i)}_3(\vec{x})$, are respectively equal. 

\begin{lemma}
    \label{lem:characterization-of-intermediate-colors}
    Suppose that $P_{i+1}^{(k-i)}(\vec{x}^{(j)})\leq 0.9$ for any $2\leq i\leq k$ and any $j\in [3]$.
    Suppose that $\mathcal{N}(\vec{y}^{(i)}(\vec{x}^{(j)}))$ is not trichromatic for any $2\leq i<k$ and $j\in [3]$. 
    For any $2\leq i\leq k$, at least one of the following holds:
    \begin{itemize}
        \item The intermediate projections are close to each other, and we have that the corresponding palettes are the same: $\vec{c}^{(i)}(\vec{x}^{(j_1)})=\vec{c}^{(i)}(\vec{x}^{(j_2)})$ for any $j_1,j_2\in [3]$.
        \item All intermediate projections are on the left/right boundaries, and the palettes may be different on an irrelevant color, but we still have that both of the following hold: 
        \begin{itemize}
            \item The color of the first non-zero coordinate $c^{(i)}\left(\vec{x}^{(j)}, i^*(\vec{y}^{(i)}(\vec{x}^{(j)}))\right)$ is the same across all  $j\in [3]$; and
            \item the 3rd color is the same, $c^{(i)}_3(\vec{x}^{(j)})=i+1$ for any $j\in [3]$.
        \end{itemize} 
    \end{itemize}
\end{lemma}

Since the output $\vec{x}^{(1)}, \vec{x}^{(2)}, \vec{x}^{(3)}$ is trichromatic in $C^{(k)}$, we have 
\begin{align}
    \label{eqn:final-trichromatic}
    \abs{\set{c^{(k)}(\vec{x}^{(j)},C(\vec{y}^{(k)}(\vec{x}^{(j)})))}_{j\in [3]}}=3~.
\end{align}
We should always have the first bullet of \cref{lem:characterization-of-intermediate-colors} for $k$ when each $\mathcal{N}(\vec{y}^{(i)}(\vec{x}^{(j)}))$ is not trichromatic, because otherwise the second bullet of \cref{lem:characterization-of-intermediate-colors} violates \cref{eqn:final-trichromatic} as $C(\vec{y}^{(k)}(\vec{x}^{(j)}))\in \{i^*(\vec{y}^{(k)}(\vec{x}^{(j)})), 3\}$.
Note that
the first bullet of \cref{lem:characterization-of-intermediate-colors} and \cref{eqn:final-trichromatic} imply that the colors in the base instance $C(\vec{y}^{(k)}(\vec{x}^{(1)}))$, $C(\vec{y}^{(k)}(\vec{x}^{(2)}))$, and $C(\vec{y}^{(k)}(\vec{x}^{(3)}))$ should be distinct. 
We can always guarantee that the tuple $\vec{y}^{(k)}(\vec{x}^{(1)}), \vec{y}^{(k)}(\vec{x}^{(2)}), \vec{y}^{(k)}(\vec{x}^{(3)})$ gives a solution to $C$.
\begin{corollary}
    \label{cor:no-sols-in-sim-implies-final-ones}
    Suppose that $P_{i+1}^{(k-i)}(\vec{x}^{(j)})\leq 0.9$ for any $2\leq i\leq k$ and any $j\in [3]$.
    Suppose that $\mathcal{N}(\vec{y}^{(i)}(\vec{x}^{(j)}))$ is not trichromatic for any $2\leq i<k$ and $j\in [3]$. 
    Then, we have $|\settxt{C(\vec{y}^{(k)}(\vec{x}^{(j)})): j\in [3]}|=3$.
\end{corollary}

\begin{proof}[Proof of \cref{lem:characterization-of-intermediate-colors}]
    We prove this lemma by induction.
    The base case is when $i=2$. 
    At the beginning of \cref{alg:approx-sperner-kgeq3}, we have $\vec{y}^{(2)}(\vec{x}^{(j)})=\vec{P}^{(k-2)}(\vec{x}^{(j)})$ and $\vec{c}^{(2)}(\vec{x}^{(j)})=(1,2,3)$ for any $j\in [3]$. 
    Because of the Lipschitzness \cref{lem:lipschitz-of-projection}, the first bullet is satisfied. 

    Consider any $i_0\geq 3$. Assume that we have established this lemma for any $i=i_0-1$. 
    Next, we establish this lemma for $i=i_0$.

    First, consider that the first bullet holds for $i-1$. We have $\abstxt{y_2^{(i-1)}(\vec{x}^{(j_1)})-y_2^{(i-1)}(\vec{x}^{(j_2)})}< 2^{-2n}=\eps^2$ for any $j_1,j_2\in [3]$. 
    Hence, $d(\vec{y}^{(i-1)}(\vec{x}^{(j_1)}), \vec{y}^{(i-1)}(\vec{x}^{(j_2)}))\leq \eps^2$ for any $j_1,j_2\in [3]$. 
    We discuss two cases on whether there is a hot point in $\vec{y}^{(i-1)}(\vec{x}^{(1)}), \vec{y}^{(i-1)}(\vec{x}^{(2)})$ and $\vec{y}^{(i-1)}(\vec{x}^{(3)})$. 
    \begin{itemize}[itemsep=0.2em,topsep=0.5em]
        \item W.l.o.g., suppose that $\vec{y}^{(i-1)}(\vec{x}^{(1)})$ is a hot point. 
        According to \cref{lem:lipschitz-of-the-converter} and \cref{line:y_1,line:y_2,line:y_3}, all the $i$-th intermediate projections are close to each other. 
        If there are two different colors in $\settxt{C(\vec{y}^{(i-1)}(\vec{x}^{(j)}))}_{j\in [3]}$, each $\vec{y}^{(i-1)}(\vec{x}^{(j)})$ is hot because we clearly have $d(\vec{y}^{(i-1)}(\vec{x}^{(j)}),\nn(\vec{y}^{(i-1)}(\vec{x}^{(j)})))<\eps^2$ and \cref{fact:hot-versus-value}. 
        Otherwise there is only one color in $\settxt{C(\vec{y}^{(i-1)}(\vec{x}^{(j)}))}_{j\in [3]}$, since $d(\cdot, \cdot)$ is a metric, we can obtain the following upper bound on each point's distance to the nearest neighbor by the triangle inequality and \cref{fact:hot-versus-value}:
        \begin{align*}
            d(\vec{y}^{(i-1)}(\vec{x}^{(j)}), \nn(\vec{y}^{(i-1)}(\vec{x}^{(1)}))) &\leq d(\vec{y}^{(i-1)}(\vec{x}^{(j)}), \vec{y}^{(i-1)}(\vec{x}^{(1)})) + d(\vec{y}^{(i-1)}(\vec{x}^{(1)}), \nn(\vec{y}^{(i-1)}(\vec{x}^{(1)}))) \\
            &< \eps^2+\eps^2 =  2\eps^2~.
        \end{align*}
        Therefore, each $\vec{y}^{(i-1)}(\vec{x}^{(j)})$ is either hot or warm by definition, and thus we have $\hat{C}_{\sf nn}(\vec{y}^{(i-1)}(\vec{x}^{(j)}))=C_{\sf nn}(\vec{y}^{(i-1)}(\vec{x}^{(j)}))$ for each $j\in [3]$. 
        Since they are in a region that is at most bichromatic, 
        the color set $\{C(\vec{y}^{(i-1)}(\vec{x}^{(j)})), \hat{C}_{\sf nn}(\vec{y}^{(i-1)}(\vec{x}^{(j)}))\}$ is then the same across all $j\in [3]$.
        Since $\vec{c}^{(i-1)}(\vec{x}^{(1)})=\vec{c}^{(i-1)}(\vec{x}^{(2)})=\vec{c}^{(i-1)}(\vec{x}^{(3)})$, we then have the same $\vec{c}^{(i)}(\vec{x}^{(j)})$ across all $j\in [3]$, which gives the first bullet of this lemma. 
        \item Otherwise, suppose that $\vec{y}^{(i-1)}(\vec{x}^{(j)})$ is warm or cold for each $j\in [3]$. We have $\rel(\vec{y}^{(i-1)}(\vec{x}^{(j)}))\in \{0,1\}$ for any $j\in [3]$ in this case.
        The $i$-th intermediate projections are on the left/right boundaries, i.e., we have $y_1^{(i)}(\vec{x}^{(j)})=0$ or $y_2^{(i)}(\vec{x}^{(j)})=0$ for each $j\in [3]$.
        Note that there is only one color in $\settxt{C(\vec{y}^{(i-1)}(\vec{x}^{(j)}))}_{j\in [3]}$, because otherwise we clearly have $d(\vec{y}^{(i-1)}(\vec{x}^{(j)}),\nn(\vec{y}^{(i-1)}(\vec{x}^{(j)})))<\eps^2$ and all points are hot. 
        Note that \cref{alg:approx-sperner-kgeq3} ensures in this scenario that 
        \begin{align*}
            i^*(\vec{y}^{(i)}(\vec{x}^{(j)}))=1 
            &\Leftrightarrow 
            \rel(\vec{y}^{(i-1)}(\vec{x}^{(j)}))=0\\
            &\Leftrightarrow
            C(\vec{y}^{(i-1)}(\vec{x}^{(j)}))<\hat{C}_{\sf nn}(\vec{y}^{(i-1)}(\vec{x}^{(j)}))\\
            &\Leftrightarrow
            c^{(i)}(\vec{x}^{(j)},1) = c^{(i-1)}(\vec{x}^{(j)}, C(\vec{y}^{(i-1)}(\vec{x}^{(j)})))
        \end{align*}
        where the second {\em iff} is obtained by \cref{fact:modified-neighboring-color}.
        We have $c^{(i)}(\vec{x}^{(j)}, i^*(\vec{y}^{(i)}(\vec{x}^{(j)}))) = c^{(i-1)}(\vec{x}^{(j)}, C(\vec{y}^{(i-1)}(\vec{x}^{(j)})))$ for each $j\in [3]$.
        Since $\vec{c}^{(i-1)}(\vec{x}^{(1)})=\vec{c}^{(i-1)}(\vec{x}^{(2)})=\vec{c}^{(i-1)}(\vec{x}^{(3)})$ and $C(\vec{y}^{(i-1)}(\vec{x}^{(1)}))=C(\vec{y}^{(i-1)}(\vec{x}^{(2)}))=C(\vec{y}^{(i-1)}(\vec{x}^{(3)}))$, 
        we have the second bullet for $i$.
    \end{itemize}

    Second, consider that the second bullet holds for $i-1$. 
    Because $|y_3^{(i-1)}(\vec{x}^{(j_1)})-y_3^{(i-1)}(\vec{x}^{(j_2)})|<\eps^{2}$ for any $j_1,j_2\in [3]$, according to the characterization of the temperature on the left/right boundaries (\cref{lem:symmetry-on-converted-coordinates}), we have 
    \begin{itemize}[itemsep=0.2em,topsep=0.5em]
        \item for each $j\in [3]$, $\vec{y}^{(i-1)}(\vec{x}^{(j)})$ is hot or warm; or 
        \item for each $j\in [3]$, $\vec{y}^{(i-1)}(\vec{x}^{(j)})$ is warm or cold. 
    \end{itemize}
    Since we have $C(1-z,0,z)=3=C(0,1-z,z)$ or $C(1-z,0,z)=1, C(0,1-z,z)=2$ for any $z\in [0,1]$ (\cref{lem:base-instance-lr-boundaries}), we have the same $c^{(i-1)}(\vec{x}^{(j)},C(\vec{y}^{(i-1)}(\vec{x}^{(j)})))$ across all $j\in [3]$.
    Next, we discuss the above two cases to finish the proof. 
    \begin{itemize}[itemsep=0.2em,topsep=0.5em]
        \item Consider when each $\vec{y}^{(i-1)}(\vec{x}^{(j)})$ is hot or warm. We have $\hat{C}_{\sf nn}(\vec{y}^{(i-1)}(\vec{x}^{(j)}))=C_{\sf nn}(\vec{y}^{(i-1)}(\vec{x}^{(j)}))$ for each $j\in [3]$. Because we have $C_{\sf nn}(1-z,0,z)=1, C_{\sf nn}(0,1-z,z)=2$, or $C_{\sf nn}(1-z,0,z)=3=C_{\sf nn}(0,1-z,z)$ for each $z\in [0,1]$ (\cref{lem:symmetry-on-converted-coordinates}), we have the same $c^{(i-1)}(\vec{x}^{(j)},\hat{C}_{\sf nn}(\vec{y}^{(i-1)}(\vec{x}^{(j)})))$ across all $j\in [3]$.
        Therefore, the palette $\vec{c}^{(i)}(\vec{x}^{(j)})$ is the same across all $j\in [3]$. 
        Because of \cref{lem:intermediate-projections-are-close-to-each-other} and the Lipschitzness of the coordinate converter on the left/right boundaries \cref{lem:lipschitz-of-the-converter-on-boundaries}, the $i$-th intermediate projections, $\vec{y}^{(i)}(\vec{x}^{(1)}), \vec{y}^{(i)}(\vec{x}^{(2)})$ and $\vec{y}^{(i)}(\vec{x}^{(3)})$, are close to each other. 
        Hence, we prove the first bullet of this lemma for $i$. 
        \item Consider when each $\vec{y}^{(i-1)}(\vec{x}^{(j)})$ is warm or cold. We have $\rel(\vec{y}^{(i-1)}(\vec{x}^{(j)}))\in \{0,1\}$ and all $i$-th intermediate projections are on the left/right boundaries. 
        Note that there is only one color in $\settxt{C(\vec{y}^{(i-1)}(\vec{x}^{(j)}))}_{j\in [3]}$, because otherwise we clearly have $d(\vec{y}^{(i-1)}(\vec{x}^{(j)}),\nn(\vec{y}^{(i-1)}(\vec{x}^{(j)})))<\eps^2$ and all points are hot. 
        According to our earlier discussions, \cref{alg:approx-sperner-kgeq3} ensures in this scenario that $c^{(i)}(\vec{x}^{(j)}, i^*(\vec{y}^{(i)}(\vec{x}^{(j)}))) = c^{(i-1)}(\vec{x}^{(j)}, C(\vec{y}^{(i-1)}(\vec{x}^{(j)})))$.
        Therefore, we have the same $c^{(i)}(\vec{x}^{(j)}, i^*(\vec{y}^{(i)}(\vec{x}^{(j)})))$ across all $j\in [3]$, and thus the second bullet of this lemma holds. \qedhere
    \end{itemize}
\end{proof}

\paragraph{Case 2: the output does not satisfy~\cref{eqn:assumption-for-simple-proof-of-recovery}.}
Next, we complete our second step by showing how to prove \cref{lem:poly-time-recovery} without the property (\cref{eqn:assumption-for-simple-proof-of-recovery}). 
Suppose $\theta^*$ is the minimum threshold $\theta\geq 2$ such that $P_{i+1}^{(k-i)}(\vec{x}^{(j)})\leq 0.9$ for any $\theta\leq i\leq k$ and $j\in[3]$. 
Such threshold always exists because otherwise we have $x^{(1)}_{k+1}, x^{(2)}_{k+1}, x^{(3)}_{k+1}\geq 0.8$. 
This implies $y_3^{(k)}(\vec{x}^{(j)})>0.8$ for any $j\in [3]$, and according to our construction of the base instance (\cref{alg:base-instance}), we have $C^{(k)}(\vec{x}^{(j)})=c^{(k)}_3(\vec{x}^{(j)}) = k+1$ for any $j\in[3]$, violating the assumption that $\vec{x}^{(1)}, \vec{x}^{(2)}, \vec{x}^{(3)}$ form a solution for the \outofk Approximate \Sperner problem. 
When $\theta^*=2$, it is equivalent with the special case satisfying \cref{eqn:assumption-for-simple-proof-of-recovery} and we have proved \cref{lem:poly-time-recovery} for this case.
On the other hand, if $\theta^*>2$, we have $P^{(k-\theta^*+1)}_{\theta^*}(\vec{x}^{(j)})\geq 0.8$ for any $j\in [3]$ because of \cref{lem:lipschitz-of-projection} and that $\normtxt{\vec{x}^{(i)}-\vec{x}^{(j)}}{\infty}\leq 2^{-3kn}$ for any $i,j\in [3]$. 
This means that we have $C(\vec{y}^{(\theta^*-1)}(\vec{x}^{(j)}))=3$.
And because of \cref{lem:trivial-converter}, we have $\rel(\vec{y}^{(\theta^*-1)}(\vec{x}^{(j)}))=1$ and $y_1^{(\theta^*)}(\vec{x}^{(j)})=0$ for any $j\in [3]$. 
Therefore, running \cref{alg:approx-sperner-kgeq3} on instance $C^{(k)}$ for $\vec{x}^{(1)},\vec{x}^{(2)}, \vec{x}^{(3)}$ is equivalent to running \cref{alg:approx-sperner-kgeq3} on instance $C^{(k')}$ with $k'=k-\theta^*+2$ for $\vec{\hat{x}}^{(1)},\vec{\hat{x}}^{(2)},\vec{\hat{x}}^{(3)}\in \Delta^{k'}$ such that
\begin{align*}
    \forall i\in [k'+1], j\in [3], \quad \hat{x}_i^{(j)} = \begin{cases}
        0 & \text{if $i=1$,}\\
        1-\sum_{i'=\theta^*+1}^{k+1} x^{(j)}_{i'} & \text{if $i=2$,}\\
        x^{(j)}_{i+\theta^*-2} & \text{if $i>2$.}
    \end{cases}
\end{align*}
Because this new instance has a smaller number of dimensions and satisfies the condition of our special cases (\cref{eqn:assumption-for-simple-proof-of-recovery}), we complete the proof for \cref{lem:poly-time-recovery}.

\section{PPAD-Hardness of Approximate Symmetric Sperner}
\label{sec:final-proof-for-approx-sperner}
In this section, we give a more elaborate reduction that guarantees the desired symmetry properties required by \cref{def:approx-symmetric-kd-sperner}, and prove the \PPAD-completeness of the Approximate Symmetric $k$D-\Sperner problem (i.e., our main technical result \cref{thm:main}).
We will construct a different chain of instances $\Csym{2},\Csym{3},\dots, \Csym{k}$.
At a high level we make the following two modifications inside the warm-up construction from \cref{sec:warmup-proof-for-approx-sperner}:
(1) We construct two distinct $2${\rm D}-\Sperner instances: $\Csym{2}$ which is used as the basis of our recursive construction,  
and $C$ which is used for lifting each $\Csym{i}$ to $\Csym{i+1}$. (2) The actual constructions of $2${\rm D}-\Sperner instances are more delicate, especially the coordinate converters.
As in \cref{sec:warmup-proof-for-approx-sperner}, we will use $2^{-n}$ and $\eps$ interchangeably in this section.
We restate our main technical we will prove as follows.

\mainthm*

As in \cref{sec:warmup-proof-for-approx-sperner}, our instances will be constructed in continuous space instead of discrete space. 
We will establish the following hardness result:
\begin{theorem}
    \label{thm:continuous-symmetry-main}
    There exists a chain of functions $\{\Csym{k}: \Delta^k \to [k+1]\}_{k\geq 2}$ that satisfy the symmetry defined in \cref{def:approx-symmetric-kd-sperner} and can be computed in $\poly(|\vec{x}|)$ time for each $\vec{x}\in \Delta^k$, where $|\vec{x}|$ is the bit complexity of $\vec{x}$ and satisfies the following property. 
    For any $k\geq 2$, it is \PPAD-hard 
    to find three points $\vec{x}^{(1)}, \vec{x}^{(2)}, \vec{x}^{(3)}\in \Delta^{k}$ such that 
    \begin{itemize}[itemsep=0.2em, topsep=0.5em]
        \item {\bf they are close enough to each other:} for any $i,j\in [3]$, $\normtxt{\vec{x}^{(i)}-\vec{x}^{(j)}}{\infty}\leq 2^{-4kn}$, 
        \item {\bf they induce a trichromatic triangle:} $|\settxt{C^{(k)}(\vec{x}^{(i)}):i\in [3]}|=3$. 
    \end{itemize}
    In particular, if $\Csym{k}$ is black-box, it requires a query complexity of $2^{\Omega(n)}/\poly(n,k)$ to find a solution satisfying the above two properties. 
\end{theorem}

\begin{remark}
    The \PPAD-hardness and query complexity lower bound for the continuous \outofk Approximate Symmetric \Sperner instances can be easily generalized for the discrete ones that follow \cref{def:approx-symmetric-kd-sperner}.
    The reduction from the continuous cases to discrete cases is as follows: for any $\vec{x}\in \Delta_{4kn}^{k}$, we define 
    \begin{align*}
        C^{(k)}_{\sf sym,dis}(\vec{x}) := \Csym{k}(\vec{x})~,
    \end{align*}
    where any trichromatic tuple $\vec{x}^{(1)}, \vec{x}^{(2)}, \vec{x}^{(3)}$ in $C^{(k)}_{\sf sym,dis}$ is 
    simply a trichromatic tuple in $\Csym{k}$. 
\end{remark}

Our plan in this section is to prove \cref{thm:continuous-symmetry-main} and its roadway is as follows. 
We will first introduce our new coordinate converter in \cref{subsec:more-elaborate-coordinate-converter}. 
Then, we will (slightly) adjust our previous hard two-dimensional instances $C^{(2)}$ so that they satisfy the symmetry constraints in \cref{subsec:symmetric-2d}.
Combining this new coordinate converter and these different hard $2${\rm D}-\Sperner instances with previous recursive construction framework for three or higher dimensions, we have finished all the constructions.
We will show the new three or higher dimensional instances (under the new combination) are symmetric and still \PPAD-hard in \cref{subsec:symmetric-kd}.
Finally, we will establish the query complexity of the instances in \cref{subsec:query-complexity}.
In the entire section, we assume that the base instance $C$ is fixed and any function in consideration will have oracle access to $C$.

\subsection{A more elaborate coordinate converter}
\label{subsec:more-elaborate-coordinate-converter}
In this subsection, we provide a more elaborate coordinate converter.  
Recall that in our hard instances for the \outofk Approximate (Unconstrained) \Sperner problem, 
we define the coordinate convertersimply via each point's $\ell_{\infty}$-nearest neighbor with a different color. 
However, using this coordinate converter, we cannot guarantee the desired symmetry (\cref{def:approx-symmetric-kd-sperner}) even in the $3${\rm D}-instances.
Consider the following two facets: $\settxt{(x,y,z,0):x+y+z=1}$ and $\settxt{(x,y,0,z):x+y+z=1}$. 
In our construction in \cref{sec:warmup-proof-for-approx-sperner}, we let $C^{(3)}(x,y,z,0)$ be  the $C(x,y,z)$-th color in the palette $(1,2,3)$. 
In contrast, for $C^{(3)}(x,y,0,z)$, we use the $C((1-z)\cdot (1-\tilde{y}),(1-z)\cdot \tilde{y},z)$-th color in the palette $(1,2,4)$, where $\tilde{y}=\rel(\frac{x}{1-z},\frac{y}{1-z},0)$. 
To ensure the symmetry between $C^{(3)}(x,y,z,0)$ and $C^{(3)}(x,y,0,z)$, 
one natural way is to guarantee that we use the same point in the 2{\rm D}-\Sperner instance for $(x,y,z,0)$ and $(x,y,0,z)$, i.e., 
\[y=(1-z)\cdot \tilde{y}=(1-z)\cdot \rel\left(1-\frac{y}{1-z},\frac{y}{1-z},0\right)~,\]
which can hold if the coordinate converter is the identity on the bottom $1$-simplex, i.e., $\rel(1-y,y,0)=y$. 

Therefore, our primary motivation for using a more elaborate coordinate converteris to ensure that it gives an identity mapping on the bottom $1$-simplex. We also have to carefully interpolate between this boundary constraint and maintaining the desiderata of the existing construction on the interior. 

\paragraph{Interpolation.} To interpolate the coordinate converter, we define the following {\em shrinking factor} 
for the distance, which depends only on the third coordinate of each point and will be applied on the second coordinate when computing the distance.
The formula \cref{eqn:shrinking-factor} of this shrinking factor can be viewed as an interpolation between $z=0$ and $z=0.05$ using exponential functions so that we will shrink distances of points with $z=0$ by a small factor of $2\eps^2$. We do not shrink the distances (i.e., $\alpha(\cdot)=1$) for those with $z\geq 0.05$. 
See \cref{lem:shrinking-factor-basics} for the details.
\begin{align}
    \label{eqn:shrinking-factor}
    \alpha(z) = \exp\left( -20(2n-1)\ln 2 \cdot (0.05-z)_+ \right).
\end{align}
The intuition of our new definition is that with the distance shrunk, the hot regions expand.
Therefore, the above issue of asymmetry between the bottom $2$-simplex and the other three $2$-simplices in the $3${\rm D}-instances could be resolved. 
The key properties of this new coordinate convertertowards proving the symmetry are concluded in the following lemma.
\begin{lemma}
    \label{lem:shrinking-factor-basics}
    The shrinking factor $\alpha$ satisfies:
    \begin{enumerate}[topsep=0.5em, itemsep=0.1em]
        \item $\alpha(0)=2^{-2n+1}=2\eps^{2}$.
        \item For any $z\geq 0.05$, $\alpha(z)=1$.
        \item Lipschitz: for any $z,z'\in [0,1]$ such that $|z-z'|\leq \eps$, $\abstxt{\alpha(z)-\alpha(z')}\leq O(n)\cdot |z-z'|$. 
    \end{enumerate}
\end{lemma}
\begin{proof}
    The first two bullets can be easily obtained via simple calculations. 
    Further, because of the second bullet, it remains to show the third bullet for $0\leq z\leq z' \leq 0.05$ such that $|z-z'|\leq 2^{-\Omega(n)}$:
    \begin{align*}
        \abs{\alpha(z)-\alpha(z')} &= \abs{
            \exp\left( -20(2n-1)\ln 2 \cdot (0.05-z) \right)
            -
            \exp\left( -20(2n-1)\ln 2 \cdot (0.05-z') \right)
        }
        \\
        &=
        \abs{
            \exp\left( -20(2n-1)\ln 2 \cdot (0.05-z) \right) 
            \cdot
            \left(
                1-\exp\left( 20(2n-1)\ln 2 \cdot (z'-z) \right)
            \right)
        }
        \\
        &\leq
        \exp\left( 20(2n-1)\ln 2 \cdot (z'-z) \right)-1
        \\
        &\leq 
        20(2n-1)\ln 2\cdot (z'-z) + O(n^2(z'-z)^2) 
        \\
        &= O(n)\cdot |z-z'|~.
        \qedhere
    \end{align*}
\end{proof}

Now we're ready to define the following (asymmetric) {\em shrunk quasimetric} 
based on the shrinking factor, 
\begin{align}
    \label{eqn:shrinking-distance}
    d^{\alpha}(\vec{x}, \vec{x'}) = \max\set{\alpha(x_3)\cdot \abstxt{x_2-x_2'},\ \abstxt{x_3-x_3'}}
\end{align}

\paragraph{The coordinate converter.} The definition of the coordinate converter follows the previous section's steps. 
More specifically, for each point $\vec{x}\in \Delta^2$, we find the infimum shrunk distance from $\vec{x}$ to all points with a fixed color $c\neq C(\vec{x})$:  
\begin{align*}
    d_{\min}^{\alpha}(\vec{x}, c) = \inf_{\vec{y}\in \Delta^2: C(\vec{y})=c} 
    d^{\alpha}(\vec{x},\vec{y})
    ~.
\end{align*}
Then, our definitions of the neighboring color and the nearest neighbor are modified accordingly as follows: 
\begin{align}
    \label{eqn:cnn-and-nn-for-symmetric}
    \begin{split}
    C_{\sf nn}^{\alpha}(\vec{x}) &= \argmin_{c\in [3]:c\neq C(\vec{x})} \left(d^{\alpha}_{\min}(\vec{x}, c),\ c\right)~,
    \\
    \nna(\vec{x}) &= \arginf_{\vec{y}\in \Delta^2: C(\vec{y})=C_{\sf nn}(\vec{x})} d^{\alpha}(\vec{x},\vec{y})~.
    \end{split}
\end{align}
We will continue to use the notions of {\em hot/warm/cold regions}, where the definitions are nearly the same as in \cref{sec:warmup-proof-for-approx-sperner} (replacing $d$ by $d^{\alpha}$). 
\begin{itemize}[itemsep=0.2em,topsep=0.5em]
    \item The {\em hot} region consists of points having a shrunk distance to the nearest neighbor strictly less than $\eps^2$, i.e., points $\vec{x}$ with $d^\alpha(\vec{x}, \nna(\vec{x}))<\eps^2$; we say that points in this region are {\em hot}.
    \item The {\em warm} region consists of points having a shrunk distance to the nearest neighbor between $\eps^2$ (inclusive) and $2\eps^2$ (exclusive), i.e., points $\vec{x}$ with $d^{\alpha}(\vec{x}, \nna(\vec{x}))\in [\eps^2,2\eps^2)$; we say that points in this region are {\em warm}.
    \item The {\em cold} region consists of points having a shrunk distance to the nearest neighbor no less than $2\eps^2$, i.e., points $\vec{x}$ with $d^\alpha(\vec{x}, \nna(\vec{x}))\geq 2\eps^2$; we say that points in this region are {\em cold}.
\end{itemize}
Finally, the new coordinate converter is in the form as the previous one \cref{eqn:coordinate-converter} except that we use everything with a shrinking factor.
\begin{align}
    \label{eqn:new-coordinate-converter}
    \rela(\vec{x}) = \begin{cases}
        \hfill \left(0.5 - 0.5\eps^{-2} \cdot d^\alpha(\vec{x}, \nna(\vec{x}))
        \right)_+ \hfill & \text{if $\vec{x}$ is {\em hot\,/\,warm} and $C^\alpha_{\sf nn}(\vec{x})>C(\vec{x})$,}\\
        \left(0.5 + 0.5\eps^{-2} \cdot d^\alpha(\vec{x}, \nna(\vec{x}))\right)_- & \text{if $\vec{x}$ is {\em hot\,/\,warm} and $C^\alpha_{\sf nn}(\vec{x})<C(\vec{x})$,}\\
        \hfill 0 \hfill & \text{if $\vec{x}$ is {\em cold} and $C(\vec{x})=1$,}\\
        \hfill 1 \hfill & \text{if $\vec{x}$ is {\em cold} and $C(\vec{x})\in \{2,3\}$.}
    \end{cases}
\end{align}
Because we have $\alpha(x_3)=1$ for any $\vec{x}$ such that $x_3\geq 0.05$, it is easy to observe that the coordinate converter, along with every intermediate function for its definition, does not change for those points.
This observation will help us continue to use a huge fraction of the properties we have established in \cref{sec:warmup-proof-for-approx-sperner}.
\begin{observation}
    \label{lem:same-def-for-converted-coordinate}
    Consider any $\vec{x}\in \Delta^2$ such that $x_3\geq 0.05$.
    We have 
    \begin{itemize}[itemsep=0em, topsep=0.3em]
        \item $d^{\alpha}(\vec{x}, \vec{x'})=d(\vec{x}, \vec{x'})$ for any $\vec{x'}\in \Delta^2$;
        \item $C^{\alpha}_{\sf nn}(\vec{x}) = C_{\sf nn}(\vec{x})$;
        \item $\nna(\vec{x})=\nn(\vec{x})$;
        \item $\rela(\vec{x})=\rel(\vec{x})$.
    \end{itemize}
\end{observation}

\subsubsection{Key properties}
Next, we establish the key properties of this new coordinate converter on this family of instances.
The only new property we will establish is that this new function is an identity on the base $1$-simplex.
In addition, we will prove the analogs of all the properties we use in \cref{sec:warmup-proof-for-approx-sperner} (\cref{lem:poly-time-oracle-and-converter,,lem:symmetry-on-converted-coordinates,,lem:lipschitz-of-the-converter,,lem:trivial-converter}) for this new coordinate converter. 

\paragraph{Property I: identity on the base}
With the new definition of the coordinate converter, we can prove that the coordinate converter is simply an identity function for points on the base 1-simplex. 
Moreover, we characterize the neighbor coloring on the base 1-simplex.
\begin{lemma}
    \label{lem:new-converted-coordinates-on-base-1-simplex}
    Points on the segment $(1-x_2,x_2,0)$ (for $x_2\in (0,1)$) are hot. 
    More specifically, for any $x_2\in [0,1]$, we have $\rela(1-x_2,x_2,0)=x_2$, and
    \begin{align*}
        C_{\sf nn}^{\alpha}(1-x_2,x_2,0)=\begin{cases}
            1 & \text{if $x_2> 0.5$,}\\
            2 & \text{if $x_2\leq 0.5$.}
        \end{cases}
    \end{align*}
\end{lemma}
\begin{proof}
    Let $\vec{x}=(1-x_2,x_2,0)$.
    According to \cref{lem:shrinking-factor-basics,,eqn:shrinking-distance}, we have 
    \begin{align*}
        d^{\alpha}((1-x_2,x_2,0), \vec{x'}) = \max\set{2\eps^2 \cdot |x_2-x'_2|, |x'_3|}. 
    \end{align*}
    Note that in our base instance $C$, we have 
    \begin{align*}
        \forall \vec{x'}\in \Delta^2, \quad
        \begin{cases}
        C(\vec{x'}) = 1 & \quad \text{if } x'_3\leq 0.1, x'_2 \leq 0.5~,\\
        C(\vec{x'}) = 2 & \quad \text{if } x'_3\leq 0.1, x'_2 > 0.5~.
        \end{cases}
    \end{align*}
    
    If we have $x_2\leq 0.5$ here, any point $\vec{x'}$ with a different color with $C(\vec{x})=1$ has either $x'_2>0.5$ or $x'_3>0.1$. 
    For the second condition ($x'_3>0.1$), we have $d^{\alpha}(\vec{x}, \vec{x'})>0.1>2\eps^2$.
    On the other hand, it is easy to observe that $(0.5, 0.5, 0)$ has the minimum distance $d^{\alpha}(\vec{x},\vec{x'})$ among those points with $x'_3\leq 0.1, x'_2\geq 0.5$. 
    Therefore, $\nna(\vec{x})=(0.5,0.5,0)$ and $C^{\alpha}_{\sf nn}(\vec{x})=2$.
    In particular, according to \cref{lem:shrinking-factor-basics}, we have $d^{\alpha}(\vec{x}, (0.5, 0.5, 0))=2\eps^2 \cdot(0.5-x_2)$. 
    According to \cref{eqn:new-coordinate-converter}, because $C^{\alpha}_{\sf nn}(\vec{x})>C(\vec{x})$,
    \begin{align*}
        \rela(\vec{x}) &= \left(0.5-0.5\eps^{-2} \cdot d^{\alpha}\left(\vec{x}, \nna(\vec{x})\right)\right)_+
        \\
        &= \left(0.5-0.5\eps^{-2} \cdot d^{\alpha}\left(\vec{x}, (0.5,0.5,0)\right)\right)_+
        \\
        &= \left(0.5-(0.5-x_2)\right) = x_2~.
    \end{align*}
    Hence, when $x_2\in (0,0.5]$, $\vec{x}$ is in the hot region.

    Similarly, if we have $x_2> 0.5$ here, we have $C(\vec{x})=2$, $\nna(\vec{x})=(0.5,0.5,0)$ and $C_{\sf nn}^{\alpha}(\vec{x})=1$. 
    According to \cref{lem:shrinking-factor-basics}, we have $d^{\alpha}(\vec{x}, (0.5, 0.5, 0))=2\eps^2 \cdot(x_2-0.5)$. 
    According to \cref{eqn:new-coordinate-converter}, because $C_{\sf nn}^\alpha(\vec{x})<C(\vec{x})$,
    \begin{align*}
        \rela(\vec{x}) &= \min\set{0.5+0.5\eps^{-2}\cdot d^{\alpha}(\vec{x}, \nna(\vec{x})),\ 1}
        \\
        &= \min\set{0.5+0.5\eps^{-2}\cdot d^{\alpha}(\vec{x}, (0.5,0.5,0)),\ 1}
        \\
        &= \min\set{0.5+(x_2-0.5),\ 1} = x_2~.
    \end{align*}
    Hence, when $x_2\in (0.5,1)$, $\vec{x}$ is in the hot region.
    
    In conclusion, for any $x_2\in (0,1)$, the point $\vec{x}$ is in the hot region.
\end{proof}

\paragraph{Property II: Polynomial time computation}
As in \cref{sec:warmup-proof-for-approx-sperner}, we show that we can always output the true converted coordinate $\rela(\vec{x})$ and, whenever $\vec{x}$ is hot or warm (i.e., $d^\alpha(\vec{x}, \nna(\vec{x}))<2\eps^2$), we can also output the neighboring color $C^\alpha_{\sf nn}(\vec{x})$.
\begin{lemma}
    \label{lem:poly-time-oracle-and-converter-for-symmetry}
    Given oracle access to $C$, there is a polynomial-time algorithm that takes any $\vec{x}\in \Delta^2$ as input and that outputs $\rela(\vec{x})$.
    Furthermore, if $d^\alpha(\vec{x}, \nna(\vec{x}))<2\eps^2$, the algorithm can also compute $C^\alpha_{\sf nn}(\vec{x})$ in polynomial time.
\end{lemma}
\begin{proof}
    Note that the converted coordinates are the same for any $\vec{x}$ with $x_3\geq 0.05$ (\cref{lem:same-def-for-converted-coordinate}), we only need to show for $x_3\leq 0.05$ because we have established \cref{lem:poly-time-oracle-and-converter}.
    W.l.o.g., we suppose that the input $\vec{x}\in \Delta^2$ satisfies $x_3\leq 0.05$. 
    
    As in the proof of \cref{lem:poly-time-oracle-and-converter}, it suffices to show how to compute $d^\alpha(\vec{x}, \nna(\vec{x}))$ and $C^\alpha_{\sf nn}(\vec{x})$ when $\vec{x}$ is hot or warm.
    Note that in our base instance $C$, we have 
    \begin{align*}
        \forall \vec{x'}\in \Delta^2, \quad
        \begin{cases}
        C(\vec{x'}) = 1 & \quad \text{if } x'_3\leq 0.1, x'_2 \leq 0.5~,\\
        C(\vec{x'}) = 2 & \quad \text{if } x'_3\leq 0.1, x'_2 > 0.5~.
        \end{cases}
    \end{align*}
    Also, note that our shrinking distance can be lower bounded by the difference on the third coordinate:
    \begin{align*}
        d^{\alpha}(\vec{x}, \vec{x'}) = \max\set{\alpha(x_3)\cdot \abstxt{x_2-x_2'},\ \abstxt{x_3-x_3'}} \geq |x_3-x'_3|~.
    \end{align*}
    We have $d^{\alpha}(\vec{x}, \nna(\vec{x})) < 2\eps^2$ only when $nn_3^{\alpha}(\vec{x})\leq 0.05+2\eps^2<0.1$. 
    Among points $\vec{x'}$ such that $x'_3<0.1$, the nearest \footnote{Here, it may be possible that the nearest differently colored point of $\vec{x}$ is given by the limit of a sequence of infinite many points that have different colors with $\vec{x}$.}
    differently colored point of $\vec{x}$ is clearly $(0.5-x_3,0.5,x_3)$. 
    Therefore, when $\vec{x}$ is hot or warm, we have $\nna(\vec{x})=(0.5-x_3,0.5,x_3)$.
    Because of this simple characterization of $\nna(\cdot)$ for hot and warm points, we can compute $d^{\alpha}(\vec{x}, (0.5-x_3,0.5,x_3))$ to decide if $\vec{x}$ is hot or warm and then compute the value of $d^\alpha(\vec{x}, \nna(\vec{x}))$. 
    In addition, for hot and warm points, $C^\alpha_{\sf nn}(\vec{x})$ equals the color in $\{1,2\}$ that does not equal $C(\vec{x})$. 
    $C^{\alpha}_{\sf nn}(\vec{x})$ is also easy to compute when $\vec{x}$ is hot or warm. 
\end{proof}

\paragraph{Property III: Symmetry} Recall \cref{lem:symmetry-on-converted-coordinates}, where we give a complete characterization of the temperature of each point on the line segments $(1-x_3,0,x_3)$ and $(0,1-x_3,x_3)$.
In the complete characterization, the set of hot and warm points consists of those with $x_3\in 0.1\pm 2\eps^2$. 
Here, we generalize it to a (slightly) incomplete characterization under the new coordinate converter, where $\vec{x}$ has the same characterization as in \cref{sec:warmup-proof-for-approx-sperner} if $x_3\in 0.1\pm 2\eps^2$, and $\vec{x}$ is not hot otherwise.
\begin{lemma}
    \label{lem:symmetry-on-converted-coordinates-for-symmetric-sperner}
    For points on the line segments $(1-x_3, 0,x_3)$ and $(0, 1-x_3,x_3)$, we have 
    \begin{itemize}[itemsep=0.2em, topsep=0.5em]
        \item if $x_3\notin 0.1\pm 2\eps^2$, then $(1-x_3,0,x_3)$ and $(0,1-x_3,x_3)$ are either cold or warm;
        \item otherwise, if $x_3\in 0.1\pm 2\eps^2$, then
        \begin{align}
            \label{eqn:stronger-symmetry-on-converted-coordinates-for-symmetry}
            \rela(1-x_3, 0,x_3) = \left(1/2 + \frac{x_3-0.1}{2 \eps^2}\right)_{[0,1]} =\rela(0, 1-x_3, x_3)~,
        \end{align}
        and the neighboring color is characterized as follows
        \begin{align*}
            C^\alpha_{\sf nn}(1-x_3,0,x_3) = \begin{cases}
                3 & \text{if $x_3\leq 0.1$,}\\
                1 & \text{if $x_3>0.1$,}
            \end{cases}
            \quad
            \text{and}
            \quad
            C^\alpha_{\sf nn}(0,1-x_3,x_3) = \begin{cases}
                3 & \text{if $x_3\leq 0.1$,}\\
                2 & \text{if $x_3>0.1$.}
            \end{cases}
        \end{align*}
    \end{itemize}
    In particular, we have $\rela(1-x_3,0,x_3) \simrel \rela(0,1-x_3,x_3) \simrel (0.5+0.5\eps^{-2}\cdot (x_3-0.1))_{[0,1]}$ for any $x_3\in [0,1]$.
\end{lemma}

\begin{proof}
    Because of \cref{lem:same-def-for-converted-coordinate}, the characterization is the same for any $z\geq 0.05$, we only need to show for $z\leq 0.05$ because we have established \cref{lem:symmetry-on-converted-coordinates}.
    That is, we want to show that for any $x_3\leq 0.05$, 
    \[
        \rela(1-x_3, 0, x_3),~\rela(0, 1-x_3, x_3)\in \{0,1\}~.
    \]
    According to our new coordinate converter \cref{eqn:new-coordinate-converter}, it suffices to show that for any $x_3\leq 0.05$ and any $\vec{x}\in \{(1-x_3,0,x_3), (0,1-x_3,x_3)\}$, 
    \begin{align}
        d^{\alpha}(\vec{x},\ \nna(\vec{x})) &\geq \eps^{2} ~.
        \label{eqn:far-from-z}
    \end{align}
    Note that in our base instance $C$, we have 
    \begin{align*}
        \forall \vec{x'}\in \Delta^2, \quad
        \begin{cases}
        C(\vec{x'}) = 1 & \quad \text{if } x'_3\leq 0.1, x'_2 \leq 0.5~,\\
        C(\vec{x'}) = 2 & \quad \text{if } x'_3\leq 0.1, x'_2 > 0.5~.
        \end{cases}
    \end{align*}
    
    First, we consider the case when $\vec{x}=(1-x_3,0,x_3)$ for $x_3\leq 0.05$.
    We have $C(\vec{x})=1$, and for any $\vec{x'}\in \Delta^2$, $C(\vec{x})\neq C(\vec{x'})$ only if $x'_2\geq 0.5$ or $x'_3\geq 0.1$.
    Since $d^{\alpha}(\vec{x},\vec{x'})\geq |x_3-x'_3|$, if $x'_3\geq 0.1$, we have $d^{\alpha}(\vec{x},\vec{x'})\geq 0.05>\eps^2$.
    On the other hand, if $x'_2\geq 0.5$, because of \cref{lem:shrinking-factor-basics}, we have 
    \begin{align*}
        d^{\alpha}(\vec{x}, \vec{x'}) \geq \alpha(x_3)\cdot |0-x'_2| \geq \alpha(0) \cdot 0.5 = \eps^2.
    \end{align*}
    Therefore, we have $d^{\alpha}(\vec{x}, \vec{x'})$ for any $C(\vec{x'})\neq C(\vec{x})$, and thus we have proved \cref{eqn:far-from-z} for this case.

    Second, we consider the case when $\vec{x}=(0,1-x_3,x_3)$ for $x_3\leq 0.05$.
    We have $C(\vec{x})=2$, and for any $\vec{x'}\in \Delta^2$, $C(\vec{x})\neq C(\vec{x'})$ only $x'_2\leq 0.5$ or $x'_3\geq 0.1$.
    Since $d^{\alpha}(\vec{x},\vec{x'})\geq |x_3-x'_3|$, if $x'_3\geq 0.1$, we have $d^{\alpha}(\vec{x},\vec{x'})\geq 0.05>\eps^2$.
    On the other hand, if $x'_2\leq 0.5$, because of \cref{lem:shrinking-factor-basics}, we have 
    \begin{align*}
        d^{\alpha}(\vec{x}, \vec{x'}) &= \alpha(x_3) \cdot |(1-x_3)-x'_2| 
        \\
        &\geq \alpha(x_3)\cdot (0.5-x_3)
        \\
        &= \exp\left( -20(2n-1)\ln 2 \cdot (0.05-x_3) \right) \cdot (0.5-x_3)
        \\
        &=
        2^{-2n+1} \cdot \exp\left( 20(2n-1)\ln 2 \cdot x_3 \right) \cdot (0.5-x_3)
        \\
        &\geq 
        2^{-2n+1} \cdot (1+20(2n-1)\ln 2\cdot x_3) \cdot (0.5-x_3)\
        \\
        &\geq 
        2^{-2n+1} \cdot 0.5 
        =2^{-2n}
        =\eps^2
        \tag{$x_3\leq 0.05$}
    \end{align*}
    Hence, we complete the proof.
\end{proof}

\paragraph{Property IV: Lipshitzness} We prove that the new coordinate converter is Lipschitz in the same sense as \cref{lem:lipschitz-of-the-converter}, but with a slightly larger Lipschitz factor.
\begin{lemma}
    \label{lem:lipschitz-of-the-converter-for-symmetry}
    For any point $\vec{x}, \vec{x'}$, at least one of the following properties is satisfied:
    \begin{itemize}[itemsep=0em]
        \item {\bf one is in a trichromatic region:} there are three different colors in $\mathcal{N}(\vec{x})$ or $\mathcal{N}(\vec{x'})$;
        \item {\bf Lipschitz in the hot\&warm regions:} $|\rela(\vec{x})-\rela(\vec{x'})|\leq \eps^{-3}\cdot \normtxt{\vec{x}-\vec{x'}}{\infty}$; 
        \item {\bf both are in the warm\&cold regions:} $\rela(\vec{x}),\rela(\vec{x'})\in \{0,1\}$;
    \end{itemize}
\end{lemma}
\begin{proof}
    Because of \cref{lem:same-def-for-converted-coordinate} and \cref{lem:lipschitz-of-the-converter}, we have established this lemma if $x_3,x'_3\geq 0.05$.
    If $|x'_3-x_3|\geq \eps^3$, the second property trivially holds. 
    Therefore, it suffices to establish this lemma for the case where $x_3,x'_3\leq 0.06$ to finish its proof.

    We will use the following fact about the shrunk distance to each point's nearest neighbor and the new coordinate converter, where the proof is based on simple calculations and deferred to \cref{proof:new-coordinate-converter-at-the-bottom-part}.
    \begin{restatable}[]{fact}{repiii}
        \label{fact:new-coordinate-converter-at-the-bottom-part}
        Consider an auxillary function $g(\vec{y}) = \alpha(y_3) \cdot (y_2-0.5)$.
        For any $\vec{y}\in \Delta^2$ such that $y_3\leq 0.06$, warm and hot points satisfies $g(\vec{y})\in \pm 2\eps^2$. Furthermore, 
        \begin{align*}
            d^\alpha(\vec{y}, \nna(\vec{y})) = |g(\vec{y})|~, \quad \text{ and } \quad 
            \rela(\vec{y}) = \left(0.5 + 0.5\eps^{-2} \cdot g(\vec{y}) \right)_{[0,1]}~.
        \end{align*}
    \end{restatable}

    Similarly as in the proof of \cref{lem:lipschitz-of-the-converter}, we only need to show that the Lipschitz property, $|\rela(\vec{x})-\rela(\vec{x'})|\leq \eps^{-3}\cdot \normtxt{\vec{x}-\vec{x'}}{\infty}$, under the following two assumptions\footnote{We don't need the earlier assumption that $\mathcal{N}(\vec{x})$ and $\mathcal{N}(\vec{x'})$ are both at most bichromatic here. This is because \cref{lem:all-solutions-of-C-are-in-the-core} ensures that $\vec{x}, \vec{x'}$ here are not in a trichromatic region.}: 
    (1) $\normtxt{\vec{x}-\vec{x'}}{\infty} \leq \eps^3$; and 
    (2) $\vec{x}$ is hot, i.e., $\rela(\vec{x})\in (0,1)$.

    Under the first assumption, we have 
    \begin{align*}
        |g(\vec{x}) - g(\vec{x'})| &= \abs{\alpha(x_3)\cdot (x_2-0.5) - \alpha(x'_3)\cdot (x'_2-0.5)}\\
        &\leq \alpha(x_3) \cdot \abs{x_2-x'_2} + \abs{\alpha(x_3)-\alpha(x'_3)} \cdot |x'_2-0.5| \\
        &\leq \abs{x_2-x'_2} + O(n)\cdot \abs{x_3-x'_3} 
        \tag{$\alpha$ is increasing and \cref{lem:shrinking-factor-basics}}
        \\
        &\leq O(n)\cdot \norm{\vec{x}-\vec{x'}}\infty~.\tag{\cref{lem:shrinking-factor-basics}}
    \end{align*}

    First, we show that $\vec{x'}$ is hot or warm. According to \cref{fact:new-coordinate-converter-at-the-bottom-part}, because $\vec{x}$ is hot and $0.5+0.5\eps^{-2}\cdot g(\vec{x})\in (0,1)$ only when $g(\vec{x})\in \eps^2$, $g(\vec{x'})\leq O(n)\normtxt{\vec{x}-\vec{x'}}\infty + g(\vec{x}) < O(n)\cdot \eps^{3} +\eps^2 < 2\eps^2$.
    Then, we establish the Lipschitz property for $\vec{x}$ and $\vec{x'}$ under the two assumptions. 
    Because of \cref{fact:new-coordinate-converter-at-the-bottom-part} and the fact that $|(a)_{[0,1]}-(b)_{[0,1]}| \leq |a-b|$, we have 
    \begin{align*}
        \abs{\rela(\vec{x}) - \rela(\vec{x'})} &\leq
        \abs{ 0.5\eps^{-2} \cdot g(\vec{x}) -   0.5\eps^{-2} \cdot g(\vec{x'}) }
        \\
        &=
        0.5\eps^{-2} \cdot \abs{g(\vec{x})-g(\vec{x'})}
        \leq \eps^{-3} \cdot \norm{\vec{x}-\vec{x'}}\infty~.
        \qedhere
    \end{align*}
\end{proof}

Similar as \cref{lem:lipschitz-of-the-converter-on-boundaries}, we can prove Lipschitzness for the converted coordinates on the left/right boundaries by \cref{lem:symmetry-on-converted-coordinates-for-symmetric-sperner}.
\begin{lemma}
    \label{lem:lipschitz-of-the-converter-on-boundaries-for-symmetry}
    For any $z,z'\in [0,1]$ and any pair of points $\vec{x}, \vec{x'}$ such that $\vec{x}\in \{(0,1-z,z),\, (1-z,0,z)\}$ and $\vec{x'}\in \{(0,1-z',z'),\, (1-z',0,z')\}$, at least one of the following properties is satisfied:
    \begin{itemize}[itemsep=0em]
        \item {\bf Lipschitz in the hot\&warm regions:} $|\rela(\vec{x})-\rela(\vec{x'})|\leq \eps^{-2}\cdot |z-z'|$; 
        \item {\bf both are in the warm\&cold regions:} $\rela(\vec{x}),\rela(\vec{x'})\in \{0,1\}$. 
    \end{itemize}
\end{lemma}

\paragraph{Property V: Cold on the top} 
Finally, we prove that the converted coordinates are trivial for the points with reasonably high values on the third dimension ($x_3>0.5$). 
Because $\rel(\vec{x})=\rela(\vec{x})$ for those $x_3>0.5$, this fact holds trivially for the new coordinate converterbecause of our earlier \cref{lem:trivial-converter}.
\begin{fact}
    \label{lem:trivial-converter-for-symmetry}
    For any $\vec{x}\in \Delta^{2}$ such that $x_3\geq 0.5$, we have $\rela(\vec{x})=1$. 
\end{fact}

\subsection{Hard instances with two dimensions}
\label{subsec:symmetric-2d}
In this subsection, we present our symmetric construction for $k=2$. 
Our construction converts a point of the $2$-simplex to points on another $2$-simplex and uses the color of the converted coordinates to guarantee symmetry. 
Given a point $\vec{x}$ in the $2$-simplex, we first project it onto the base $1$-simplex. 
Suppose the projection gives $(1-y_0,y_0)$. 
We define a continuous and piecewise linear mapping from the $1$-simplex to itself: 
\begin{align*}
    \tilde{y}_0 =  \begin{cases}
        \hfill 0 \hfill & \text{if } y_0 \leq 0.1 - \eps^2, \\
        0.5 + 0.5\eps^{-2} \cdot (y_0-0.1) & \text{if } y_0 \in 0.1 \pm  \eps^2, \\
        \hfill 1 \hfill & \text{if } y_0 \geq 0.1 + \eps^2.
    \end{cases}
\end{align*}
Finally, we use $((1-x_3)\cdot (1-\tilde{y}_0), (1-x_3)\cdot \tilde{y}_0, x_3)$ as the intermediate projection $\vec{y}^{(2)}$ of $\vec{x}$ and define the color of $\Csym{2}(\vec{x})$ as the color of $\vec{y}^{(2)}$ in the base instance $C$.
The first \ref{line:first-intermediate-projection} lines of \cref{alg:approx-symmetric-sperner-kgeq3} concludes our symmetric construction for $k=2$.

\begin{algorithm2e}[t]
    \caption{\outofk Approximate Symmetric \Sperner Instance $\Csym{k}(\vec{x})$}
    \label{alg:approx-symmetric-sperner-kgeq3}

    \DontPrintSemicolon
    \SetKwInOut{Input}{Input}
    \SetKwInOut{Output}{Output}

    \Input{~vector $\vec{x}\in \Delta^k$}

    \Output{~color $c\in [k+1]$}

    $\eps\gets 2^{-n}$

        $y_0 \gets P_2^{(k-1)}(\vec{x})$ 
    
    $\tilde{y}_0 \gets \left(0.5 + 0.5\eps^{-2} \cdot (y_0-0.1)\right)_{[0,1]}$ \tcp*{convert $\vec{x}$'s projection to the base $1$-simplex}
    \label{line:first-converted-coordinate}

    $z \gets P_3^{(k-2)}(\vec{x})$
    
    $\vec{y}^{(2)}\gets ((1-z)\cdot (1-\tilde{y}_0), (1-z)\cdot \tilde{y}_0, z)$ 
    \label{line:first-intermediate-projection}
    
    \tcp*{initiate $\vec{y}$ by the convertion and $\vec{x}$'s projection to the base $2$-simplex}

    $\vec{c}^{(2)}\gets (1,2,3)$ \tcp*{initiate the set of 3 colors}

    \For{$i\in \{3,\dots, k\}$}{
        $y_3^{(i)} \gets P_{i+1}^{(k-i)}(\vec{x})$
        \label{line:symmetric-y^{(i)}_3}
        \tcp*{find the new $z$-coordinate from the projection of $\vec{x}$}

        $y_2^{(i)} \gets (1-y_3^{(i)}) \cdot \rela ( \vec{y}^{(i-1)} )$
        \label{line:symmetric-y^{(i)}_2}

        $y_1^{(i)} \gets (1-y_3^{(i)}) \cdot (1-\rela ( \vec{y}^{(i-1)} ))$
        \label{line:symmetric-y^{(i)}_1}
        \tcp*{convert $\vec{y}$ to a point on the $2$-simplex}

        $j_1\gets \min \left\{C(\vec{y}^{(i-1)}),\ \hat{C}^\alpha_{\sf nn}(\vec{y}^{(i-1)})\right\}$

        $j_2\gets \max \left\{C(\vec{y}^{(i-1)}),\ \hat{C}^\alpha_{\sf nn}(\vec{y}^{(i-1)})\right\}$

        $\vec{c}^{(i)} \gets (c_{j_1}^{(i-1)}, c_{j_2}^{(i-1)}, i+1)$
        \label{line:obtain-new-set-of-colors-for-symmetry}
        \tcp*{obtain the new set of 3 colors}
    }

    $j \gets C \left(\vec{y}^{(k)}\right)$
            
    \Return{$c_j^{(k)}$}
\end{algorithm2e}

\paragraph{Symmetry.} For the 2-dimensional instances, the symmetry results from the simple calculation of the colors on the three $1$-simplices on the boundary.

\begin{lemma}
\label{lem:2d-symmetry}
\cref{alg:approx-symmetric-sperner-kgeq3} gives a valid \outofk Approximate Symmetric \Sperner instance for $k=2$. 
\end{lemma}
\begin{proof}
    We will give a complete characterization of the coloring on the three boundaries: for any $x\in [0,1]$,
    \begin{align}
        \Csym{2}(x,1-x,0)&=1+\ind(x<0.9)~, \label{eqn:2d-symmetry-xy}\\
        \Csym{2}(x,0,1-x)&=1+2\cdot \ind(x<0.9)~, \label{eqn:2d-symmetry-xz}\\
        \Csym{2}(0,x,1-x)&=2+\ind(x<0.9)~. \label{eqn:2d-symmetry-yz}
    \end{align}
    This directly implies $(x,1-x,0)\sim_{\Csym{2}}(x,0,1-x)\sim_{\Csym{2}}(0,x,1-x)$ for any $x\in [0,1]$. 

    For $(x,1-x,0)$, the projection is trivially $(x,1-x)$. Then, we discuss three cases to prove \cref{eqn:2d-symmetry-xy}:
    \begin{itemize}[itemsep=0.2em, topsep=0.5em]
        \item If $x\in 0.9\pm \eps^2$, $0.5\eps^{-2}\cdot (y_0-0.1)\in \pm 0.5$ and thus $\tilde{y}_0=0.5+0.5\eps^{-2} \cdot (0.9-x)$. The projection $\vec{y}^{(2)}$ in consideration, is then 
        \[
            \left(0.5-0.5\eps^{-2} \cdot (0.9-x),\ 0.5+0.5\eps^{-2} \cdot (0.9-x),\ 0\right)~. 
        \]
        If $x<0.9$, $y_{2}^{(2)}>0.5$. Otherwise if $x\geq 0.9$, $y_{2}^{(2)}\leq 0.5$. 
        Since $C(1-y, y,0) = 1+\ind(y > 0.5)$ (\cref{lem:base-instance-bottom-slice}), we have $\Csym{2}(x, 1-x, 0)=1+\ind(x<0.9)$ for this subcase.
        \item If $x>0.9+\eps^2$, $0.5\eps^{-2}\cdot (y_0-0.1)<-0.5$ and thus $\tilde{y}_0=0$. The intermediate projection $\vec{y}^{(2)}$ in consideration, is then $(1,0,0)$. Since $C(1,0,0)=1$ (\cref{lem:base-instance-lr-boundaries}), we have $\Csym{2}(x, 1-x, 0)=1$ in this subcase.
        \item If $x<0.9-\eps^2$, $0.5\eps^{-2}\cdot (y_0-0.1)>0.5$ and thus $\tilde{y}_0=1$. The intermediate projection $\vec{y}^{(2)}$ in consideration, is then $(0,1,0)$. Since $C(0,1,0)=2$ (\cref{lem:base-instance-lr-boundaries}), we have $\Csym{2}(x, 1-x, 0)=2$ in this subcase.
    \end{itemize}

    For $(x,0,1-x)$, the projection is $(1,0)$. Therefore, $\tilde{y}_0=0$ and the intermediate projection $\vec{y}^{(2)}=(x,0,1-x)$.
    According to \cref{lem:base-instance-lr-boundaries}, $\Csym{2}(x,0,1-x)=C(x,0,1-x)=1+2\cdot \ind(x<0.9)$, matching \cref{eqn:2d-symmetry-xz}.

    For $(0,x,1-x)$, the projection is $(0,1)$.
    Therefore, $\tilde{y}_0=1$ and the intermediate projection $\vec{y}^{(2)}=(0,x,1-x)$.
    According to \cref{lem:base-instance-lr-boundaries}, $\Csym{2}(0,x,1-x)=C(0,x,1-x)=2+ \ind(x<0.9)$, matching \cref{eqn:2d-symmetry-yz}.
\end{proof}

\paragraph{Hardness.} Note that we embed a \PPAD-hard 
instance $C$ inside $\Csym{2}$ with a scaling factor of $2^{-2n+1}=2\eps^2$ on the first and the second  coordinates.
It is then easy to obtain the following \PPAD-hardness result for $\Csym{2}$.

\begin{lemma}
    It is \PPAD-hard to find three points $\vec{x}^{(1)}, \vec{x}^{(2)}, \vec{x}^{(3)}$ such that 
    \begin{itemize}[itemsep=0.2em,topsep=0.5em]
        \item {\bf they are close enough to each other:} for any $i,j\in [3]$, $\normtxt{\vec{x}^{(i)}-\vec{x}^{(j)}}{\infty}\leq 2^{-3n}=\eps^3$, 
        \item {\bf they induce a trichormatic triangle:} $|\settxt{\Csym{2}(\vec{x}^{(i)}):i\in [3]}|=3$. 
    \end{itemize}
\end{lemma}

\subsection{Hard instances with three or more dimensions}
\label{subsec:symmetric-kd}
For three or higher dimensions, we use almost the same recursive definition as in \cref{sec:warmup-proof-for-approx-sperner}.
We repeat it in \cref{alg:approx-symmetric-sperner-kgeq3} for convenience. 
In the rest of this subsection, we will show that our construction is symmetric and still \PPAD-hard.
Here, we use the same way to define the {\em modified neighboring color}:
\begin{align}
    \label{eqn:modified-neighboring-color-symmetry}
    \hat{C}^\alpha_{\sf nn}(\vec{x}) = \begin{cases}
        C^\alpha_{\sf nn}(\vec{x}) & \text{if $d^\alpha(\vec{x}, \nna(\vec{x}))< 2\eps^2$,}\\
        \hfill 2 \hfill & \text{if $d^\alpha(\vec{x}, \nna(\vec{x}))\geq 2\eps^2$ and $C(\vec{x})=1$,}\\
        \hfill 1 \hfill & \text{if $d^\alpha(\vec{x}, \nna(\vec{x}))\geq 2\eps^2$ and $C(\vec{x})\in \{2,3\}$.}
    \end{cases}
\end{align}
\begin{fact}
\label{fact:modified-neighboring-color-for-symmetry}
    For any $\vec{x}\in \Delta^2$, we have $C(\vec{x})<\hat{C}^\alpha_{\sf nn}(\vec{x})$ if $\rela(\vec{x})=0$, and $C(\vec{x})>\hat{C}^\alpha_{\sf nn}(\vec{x})$ if $\rela(\vec{x})=1$.
\end{fact}

We continue to use $\vec{y}^{(i)}(\vec{x}) = (y^{(i)}_1(\vec{x}), y^{(i)}_2(\vec{x}), y^{(i)}_3(\vec{x}))$
to denote the intermediate projections we compute in \cref{alg:approx-symmetric-sperner-kgeq3} when the input is $\vec{x}$, and use
$\vec{c}^{(i)}(\vec{x}) = (c^{(i)}_1(\vec{x}), c^{(i)}_2(\vec{x}), c^{(i)}_3(\vec{x}))$ 
to denote the intermediate palettes we use in \cref{alg:approx-symmetric-sperner-kgeq3}.
In addition, we introduce a new notation $\tilde{y}_0(\vec{x})$ for the converted coordinate we compute on \cref{line:first-converted-coordinate}.
Because the subscripts we will use for $\vec{c}^{(i)}$ can be very complicated, we will use $c_j^{(i)}(\vec{x})$ and $c^{(i)}(\vec{x}, j)$ interchangeably for better presentation.
For any vector $\vec{y}\in \Delta^2$, we use $i^*(\vec{y})$ to denote the first non-zero index of $\vec{y}$, i.e., 
\begin{align*}
    i^*(\vec{y}) = \begin{cases}
        1 & \text{if $y_1>0$,}\\
        2 & \text{if $y_1=0$ and $y_2>0$,}\\
        3 & \text{otherwise.}
    \end{cases}
\end{align*}

\subsubsection{Symmetry}
To show that $C^{(k)}$ are valid \outofk Approximate Symmetric \Sperner instances, we consider a stronger symmetry property of our construction 
for our induction hypothesis.
In this stronger symmetry, we consider the following two colors for each point $\vec{x}$:
\begin{enumerate}[itemsep=0.2em, topsep=0.5em]
    \item {\bf the output color:} 
    $\Csym{k}(\vec{x}) = c^{(k)}(\vec{x},\, C(\vec{y}^{(k)}(\vec{x})))$ (by definition of \cref{alg:approx-symmetric-sperner-kgeq3}). 
    \item {\bf the final neighboring color:} we use $C^{(k)}_{\sf nn}(\vec{x})$ to denote the intermediate color $c^{(k)}(\vec{x},\, \hat{C}_{\sf nn}^{\alpha}(\vec{y}^{(k)}(\vec{x})))$, which further equals to $c^{(k)}(\vec{x},\, C_{\sf nn}^{\alpha}(\vec{y}^{(k)}(\vec{x})))$ when $\vec{y}^{(k)}(\vec{x})$ is hot.
\end{enumerate}

The stronger symmetry states that symmetric points $\vec{x}^{(1)}, \vec{x}^{(2)}$ on the boundaries are symmetric both in the output color and the final neighboring color.
Further, if we try to convert their last intermediate projections, $\vec{y}^{(k)}(\vec{x}^{(1)})$ and $\vec{y}^{(k)}(\vec{x}^{(2)})$, we can obtain equivalent converted coordinates. 
For convenience, we also restate the definition of the symmetry property we are going to establish.
\symmetrydef*
\symsperner*
\begin{lemma}
    \label{lem:stronger-symmetry}
    Consider any $k\geq 2$, any two indices $j_1,j_2\in [k+1]$ and any $\vec{x}^{(1)}, \vec{x}^{(2)}\in \Delta^{k}$ such that $x^{(1)}_{j_1}=x^{(2)}_{j_2}=0$ and $\vec{x}^{(1)}_{-j_1} = \vec{x}^{(2)}_{-j_2}$.  
    $\vec{x}^{(1)}$ and $\vec{x}^{(2)}$ are
    \begin{itemize}[itemsep=0.2em, topsep=0.5em]
        \item {\bf symmetric in the output color:} $\vec{x}^{(1)}\sim_{\Csym{k}} \vec{x}^{(2)}$;
        \item {\bf their subsequent converted coordinates are equivalent:} $\rela(\vec{y}^{(k)}(\vec{x}^{(1)})) \simrel \rela(\vec{y}^{(k)}(\vec{x}^{(2)}))$; and
        \item {\bf symmetric in the final neighboring color:} $\vec{x}^{(1)} \sim_{C^{(k)}_{\sf nn}} \vec{x}^{(2)}$ if $\vec{y}^{(k)}(\vec{x}^{(1)}), \vec{y}^{(k)}(\vec{x}^{(2)})$ are hot. 
    \end{itemize}
\end{lemma}
\noindent
Note that, in this lemma, the first bullet directly implies the symmetry property we desire for Symmetric \Sperner instances (\cref{def:approx-symmetric-kd-sperner}). 
\begin{corollary}
    For any $k\geq 2$, $\Csym{k}$ constructed by $\cref{alg:approx-symmetric-sperner-kgeq3}$ is a valid \outofk Approximate Symmetric \Sperner instance.
\end{corollary}

Next, we establish this stronger symmetry by induction. 
The base case is when $k=2$, which we will prove by simple calculations. 
Later, we will establish this stronger symmetry for $k=k_0\geq 3$ under the assumption that we have established it for every $k<k_0$.
More specifically, we shall discuss three cases on the indices $j_1,j_2$ of asymmetric zero entries of $\vec{x}^{(1)}, \vec{x}^{(2)}$: 
\begin{quote}
(1) $j_1,j_2\leq k$; \hfill (2) $j_1=k$ and $j_2=k+1$; \hfill (3) $j_1<k$ and $j_2=k+1$. 
\end{quote}
Case (1) can be interpreted as appending a same $k+1$-th coordinate to two vectors $\vec{x}^{(1)}_{-(k+1)}, \vec{x}^{(2)}_{-(k+1)}$ that are symmetric in the smaller instance defined on $\Delta^{k-1}$.
We will use our induction hypothesis to establish this case. 
Case (2) can be interpreted as appending a same (possibly non-zero) coordinate and a zero coordinate to the same vector $\vec{x}^{(1)}_{1:k-1}=\vec{x}^{(2)}_{1:k-1}$ but in different orders for $\vec{x}^{(1)}$ and $\vec{x}^{(2)}$.
We will simulate the algorithm to establish this case.
Finally, Case (3) include all the other cases and can be decomposed into a piece of Case (1) and a piece of Case (2), by finding an intermediate vector $\vec{x}^{(3)}$ such that the symmetry between $\vec{x}^{(1)}, \vec{x}^{(3)}$ follows Case (1) and the symmetry between $\vec{x}^{(2)}, \vec{x}^{(3)}$ follows Case (2).
We will simply establish this case by the fact that symmetry is an equivalence relation (\cref{lem:symmetry-is-equiv-relation}). 

One useful lemma is that the symmetry between points is preserved under the max/min operations on the colorings. 
\begin{lemma}
    \label{lem:symmetry-preserved-under-max/min}
    For any two colorings $C_1,C_2$ and two vectors $\vec{x}, \vec{y}$, if $\vec{x}\sim_{C_1} \vec{y}$ and $\vec{x}\sim_{C_2}\vec{y}$, they are also symmetric in their max-coloring $\max(C_1,C_2)$ (i.e., $\max(C_1,C_2)(\vec{x}) = \max\settxt{C_1(\vec{x}), C_2(\vec{x})}$ for any $\vec{x}$) and min-coloring $\min(C_1,C_2)$ (resp., $\max(C_1,C_2)(\vec{x}) = \max\settxt{C_1(\vec{x}), C_2(\vec{x})}$ for any $\vec{x}$). 
\end{lemma}
\begin{proof}
    According to \cref{def:symmetry-in-coloring}, the given symmetries both imply $|\mathcal{I}_{>0}(\vec{x})| = |\mathcal{I}_{>0}(\vec{y})|$. 
    For simplicity, we let $\mathcal{I}_x = \mathcal{I}_{>0}(\vec{x})$ and $\mathcal{I}_y = \mathcal{I}_{>0}(\vec{y})$. 
    In addition, $\vec{x}\sim_{C_1} \vec{y}$ implies $\idx(\mathcal{I}_{x}, C_1(\vec{x})) = \idx(\mathcal{I}_{y}, C_1(\vec{y}))$, while $\vec{x}\sim_{C_2} \vec{y}$ implies $\idx(\mathcal{I}_{x}, C_2(\vec{x})) = \idx(\mathcal{I}_{y}, C_2(\vec{y}))$.
    If $C_1(\vec{x})=C_2(\vec{x})$, we have $C_1(\vec{y})=C_2(\vec{y})$ and the lemma trivially holds. 
    W.l.o.g., we suppose $C_1(\vec{x})< C_2(\vec{x})$ next. 
    Note that, according to \cref{def:indexing}, for any increasing array with positive integer entries $\vec{a}$, $\idx(\vec{a}, v)$ is a strictly increasing function for $v\in \vec{a}$. 
    We have $\idx(\mathcal{I}_{x}, C_1(\vec{x}))<\idx(\mathcal{I}_{x}, C_2(\vec{x}))$ and thus $\idx(\mathcal{I}_{y}, C_1(\vec{y}))<\idx(\mathcal{I}_{y}, C_2(\vec{y}))$.
    The later inequality further implies $C_1(\vec{y})<C_2(\vec{y})$. 
    Therefore, $\min(C_1,C_2)(\vec{x})=C_1(\vec{x})$ and $\min(C_1, C_2)(\vec{y})=C_1(\vec{y})$. 
    We have $\idx(\mathcal{I}_{x}, min(C_1,C_2)(\vec{x})) = \idx(\mathcal{I}_{y}, \min(C_1,C_2)(\vec{y}))$.
    According to \cref{def:symmetry-in-coloring}, we have $\vec{x}\sim_{\min(C_1,C_2)} \vec{y}$.
    Similarly, we have $\max(C_1,C_2)(\vec{x})=C_2(\vec{x})$ and $\max(C_1, C_2)(\vec{y})=C_2(\vec{y})$, and $\vec{x}\sim_{\max(C_1,C_2)} \vec{y}$.
\end{proof}

\paragraph{The base case ($k=2$).}
The first bullet for the base case has been established by \cref{lem:2d-symmetry}.
Next, we consider any fixed $z\in [0,1]$.
We will establish the second and the third bullet for $\vec{x}^{(1)}, \vec{x}^{(2)}$ that are in the set $\settxt{(1-z,z,0),\, (1-z,0,z),\, (0,1-z,z)}$.
According to the first two lines of \cref{alg:approx-symmetric-sperner-kgeq3}, $\tilde{y}_0(1-z,z,0)=(0.5+0.5\eps^{-2}\cdot(z-0.1))_{[0,1]}$, $\tilde{y}_0(1-z,0,z)=0$ and $\tilde{y}_0(0,1-z,z)=1$. 
Therefore, we have
\begin{align*}
    \vec{y}^{(2)}(1-z,z,0) &= \left(\left(0.5 - 0.5\eps^{-2}\cdot(z-0.1)\right)_{[0,1]},\, \left(0.5 + 0.5\eps^{-2}\cdot(z-0.1)\right)_{[0,1]},\, 0\right)~,
    \\
    \vec{y}^{(2)}(1-z,0,z) &= (1-z,0,z)~,\\
    \vec{y}^{(2)}(0,1-z,z) &= (0,1-z,z)~.
\end{align*}
According to \cref{lem:new-converted-coordinates-on-base-1-simplex,,lem:symmetry-on-converted-coordinates-for-symmetric-sperner}, we can obtain the second bullet for $k=2$: 
\begin{align*}
    \rela(\vec{y}^{(2)}(1-z,0,z)) \,\simrel\,
    \rela(\vec{y}^{(2)}(0,1-z,z)) &\,\simrel\,
    \\
    \rela(\vec{y}^{(2)}(1-z,z,0)) &\,=\,
    \left(0.5 + 0.5\eps^{-2}\cdot(z-0.1)\right)_{[0,1]}~.
\end{align*}
This equivalence on the subsequent converted coordinate implies that $\vec{y}^{(2)}(\vec{x}^{(1)}), \vec{y}^{(2)}(\vec{x}^{(2)})$ are hot if and only if $z\in 0.1\pm \eps^2$. 
Further, according to \cref{def:symmetry-in-coloring}, \cref{lem:new-converted-coordinates-on-base-1-simplex,,lem:symmetry-on-converted-coordinates-for-symmetric-sperner}, it is easy to obtain $(1-z,z,0)\sim_{C_{\sf nn}^{(2)}} (1-z,0,z) \sim_{C_{\sf nn}^{(2)}} (0,1-z,z)$. 

\paragraph{General case 1 ($k\geq 3$ and $j_1,j_2\leq k$).}
In this subcase, we have $x^{(1)}_{k+1}=x^{(2)}_{k+1}$.
We can w.l.o.g. assume that $x^{(1)}_{k+1}<1$ because otherwise $\vec{x}^{(1)}=\vec{x}^{(2)}$ and this stronger symmetry trivially holds.
Let $\vec{\hat{x}}^{(1)} = \vec{P}(\vec{x}^{(1)})$ and $\vec{\hat{x}}^{(2)} = \vec{P}(\vec{x}^{(2)})$.
According to the definition of projection operation $\vec{P}$ (\cref{def:projection}), we have 
\begin{align*}
    \vec{\hat{x}}^{(1)} &= \frac{\vec{x}^{(1)}_{1:k}}{1-x^{(1)}_{k+1}}~, \\
    \vec{\hat{x}}^{(2)} &= \frac{\vec{x}^{(2)}_{1:k}}{1-x^{(2)}_{k+1}} = \frac{\vec{x}^{(2)}_{1:k}}{1-x^{(1)}_{k+1}}~. 
\end{align*}
Hence, they satisfy the same symmetry as $\vec{x}^{(1)}$ and $\vec{x}^{(2)}$, i.e.,
we have $\hat{x}^{(1)}_{j_1}=\hat{x}^{(2)}_{j_2}=0$ and $(\vec{\hat{x}}^{(1)})_{-j_1}=(\vec{\hat{x}}^{(2)})_{-j_2}$. 
An observation is that the first $k-2$ intermediate projections we use for any input $\vec{x}\in \Delta^k$ 
are exactly the same as those for its one-step projection $\vec{P}(\vec{x})\in \Delta^{k-1}$ (in a different run of \cref{alg:approx-symmetric-sperner-kgeq3} on a $k-1$-dimensional simplex).
\begin{observation}
    \label{obs:equal-y-on-projection}
    For any $k\geq 3$, any $\vec{x}\in \Delta^k$ and any $i\in [k-1]$, we have $\vec{y}^{(i)}(\vec{x}) = \vec{y}^{(i)}(\vec{P}(\vec{x}))$. 
\end{observation}
\noindent Using this observation and the second bullet of our induction hypothesis for $k-1$, we can easily obtain the equivalence between the converted coordinates of the $(k-1)$-th intermediate projections $\vec{y}^{(k-1)}(\cdot)$ of $\vec{x}^{(1)}$ and $\vec{x}^{(2)}$.
\begin{align}
    \rela\left(\vec{y}^{(k-1)}(\vec{x}^{(1)})\right) 
    = 
    \rela\left(\vec{y}^{(k-1)}(\vec{\hat{x}}^{(1)})\right) 
    \simrel
    \rela\left(\vec{y}^{(k-1)}(\vec{\hat{x}}^{(2)})\right) 
    = 
    \rela\left(\vec{y}^{(k-1)}(\vec{x}^{(2)})\right).
    \label{eqn:symmetric-converted-induction-hypothesis}
\end{align}
Since \cref{line:symmetric-y^{(i)}_3} gives $y^{(k)}_3(\vec{x}^{(1)}) = x^{(1)}_{k+1}$ and $y^{(k)}_3(\vec{x}^{(2)}) = x^{(2)}_{k+1}$, we clearly have $y^{(k)}_3(\vec{x}^{(1)})=y^{(k)}_3(\vec{x}^{(2)})$.
For simplicity, we let $z=y^{(k)}_3(\vec{x}^{(1)})$ next. 
According to \cref{line:symmetric-y^{(i)}_2,,line:symmetric-y^{(i)}_1} of \cref{alg:approx-symmetric-sperner-kgeq3} and
because of \cref{eqn:symmetric-converted-induction-hypothesis}, we also have either 
\begin{itemize}[itemsep=0em,topsep=0.5em]
    \item $y^{(k)}_2(\vec{x}^{(1)}) = y^{(k)}_2(\vec{x}^{(2)})\in (0,1-z)$, which implies $\vec{y}^{(k)}(\vec{x}^{(1)}) = \vec{y}^{(k)}(\vec{x}^{(2)})$, or 
    \item $y^{(k)}_2(\vec{x}^{(1)}),\ y^{(k)}_2(\vec{x}^{(2)}) \in \{0,1-z\}$.
\end{itemize} 
According to \cref{lem:symmetry-on-converted-coordinates-for-symmetric-sperner}, we have $\rela(\vec{y}^{(k)}(\vec{x}^{(1)})) \simrel \rela(\vec{y}^{(k)}(\vec{x}^{(2)}))$, the second bullet of \cref{lem:stronger-symmetry}. 
Next, we prove the first bullet and the third bullet by further discussing these two subcases.
Recall the following definitions:
\begin{align*}
    \Csym{k}(\vec{x}) &:= c^{(k)}\left(\vec{x},\ C(\vec{y}^{(k)}(\vec{x}))\right)~,
    \\
    C_{\sf nn}^{(k)}(\vec{x}) &:= c^{(k)}\left(\vec{x},\ C_{\sf nn}^{\alpha}(\vec{y}^{(k)}(\vec{x}))\right)~, \quad \text{when $\vec{y}^{(k)}(\vec{x})$ is hot}~.
\end{align*}

\paragraph{~~}
First, we consider the subcase in the earlier first bullet. 
That is, we suppose that $y^{(k)}_2(\vec{x}^{(1)}) = y^{(k)}_2(\vec{x}^{(2)})\in (0,1-z)$, which implies $\vec{y}^{(k)}(\vec{x}^{(1)}) = \vec{y}^{(k)}(\vec{x}^{(2)})$.
According to \cref{line:symmetric-y^{(i)}_2}, we have 
\begin{align*}
    \rela(\vec{y}^{(k-1)}(\vec{x}^{(1)})),\ \rela(\vec{y}^{(k-1)}(\vec{x}^{(2)}))\in (0,1)~,
\end{align*}
and thus the $(k-1)$-th intermediate projections $\vec{y}^{(k-1)}(\vec{x}^{(1)})$ and $\vec{y}^{(k-1)}(\vec{x}^{(2)})$ are hot.
Recall that our induction hypothesis gives that $\vec{\hat{x}}^{(1)}\sim_{\Csym{k-1}}\vec{\hat{x}}^{(2)}$ and $\vec{\hat{x}}^{(1)}\sim_{C_{\sf nn}^{(k-1)}} \vec{\hat{x}}^{(2)}$. 
According to \cref{line:obtain-new-set-of-colors-for-symmetry}, the first colors in the intermediate color palette $\vec{c}^{(k)}$ of any $\vec{x}$ form the set $\{\Csym{k-1}(\vec{P}(\vec{x})), C^{(k-1)}_{\sf nn}(\vec{P}(\vec{x}))\}$, in an order where $c_1^{(k)}(\vec{x})<c_2^{(k)}(\vec{x})$.
Because min/max operations on the two colorings preserve their symmetry (\cref{lem:symmetry-preserved-under-max/min}),
$\vec{x}^{(1)}$ and $\vec{x}^{(2)}$ are symmetric in the intermediate colors $c^{(k)}_1(\cdot)$ and $c^{(k)}_2(\cdot)$, i.e., 
\begin{align*}
    \vec{x}^{(1)} \sim_{c^{(k)}_1(\cdot)} \vec{x}^{(2)}~,
    \quad \text{ and } \quad
    \vec{x}^{(1)} \sim_{c^{(k)}_2(\cdot)} \vec{x}^{(2)}~.
\end{align*}
Note that $c_3^{(k)}(\vec{x}^{(1)})=k+1=c_3^{(k)}(\vec{x}^{(2)})$.
Then, we clearly have the first and the third bullet of \cref{lem:stronger-symmetry}, $\vec{x}^{(1)}\sim_{\Csym{k}} \vec{x}^{(2)}$ and $\vec{x}^{(1)}\sim_{C^{(k)}_{\sf nn}} \vec{x}^{(2)}$ (when $\vec{y}^{(k)}(\vec{x}^{(1)})$ and $\vec{y}^{(k)}(\vec{x}^{(2)})$ are hot), because we trivially have $C(\vec{y}^{(k)}(\vec{x}^{(1)}))=C(\vec{y}^{(k)}(\vec{x}^{(2)}))$ and $C_{\sf nn}^{\alpha}(\vec{y}^{(k)}(\vec{x}^{(1)}))=C_{\sf nn}^{\alpha}(\vec{y}^{(k)}(\vec{x}^{(2)}))$ by $\vec{y}^{(k)}(\vec{x}^{(1)})=\vec{y}^{(k)}(\vec{x}^{(2)})$.

\paragraph{~~}
On the other hand, we consider the subcase in the earlier second bullet. That is, we suppose that $y^{(k)}_2(\vec{x}^{(1)}),\ y^{(k)}_2(\vec{x}^{(2)}) \in \{0,1-z\}$ (i.e., we have $\vec{y}^{(k)}(\vec{x}^{(1)}), \vec{y}^{(k)}(\vec{x}^{(2)})\in \{(1-z,0,z), (0,1-z,z)\}$).
Our induction hypothesis only gives us $\vec{\hat{x}}^{(1)}\sim_{\Csym{k-1}}\vec{\hat{x}}^{(2)}$ here. 
According to \cref{line:symmetric-y^{(i)}_2}, we have 
\begin{align*}
    \rela(\vec{y}^{(k-1)}(\vec{x}^{(1)})),\ \rela(\vec{y}^{(k-1)}(\vec{x}^{(2)}))\in \{0,1\}~,
\end{align*} 
and thus the $(k-1)$-th intermediate projections $\vec{y}^{(k-1)}(\vec{x}^{(1)})$ and $\vec{y}^{(k-1)}(\vec{x}^{(2)})$ are not hot.
Note that we have $\Csym{k-1}(\vec{\hat{x}}^{(1)}) = c^{(k-1)}(\vec{x}^{(1)}, C(\vec{y}^{(k-1)}(\vec{x}^{(1)})))<c^{(k-1)}(\vec{x}^{(1)}, \hat{C}^{\alpha}_{\sf nn}(\vec{y}^{(k-1)}(\vec{x}^{(1)})))$ if and only if $y_2^{(k)}(\vec{x}^{(1)})=0$ (similarly, for $\vec{x}^{(2)}$).
Therefore, according to \cref{line:obtain-new-set-of-colors-for-symmetry} of \cref{alg:approx-symmetric-sperner-kgeq3}, we have
\begin{align*}
 c^{(k)}(\vec{x}^{(1)}, i^*(\vec{y}^{(k)}(\vec{x}^{(1)})))=\Csym{k-1}(\vec{\hat{x}}^{(1)}) 
 \quad \text{and} \quad 
 c^{(k)}(\vec{x}^{(2)}, i^*(\vec{y}^{(k)}(\vec{x}^{(2)})))=\Csym{k-1}(\vec{\hat{x}}^{(2)})~.
\end{align*}
Note we have $C(1-z,0,z)=1,C(0,1-z,z)=2$ (i.e., $C(1-z,0,z)=i^*(1-z,0,z)$ and $C(0,1-z,z)=i^*(0,1-z,z)$), or $C(1-z,0,z)=C(0,1-z,z)=3$ for any fixed $z\in [0,1]$ (\cref{lem:base-instance-lr-boundaries}), and the same thing holds if we replace $C$ by $C_{\sf nn}^{\alpha}$ (\cref{lem:symmetry-on-converted-coordinates-for-symmetric-sperner}) and we have $\rela(1-z,0,z)\in (0,1)$.
Because $\Csym{k}(\vec{x})$ is defined as $c^{(k)}(\vec{x}, C(\vec{y}^{(k)}(\vec{x})))$ and because of our induction hypothesis that $\vec{\hat{x}}^{(1)} \sim_{\Csym{k-1}} \vec{\hat{x}}^{(2)}$, we have $\vec{x}^{(1)} \sim_{\Csym{k}} \vec{x}^{(2)}$ and $\vec{x}^{(1)} \sim_{C^{(k)}_{\sf nn}} \vec{x}^{(2)}$ if $\vec{y}^{(k)}(\vec{x}^{(1)})$ is hot (i.e., $\rela(\vec{y}^{(k)}(\vec{x}^{(1)}))\in (0,1)$).

\paragraph{General case 2 ($k\geq 3$, $j_1=k$ and $j_2=k+1$).} 
In this case, we have $x^{(1)}_{k+1}=x^{(2)}_{k}$ and $x^{(1)}_{k}=x^{(2)}_{k+1}=0$. 
Let $\vec{\hat{x}}^{(1)}=\vec{P}^{(2)}(\vec{x}^{(1)})$ and $\vec{\hat{x}}^{(2)}=\vec{P}^{(2)}(\vec{x}^{(2)})$. 
According to the definition of the projection step $\vec{P}$ (\cref{def:projection}), we have 
\begin{align*}
    \vec{\hat{x}}^{(1)}&=\frac{\vec{x}^{(1)}_{1:k-1}}{1-x^{(1)}_{k+1}}~,\\
    \vec{\hat{x}}^{(2)}&=\frac{\vec{x}^{(2)}_{1:k-1}}{1-x^{(2)}_{k}}=\frac{\vec{x}^{(2)}_{1:k-1}}{1-x^{(1)}_{k+1}}~.
\end{align*}
Because $\vec{x}^{(1)}_{-(k+1)}=\vec{x}^{(2)}_{-k}$, which implies $\vec{x}^{(1)}_{1:k-1}=\vec{x}^{(2)}_{1:k-1}$, we have $\vec{\hat{x}}^{(1)}=\vec{\hat{x}}^{(2)}$. 

Note that $k\geq 3$.
We have $\vec{\hat{x}}^{(1)}, \vec{\hat{x}}^{(2)}\in \Delta^{k-2}$.
Therefore, we have $\vec{P}^{(k-2)}(\vec{x}^{(1)})=\vec{P}^{(k-2)}(\vec{x}^{(2)})$ and further $\tilde{y}_0(\vec{x}^{(1)})=\tilde{y}_0(\vec{x}^{(2)})$. 
In addition, we trivially have $\vec{c}^{(2)}(\vec{x}^{(1)}) = (1,2,3) = \vec{c}^{(2)}(\vec{x}^{(2)})$. 
Similarly, according to \cref{obs:equal-y-on-projection}, we can prove that $\vec{y}^{(k-2)}(\vec{x}^{(1)})=\vec{y}^{(k-2)}(\vec{x}^{(2)})$ and $\vec{c}^{(k-1)}(\vec{x}^{(1)})=\vec{c}^{(k-1)}(\vec{x}^{(2)})$ if $k\geq 4$ (because \cref{alg:approx-symmetric-sperner-kgeq3} computes $\vec{c}^{(k-1)}(\vec{x})$ only based on $\vec{y}^{(k-2)}(\vec{x})$). 

For simplicity, we let $\tilde{y} = \rela(\vec{y}^{(k-2)}(\vec{x}^{(1)}))$ if $k\geq 4$ and let $\tilde{y} = \tilde{y}_0(\vec{x}^{(1)})$ if $k=3$.
According to our previous discussions, we have $\tilde{y}=\rela(\vec{y}^{(k-2)}(\vec{x}^{(2)}))$ if $k\geq 4$ and $\tilde{y} = \tilde{y}_0(\vec{x}^{(2)})$ if $k=3$.
Further we let $z:=x^{(1)}_{k+1}$ and $(c_1,c_2,c_3):=\vec{c}^{(k-1)}(\vec{x}^{(1)})\in [k]^{3}$, which clearly gives $c_3=k$ according to \cref{line:obtain-new-set-of-colors-for-symmetry} of \cref{alg:approx-symmetric-sperner-kgeq3}.  
Again, according to our previous discussion, we have $z=x^{(2)}_k$ and $\vec{c}^{(k-1)}(\vec{x}^{(2)})=(c_1,c_2,c_3)$. 
Also, recall that $x_{k}^{(1)}=x_{k+1}^{(2)}=0$.
Finally, we use $\vec{y}^*:=((1-z)\cdot (1-\tilde{y}),\; (1-z)\cdot \tilde{y},\; z)$ to denote a key vector in our following proof, and use $\tilde{y}^*:=\rela(\vec{y^*})$ to denote its converted coordinate.

First, we prove the second bullet of \cref{lem:stronger-symmetry}, $\rela(\vec{y}^{(k)}(\vec{x}^{(1)})) = \rela(\vec{y}^{(k)}(\vec{x}^{(2)}))$, by the following simple simulations for \cref{alg:approx-symmetric-sperner-kgeq3}, where
we prove that $\rela(\vec{y}^{(k)}(\vec{x}^{(1)})) = \rela(\vec{y}^{(k)}(\vec{x}^{(2)}))=\tilde{y}^*$.
\begin{itemize}[itemsep=0.2em, topsep=0.5em]
    \item For $\vec{x}^{(1)}$, we have $\vec{y}^{(k-1)}(\vec{x}^{(1)})=(1-\tilde{y}, \tilde{y}, 0)$. According to \cref{lem:new-converted-coordinates-on-base-1-simplex}, $\rela(\vec{y}^{(k-1)}(\vec{x}^{(1)})) = \tilde{y}$. 
    Therefore, we have $\vec{y}^{(k)}(\vec{x}^{(1)}) = ((1-z)\cdot (1-\tilde{y}), (1-z)\cdot \tilde{y}, z)=\vec{y^*}$ and thus $\rela(\vec{y}^{(k)}(\vec{x}^{(1)}))=\tilde{y}^*$. 
    \item For $\vec{x}^{(2)}$, we have $\vec{y}^{(k-1)}(\vec{x}^{(2)})=((1-z)\cdot (1-\tilde{y}), (1-z) \cdot \tilde{y}, z)=\vec{y^*}$.
    According to our definition of $\tilde{y}^*$, $\rela(\vec{y}^{(k)}(\vec{x}^{(2)}))=\tilde{y}^*$ and thus $\vec{y}^{(k)}(\vec{x}^{(2)})=(1-\tilde{y}^*,\tilde{y}^*,0)$.
    According to \cref{lem:new-converted-coordinates-on-base-1-simplex}, we have $\rela(\vec{y}^{(k)}(\vec{x}^{(2)}))=\tilde{y}^*$.
\end{itemize}

Next, we prove the first and the third bullet of \cref{lem:stronger-symmetry}, $\vec{x}^{(1)} \sim_{\Csym{k}} \vec{x}^{(2)}$ and $\vec{x}^{(1)}\sim_{C^{(k)}_{\sf nn}} \vec{x}^{(2)}$.
Our proof is also based on simulations for \cref{alg:approx-symmetric-sperner-kgeq3}, where we give characterizations of $\Csym{k}(\vec{x}^{(1)}), C^{(k)}_{\sf nn}(\vec{x}^{(1)})$ and $\Csym{k}(\vec{x}^{(2)}), C^{(k)}_{\sf nn}(\vec{x}^{(2)})$ by \cref{eqn:colors-of-x1-on-symmetry-gc2,,eqn:colors-of-x2-on-symmetry-gc2}, respectively.
\begin{itemize}[itemsep=0.2em, topsep=0.5em]
    \item For $\vec{x}^{(1)}$, according to our previous calculations, $\vec{y}^{(k-1)}(\vec{x}^{(1)})=(1-\tilde{y}, \tilde{y}, 0)$ and $\vec{y}^{(k)}(\vec{x}^{(1)})=\vec{y^*}$. 
    Because of \cref{lem:base-instance-bottom-slice,,lem:new-converted-coordinates-on-base-1-simplex,,eqn:modified-neighboring-color-symmetry}, $\vec{y}^{(k-1)}(\vec{x}^{(1)})$ is either hot or warm, and we have $C(\vec{y}^{(k-1)}(\vec{x}^{(1)})),$ $\hat{C}_{\sf nn}^{\alpha}(\vec{y}^{(k-1)}(\vec{x}^{(1)})) \in \{1,2\}$. 
    Therefore, $\vec{c}^{(k)}(\vec{x}^{(1)}) = (c_1,c_2,k+1)$.
    The colors $\Csym{k}$ and $C^{(k)}_{\sf nn}$ of $\vec{x}^{(1)}$ are then characterized as follows:
    \begin{align}
        \label{eqn:colors-of-x1-on-symmetry-gc2}
        \Csym{k}(\vec{x}^{(1)}) = \begin{cases}
            \hfill c_1 \hfill  & \text{if $C(\vec{y}^*)=1$,} \\
            \hfill c_2 \hfill & \text{if $C(\vec{y}^*)=2$,} \\
            k+1 & \text{if $C(\vec{y}^*)=3$;}
        \end{cases}
        \quad
        \text{ and }
        \quad 
        C_{\sf nn}^{(k)}(\vec{x}^{(1)}) = \begin{cases}
            \hfill c_1 \hfill  & \text{if $\hat{C}^{\alpha}_{\sf nn}(\vec{y}^*)=1$,} \\
            \hfill c_2 \hfill & \text{if $\hat{C}^{\alpha}_{\sf nn}(\vec{y}^*)=2$,} \\
            k+1 & \text{if $\hat{C}^{\alpha}_{\sf nn}(\vec{y}^*)=3$.}
        \end{cases}
    \end{align}
    \item For $\vec{x}^{(2)}$, according to our previous calculations, $\vec{y}^{(k-1)}(\vec{x}^{(2)})=\vec{y^*}$ and $\vec{y}^{(k)}(\vec{x}^{(2)}) =(1-\tilde{y}^*, \tilde{y}^*, 0)$. 
    We have $c^{(k)}_1(\vec{x}^{(2)}), c^{(k)}_2(\vec{x}^{(2)}) \in \{c_i: i\in \{C(\vec{y^*}), \hat{C}_{\sf nn}^{\alpha}(\vec{y^*})\}\}$.
    Because of \cref{lem:base-instance-bottom-slice,,lem:new-converted-coordinates-on-base-1-simplex,,eqn:modified-neighboring-color-symmetry},
    $(\tilde{y}^*, 1-\tilde{y}^*, 0)$ is either hot or warm, and we have $C(1-\tilde{y}^*,\tilde{y}^*,0) \neq \hat{C}^{\alpha}_{\sf nn}(1-\tilde{y}^*,\tilde{y}^*,0) \in \{1,2\}$. 
    Also, note that we have $C(1-\tilde{y}^*,\tilde{y}^*,0)=1$ if and only if $\tilde{y}^*\leq 0.5$ if and only if $C(\vec{y^*})<\hat{C}_{\sf nn}^{\alpha}(\vec{y^*})$ (\cref{eqn:new-coordinate-converter,,fact:modified-neighboring-color}) if and only if $c_1^{(k)}(\vec{x}^{(2)})=c_\ell$ for $\ell=C(\vec{y^*})$. 
    Because $\Csym{k}(\vec{x}^{(2)})=c_{\ell}^{(k)}(\vec{x}^{(2)})$ for $\ell=C(1-\tilde{y}^*,\tilde{y}^*,0)$ and $C^{(k)}_{\sf nn}(\vec{x}^{(2)})=c_{\ell'}^{(k)}(\vec{x}^{(2)})$ for $\ell'=\hat{C}^{\alpha}_{\sf nn}(1-\tilde{y}^*,\tilde{y}^*,0)$, we have $\Csym{k}(\vec{x}^{(2)})=c_{\ell}$ for $\ell=C(\vec{y^*})$ and $C^{(k)}_{\sf nn}(\vec{x}^{(2)})=c_{\ell'}$ for $\ell'=\hat{C}_{\sf nn}^{\alpha}(\vec{y^*})$.
    Note that $c_3=k$.
    We can characterize the colors $\Csym{k}$ and $C^{(k)}_{\sf nn}$ of $\vec{x}^{(2)}$ as follows:
    \begin{align}
        \label{eqn:colors-of-x2-on-symmetry-gc2}
        \Csym{k}(\vec{x}^{(2)}) = \begin{cases}
            \hfill c_1 \hfill  & \text{if $C(\vec{y}^*)=1$,} \\
            \hfill c_2 \hfill & \text{if $C(\vec{y}^*)=2$,} \\
            \hfill k \hfill & \text{if $C(\vec{y}^*)=3$;}
        \end{cases}
        \quad
        \text{ and }
        \quad 
        C_{\sf nn}^{(k)}(\vec{x}^{(2)}) = \begin{cases}
            \hfill c_1 \hfill  & \text{if $\hat{C}^{\alpha}_{\sf nn}(\vec{y}^*)=1$,} \\
            \hfill c_2 \hfill & \text{if $\hat{C}^{\alpha}_{\sf nn}(\vec{y}^*)=2$,} \\
            \hfill k \hfill & \text{if $\hat{C}^{\alpha}_{\sf nn}(\vec{y}^*)=3$.}
        \end{cases}
    \end{align}
\end{itemize}

\paragraph{General case 3 ($k\geq 3$, $j_1<k$ and $j_2=k+1$).} This case can be easily derived by the equivalence for the earlier two cases and the fact the symmetry in any coloring is an equivalence relation (\cref{lem:symmetry-is-equiv-relation}). 
That is, letting $\vec{x}=\vec{x}^{(2)}_{1:k}$ and $\vec{x}^{(3)}=(\vec{x}_{1:k-1},0,x_{k})$, the previous two general cases have already given us 
\begin{align*}
    &\vec{x}^{(1)} \sim_{\Csym{k}} \vec{x}^{(3)} \sim_{\Csym{k}} \vec{x}^{(2)}~,\\
    &\rela\left(\vec{y}^{(k)}(\vec{x}^{(1)})\right) \simrel \rela\left(\vec{y}^{(k)}(\vec{x}^{(3)})\right) \simrel \rela\left(\vec{y}^{(k)}(\vec{x}^{(2)})\right)~,\\
    &\vec{x}^{(1)} \sim_{C_{\sf nn}^{(k)}} \vec{x}^{(3)} \sim_{C_{\sf nn}^{(k)}} \vec{x}^{(2)} \qquad \text{if $\rela\left(\vec{y}^{(k)}(\vec{x}^{(1)})\right)\in (0,1)$}~.
\end{align*}

\subsubsection{Hardness}
It suffices to show \PPAD-hardness of the \outofk Approximate Symmetric \Sperner problem by showing that we can always recover a solution for the $2${\rm D}-\Sperner instance $C$ in polynomial time. 

\begin{restatable}[]{lemma}{repiv}
    \label{lem:poly-time-recovery-symmetry}
    Given oracle access to a $2${\rm D}-\Sperner instance $C$ and a tuple of points $(\vec{x}^{(1)}, \vec{x}^{(2)}, \vec{x}^{(3)})$ which satisfy $\normtxt{\vec{x}^{(i)}-\vec{x}^{(j)}}{\infty}\leq 2^{-4kn}$ for any $i,j\in [3]$, and which are trichromatic in the \outofk Approximate Symmetric \Sperner instance $\Csym{k}$ constructed by \cref{alg:approx-symmetric-sperner-kgeq3},
    then there is a polynomial-time algorithm that finds a tuple of points $(\vec{\hat{x}}^{(1)}, \vec{\hat{x}}^{(2)}, \vec{\hat{x}}^{(3)})$ which satsify $\normtxt{\vec{\hat{x}}^{(i)}-\vec{\hat{x}}^{(j)}}{\infty}\leq 2^{-n}$ for any $i,j\in [3]$ and which are trichromatic in the $2${\rm D}-\Sperner instance $C$.
\end{restatable}

The proof for the polynomial-time recovery algorithm (\cref{alg:recover-2d-sol-symmetry}) is almost the same as that in the warm-up \cref{sec:warmup-proof-for-approx-sperner},
as we have established all the analogs of the key lemmas used in the warm-up \cref{sec:warmup-proof-for-approx-sperner}. 
The only differences are 
\begin{enumerate}[itemsep=0em,topsep=0.5em]
    \item We are showing hardness for side-length $2^{-4kn}$ instead of $2^{-3kn}$ (caused by \cref{lem:lipschitz-of-the-converter-for-symmetry} vs. \cref{lem:lipschitz-of-the-converter}).
    \item Proving the Lipschitzness for $\vec{y}^{(2)}(\vec{x})$ (as $\vec{y}^{(2)}$ on \cref{line:first-intermediate-projection} of \cref{alg:approx-symmetric-sperner-kgeq3} is more sophisticated).
    \item Resolving with the issues caused by the fact that $d^{\alpha}(\cdot,\cdot)$ is no longer a metric. As we only use the fact $d(\cdot,\cdot)$ is a metric when we argue in \cref{lem:characterization-of-intermediate-colors} that points close to a hot point is hot or warm, we prove this argument in \cref{lem:close-to-hot-points} to resolve the issue.
\end{enumerate}
We defer the technical proofs to \cref{proof:poly-time-recovery-symmetry}.

\begin{algorithm2e}[t]
    \caption{Recover a $2${\rm D}-\Sperner solution from \outofk Approximate Symmetric \Sperner solutions}
    \label{alg:recover-2d-sol-symmetry}

    \DontPrintSemicolon
    \SetKwInOut{Input}{Input\,}
    \SetKwInOut{Output}{Output\,}

    \Input{~vectors $\vec{x}^{(1)}, \vec{x}^{(2)}, \vec{x}^{(3)}\in \Delta^k$}

    \Output{~vectors $\vec{\hat{x}}^{(1)}, \vec{\hat{x}}^{(2)}, \vec{\hat{x}}^{(3)}\in \Delta^2$}

    \For{$i\in \{2,3,\dots,k-1\}$}{
        \For{$j\in [3]$}{

            \If{$C$ has $3$ different colors in $\mathcal{N}(\vec{y}^{(i)}(\vec{x}^{(j)}))$ }{
                $\vec{\hat{x}}^{(1)}, \vec{\hat{x}}^{(2)}, \vec{\hat{x}}^{(3)} \gets $ any 3 points colored differently by $C$ in $\mathcal{N}(\vec{y}^{(i)}(\vec{x}^{(j)}))$
                \tcp*{the definition of $\vec{y}^{(i)}(\vec{x}^{(j)})$ follows \cref{alg:approx-symmetric-sperner-kgeq3}}

                \Return{$\vec{\hat{x}}^{(1)}, \vec{\hat{x}}^{(2)}, \vec{\hat{x}}^{(3)}$} \tcp*{trichromatic triangle found while simulation}
            }
        }
    }
    \Return{$\vec{y}^{(k)}(\vec{x}^{(1)}), \vec{y}^{(k)}(\vec{x}^{(2)}), \vec{y}^{(k)}(\vec{x}^{(3)})$}
    \tcp*{trichromatic triangle found when lifting $C^{(k-1)}$ to $C^{(k)}$}
\end{algorithm2e}

\subsection{Query Complexity}
\label{subsec:query-complexity}
Finally, in this subsection, we will establish the query complexity of the continuous version of the \outofk Approximate Symmetric \Sperner problem, which is $2^{\Omega(n)}/\poly(n,k)$ (\cref{thm:continuous-symmetry-main}). 
Our proof will be a careful step-by-step examination on our reduction.
Recall that our reduction consists of the following steps:
\begin{enumerate}[itemsep=0.2em,topsep=0.5em]
    \item the reduction from the 2{\rm D}-\rect\Sperner problem (instance $C_{\rect}$) to the continuous (triangular) version of the 2{\rm D}-\Sperner problem (the base instance $C$), 
    \item the reduction from the continuous version of the 2{\rm D}-\Sperner problem (the base instance $C$) to the continuous version of the \outofk Approximate Symmetric \Sperner problem (instance $\Csym{k}$).
\end{enumerate}

\paragraph{Step 1: reduction from $C_{\rect}$ to $C$.}
Recall that our construction of the base instance $C$ is transforming each point of $C_{\rect}$ to a $1.6\eps\times 1.6\eps$ square and putting it in the corresponding position in the core region of the triangle. To implement the oracle $C$ for a point inside the core region, we only need to query its corresponding point in $C_{\rect}$. For points outside the core region, we can output according to \cref{alg:base-instance}, without any query to $C_{\rect}$. 
Then, when we receive a small trichromatic triangle in $C$, which is always inside the core region according to \cref{lem:all-solutions-of-C-are-in-the-core}, we can compute the corresponding points of the small triangle in $C_{\rect}$ to obtain a solution of $C_{\rect}$ (again, without any query to $C_{\rect}$). 

Because the query complexity of $C_{\rect}$ is $2^{\Omega(n)}$ according to \cref{thm:previous-result},
we can easily obtain a lower bound of $2^{\Omega(n)}$ for the query complexity of the base instance $C$.
\begin{lemma}
    \label{lem:query-complexity-of-2D}
    Suppose $C$ is constructed by \cref{alg:base-instance} using a black-box of $C_{\rect}$ with side-length of $2^{-(n-3)}$.
    Then, it requires a query complexity of $2^{\Omega(n)}$ (using oracle $C$) to find a trichromatic region in $C$. 
\end{lemma}

\paragraph{Step 2: reduction from $C$ to $\Csym{k}$.} Recall our construction \cref{alg:approx-symmetric-sperner-kgeq3}.
We use a total number of $O(k)$ queries to the following functions that require oracle access to $C$:
\begin{itemize}[itemsep=0em, topsep=0.2em]
    \item $C(\vec{x})$: the coloring of the base $2${\rm D}-\Sperner instance;
    \item $\hat{C}_{\sf nn}^{\alpha}(\vec{x})$: the modified neighboring color; and 
    \item $\rela(\vec{x})$: the coordinate converter. 
\end{itemize}
Recall that $\hat{C}_{\sf nn}^{\alpha}(\vec{x})$ is defined as $C_{\sf nn}^{\alpha}(\vec{x})$ when $\vec{x}$ is not cold, and defined according to $C(\vec{x})$ otherwise (see \cref{eqn:modified-neighboring-color-symmetry}).
Recall that \cref{lem:poly-time-oracle-and-converter-for-symmetry}, where we prove that $\rela(\vec{x})$ and $C^{\alpha}(\vec{x})$ (when $\vec{x}$ is hot or warm) can be computed in polynomial-time with oracle access to $C$.
This implies that the 
last two functions, each $\rela(\vec{x})$ and $\hat{C}_{\sf nn}^{\alpha}(\vec{x})$, can be computed with a $\poly(n)$ queries to $C$. 
Therefore, we can construct the instances $\Csym{k}$ with a total number of $k\cdot \poly(n)$ queries to $C$. 

On the other hand, in the recovery algorithm \cref{alg:recover-2d-sol-symmetry}, we have oracle access to $C$ for the following process:
\begin{itemize}[itemsep=0em,topsep=0.5em]
    \item simulating \cref{alg:approx-symmetric-sperner-kgeq3} for the given solution $\vec{x}^{(1)}, \vec{x}^{(2)}, \vec{x}^{(3)}$ for $\Csym{k}$; and
    \item computing $\mathcal{N}(\vec{y}^{(i)}(\vec{x}^{(j)}))$ for each $i\in \{2,3,\dots,k-1\}$ and $j\in [3]$. 
\end{itemize}
According to earlier discussions, the simulation part can be done in a total number of $\poly(n,k)$ queries to $C$. 
Recall the definition of the neighborhood (\cref{def:neighbourhood}), where the neighbourhood of each point is an $\eps\times \eps$ square. 
Because $C$ is constructed by transforming each point in $C_{\rect}$ to a $1.6\eps\times 1.6\eps$ square, we only need to look at at most 8 points in $C_{\rect}$ to characterize each $\mathcal{N}(\cdot)$. 
Note that, according to \cref{alg:base-instance}, each $C_{\rect}(\cdot)$ can be computed by 1 query to $C$.
Hence, each step in \cref{alg:recover-2d-sol-symmetry} using $\mathcal{N}(\vec{y}^{(i)}(\vec{x}^{(j)}))$ can be done with a constant number of queries to $C$.
In total, we can recover a solution for the base instance $C$ from a solution to $\Csym{k}$ using a total number of $\poly(n,k)$ queries to $C$. 

Suppose the query complexity of finding a trichromatic region in $\Csym{k}$ is $f(n,k)$.
According to our previous discussion, we can use our construction algorithm \cref{alg:approx-symmetric-sperner-kgeq3} and recovery algorithm \cref{alg:recover-2d-sol-symmetry} to find a trichrmatic region in $C$ in a total number of $\poly(n,k)\cdot f(n,k) + \poly(n,k)$ queries to $C$. 
Because finding a trichromatic region in $C$ requires a query complexity of $2^{\Omega(n)}$ (\cref{lem:query-complexity-of-2D}), we have 
\begin{align*}
    \poly(n,k)\cdot f(n,k) + \poly(n,k) \geq 2^{\Omega(n)} \quad \Longrightarrow \quad f(n,k)  \geq 2^{\Omega(n)}/\poly(n,k)~.
\end{align*}

All of the above discussions can be concluded by the following lemma, which directly implies the query complexity lower-bound in \cref{thm:continuous-symmetry-main} together with \cref{lem:query-complexity-of-2D}. 

\begin{lemma}
    \label{lem:query-complexity-of-approx-kD}
    Suppose $\Csym{k}$ is constructed by \cref{alg:approx-symmetric-sperner-kgeq3} using a black-box of $C$ in \cref{lem:query-complexity-of-2D}.
    Then, it requires a query complexity of $2^{\Omega(n)}/\poly(n,k)$ (using oracle $\Csym{k}$) to find a trichromatic region in $\Csym{k}$. 
\end{lemma}

\section{Applications: Envy-Free Cake-Cutting}
\label{sec:cake-cutting}
In this section, we use the hardness result for \outofk Approximate Symmetric \Sperner to obtain a hardness result for the cake-cutting problem. 
In \Cref{sec:making-almost-everyone-happy}, we focus on the relaxation allowing some (or even most) agents to envy others.  Subsequently, we extend our hardness result to include the relaxation allowing redundant cuts \Cref{sec:three-agents-non-contiguous}.

\subsection{Making Almost Every Agent Envy-Free}\label{sec:making-almost-everyone-happy}

In this section, we assume that there are $p = k + 1$ agents and $k$ cuts. Additionally, we assume that each agent is allocated exactly one piece of the cake.

\begin{definition}[\mccc]
Let $\vec{x}$ be a $k$-cut. We say $\vec{x}$ is a \mccs{} if there exists a permutation $\pi: [p] \rightarrow [p]$ and a subset of agents $S \subseteq [p]$ consisting of at least three agents such that for each agent $d \in S$, it holds $u_d(\vec{x}, X_{\pi(d)}) + \eps \geq u_d(\vec{x}, X_{i})$ for any $i$. 
\end{definition}

Our approach involves constructing the utility function and preference list for the \mccs{} problem from the hard instance of \outofk Approximate Symmetric \Sperner (\Cref{thm:main}). Then, we show that if we can find a solution for the \mccs{} problem, the corresponding point representing the $k$-cut lies within a base simplex with at least three colors in the hard instance of \Cref{thm:main}, indicating the problem's PPAD-hardness. It is important to note that we use the same preference for each agent, denoted as $P$, and their utility as $u$.

\paragraph{Obtaining preference and utility from the simplex colors:} 
\begin{itemize}
    \item \textbf{Preference and Utility on discrete points:} Let $N = 2^n - 1$. For $\vec{x} \in \Delta^k_n$, we let $P(\vec{x}) = C(\vec{x})$, i.e.~the preference of agents for $k$-cut $\vec{x}$ is the color of point $\vec{x}$ in the hard instance of \Cref{thm:main}.  We first define a {\em pseudo-utility} $u'$, i.e.~the value of a piece depends on the entire $k$-cut instead of the equivalent class of the $k$-cut. We use $u'(\vec{x}, X_i)$ to denote the pseudo-utility of piece $X_i$. We let $u'(\vec{x}, X_{P(\vec{x})}) = 1/(2N)$. Also, if $X_i$ is empty, we let $u'(\vec{x}, X_i) = 0$. For all other pieces $X_j$, we let $u'(\vec{x}, X_j) = 1/(10k^2N)$. 
    \item \textbf{Interpolation:} Up to this point, we have only selected preferences and utilities for points whose coordinates are multiples of $1/N$, i.e.~nodes in the discretized simplex. Let $\vec{x}$ be a point within a base simplex with corners $\vec{x^{(0)}}, \vec{x^{(1)}}, \ldots, \vec{x^{(k)}}$. Since $\vec{x}$ is inside this base simplex, it can be written uniquely as a convex combination of its corners, i.e.~$\vec{x} = \sum_{i=0}^k \alpha_i \cdot \vec{x^{(i)}}$ such that $\sum_{i=0}^k \alpha_i = 1$. For piece $X_j$, we let $u'(\vec{x}, X_j) = \sum_{i=0}^k \alpha_i \cdot u'(\vec{x^{(i)}}, X^{(i)}_j)$. Finally, for an equivalent class $\eqClass{\vec{x}}$ we let $u(\eqClass{\vec{x}}, X_i) = u'(\vec{x}, X_i)$ where $\vec{x}$ is an arbitrary vector in this equivalent class. We will later (statement in \Cref{prt:first-symmetry} and proved in \Cref{lem:satisfying-the-first-property}) show that for two $k$-cuts that are equivalent, the pseudo-utility of the similar pieces is going to be equal anyway. 
Now, we prove that the utility that we defined has both the nonnegativity and Lipschitz conditions.
\end{itemize}

\begin{claim}\label{clm:majority-lipschitz-holds}
    The defined utility function $u$ 
    is Lipschitz.
\end{claim}
\begin{proof}
    Consider the $j$-th piece for two distinct points $\vec{x}, \vec{y} \in \Delta^k_n$. First, we have $\norm{\vec{x} - \vec{y}}{1} \geq 1/N$. Also, $|u(\eqClass{\vec{x}}, X_j) - u(\eqClass{\vec{y}}, Y_j)| \leq 1/(2N)$ since the utility of a piece cannot be larger than $1/(2N)$ by the definition of the utility function. Hence, we have $|u(\eqClass{\vec{x}}, X_j) - u(\eqClass{\vec{y}}, Y_j)| \leq \norm{\vec{x} - \vec{y}}{1}$. Now we prove the claim for two points within the same base simplex. Let $\vec{x}$ and $\vec{y}$ be two points within the base simplex with corners $\vec{z^{(0)}}, \vec{z^{(1)}}, \ldots, \vec{z^{(k)}}$. Thus, $\vec{x} = \sum_{i=0}^k \alpha_i \vec{z^{(i)}}$ and $\vec{y} = \sum_{i=0}^k \beta_i \vec{z^{(i)}}$ s.t. $\sum_{i=0}^k \alpha_i = \sum_{i=0}^k \beta_i = 1$. For the $j$-th piece, 
    \begin{align*}
        \left| u(\eqClass{\vec{x}}, X_j) - u(\eqClass{\vec{y}}, Y_j)\right| = \left| \sum_{i=0}^k \alpha_i u(\eqClass{\vec{z^{(i)}}}, X^{(i)}_j) - \sum_{i=0}^k \beta_i u(\eqClass{\vec{z^{(i)}}}, X^{(i)}_j) \right| &=  \left|\sum_{i=0}^k  u(\eqClass{\vec{z^{(i)}}}, X^{(i)}_j) \cdot (\alpha_i - \beta_i)\right|\\
        & \leq \sum_{i=0}^k  u(\eqClass{\vec{z^{(i)}}}, X^{(i)}_j) \cdot \left|\alpha_i - \beta_i \right|\\
        & \leq \frac{1}{2N} \sum_{i=0}^k \left| \alpha_i - \beta_i \right|.
    \end{align*}
    Since points $\vec{z^{(0)}}, \vec{z^{(1)}}, \ldots, \vec{z^{(k)}}$ are corners of a base simplex $\Delta^k_n$, without loss of generality we can assume that there exists a permutation $\pi \in \mathcal{S}_k$ such that $\vec{z^{(i)}} = \vec{z^{(0)}}+\sum_{\ell=1}^i 1/N \cdot \vec{e_{\pi(\ell)}}$. Then, we have 
    \begin{align*}
         \norm{\vec{x} - \vec{y}}{1} = \sum_{\ell=1}^k \left| \sum_{i=0}^k  \alpha_i z^{(i)}_\ell - \sum_{i=0}^k  \beta_i z^{(i)}_\ell\right| &= \sum_{\ell=1}^k \left| \sum_{i=0}^k   z^{(i)}_\ell (\alpha_i - \beta_i)  \right|\\
         & = \sum_{\ell=1}^k \left| \sum_{i=0}^k   z^{(i)}_{\pi(\ell)} (\alpha_i - \beta_i)  \right|\\
         & = \sum_{\ell=1}^k \left| \sum_{i=0}^{\ell - 1}   z^{(0)}_{\pi(\ell)} (\alpha_i - \beta_i) +  \sum_{i=\ell}^{k}   (z^{(0)}_{\pi(\ell)} + \frac{1}{N})\cdot (\alpha_i - \beta_i) \right|\\
         & = \sum_{\ell=1}^k \left|z^{(0)}_{\pi(\ell)} \sum_{i=0}^{k}    (\alpha_i - \beta_i) +  \frac{1}{N}\sum_{i=\ell}^{k}  (\alpha_i - \beta_i) \right|\\
         & = \frac{1}{N} \sum_{\ell=1}^k \left| \sum_{i=\ell}^{k}  (\alpha_i - \beta_i) \right|,
    \end{align*}
    where the last equality is followed by $\sum_{i=0}^k (\alpha_i - \beta_i) =0$. Further,
    \begin{align*}
        \sum_{\ell=0}^k \left| \alpha_\ell - \beta_\ell \right| = \sum_{\ell = 0}^k \left| \sum_{i=\ell}^k (\alpha_i - \beta_i) - \sum_{i=\ell+1}^k (\alpha_i - \beta_i)\right| &\leq \sum_{\ell = 0}^k \left| \sum_{i=\ell}^k (\alpha_i - \beta_i) \right|  + \left| \sum_{i=\ell+1}^k (\alpha_i - \beta_i)\right|\\
        & \leq \sum_{\ell = 0}^k \left| 2\sum_{i=\ell}^k (\alpha_i - \beta_i) \right|\\
        & = 2\sum_{\ell = 1}^k \left| \sum_{i=\ell}^k (\alpha_i - \beta_i) \right|
    \end{align*}
    where the last equality is followed by the fact that for $\ell = 0$, we have $\sum_{i=\ell}^k (\alpha_i - \beta_i) =0$. Combining the above bounds, we obtain
    \begin{align*}
         \left| u(\eqClass{\vec{x}}, X_j) - u(\eqClass{\vec{y}}, Y_j)\right| \leq \frac{1}{2N} \sum_{i=0}^k \left| \alpha_i - \beta_i \right| \leq \frac{1}{N}\sum_{\ell = 1}^k \left| \sum_{i=\ell}^k \alpha_i - \beta_i \right| = \norm{\vec{x} - \vec{y}}{1},
    \end{align*}
    which completes the proof for two points $\vec{x}$ and $\vec{y}$ within the same base simplex. For two points $\vec{x}$ and $\vec{y}$ that are not in the same base simplex, suppose that there are $r > 1$ base simplices between them (including simplices that $\vec{x}$ and $\vec{y}$ belongs to) to reach $\vec{y}$ from $\vec{x}$ in the closest path. Hence, there exist points $\vec{z^{(1)}}, \vec{z^{(2)}}, \ldots, \vec{z^{(r-1)}}$ such that point $\vec{z^{(i)}}$ is on both the $i$-th and $(i+1)$-th simplices and
    \begin{align*}
        \norm{\vec{x} - \vec{y}}{1} = \sum_{i=1}^{r} \norm{\vec{z^{(i-1)}}-\vec{z^{(i)}}}{1},
    \end{align*}
    if we define $\vec{z^{(0)}} = \vec{x}$ and $\vec{z^{(r)}} = \vec{y}$. Moreover, since $\vec{z^{(i)}}$ and $\vec{z^{(i+1)}}$ are within the same base simplex, we have $\left|u(\eqClass{\vec{z^{(i)}}}, Z^i_j) - u(\eqClass{\vec{z^{(i + 1)}}}, Z^i_j)\right| \leq  \norm{\vec{z^{(i)}}-\vec{z^{(i+1)}}}{1}$. Therefore,
    \begin{align*}
         \left| u(\eqClass{\vec{x}}, X_j) - u(\eqClass{\vec{y}}, Y_j)\right| &=  \left| \sum_{i=1}^r u(\eqClass{\vec{z^{(i)}}}, Z^{(i)}_j) - u(\eqClass{\vec{z^{(i-1)}}}, Z^{(i-1)}_j) \right|\\
         & \leq  \sum_{i=1}^r \left| u(\eqClass{\vec{z^{(i)}}}, Z^{(i)}_j) - u(\eqClass{\vec{z^{(i-1)}}}, Z^{(i-1)}_j) \right|\\
         & \leq \sum_{i=1}^r \norm{\vec{z^{(i)}} - \vec{z^{(i-1)}}}{1}\\
         & =  \norm{\vec{x} - \vec{y}}{1},
    \end{align*}
    which concludes the proof.
\end{proof}

\begin{claim}\label{clm:majority-nonnegativity-holds}
    The defined utility function $u$ for the \mccs{} problem satisfies the nonnegativity condition.
\end{claim}
\begin{proof}
    For points $\vec{x} \in \Delta^k_n$, the utility of a piece is equal to zero if and only if it is an empty piece according to the way that we define the utility function. Now let $\vec{x} \in \Delta^k$ and $\vec{x^{(0)}}, \vec{x^{(1)}}, \ldots, \vec{x^{(k)}}$ be the corner of base simplex that it belongs to. So we have $\vec{x} = \sum_{i=0}^k \alpha_i \vec{x^{(i)}}$ s.t. $\sum_{i=0}^k \alpha_i = 1$.

    Let $i_1, \ldots, i_r$ be the indices such that the $j$-th piece of $\vec{x^{(i_1)}}, \ldots, \vec{x^{(i_r)}}$ is empty, i.e.~for each $i \in \{i_1, \ldots, i_r\}$, we have $x^{(i)}_j = x^{(i)}_{j-1}$ since the $j$-th piece is empty. Moreover,
    \begin{align*}
        x_j - x_{j-1} = \sum_{i = 0}^k \alpha_i x^{(i)}_j - \sum_{i = 0}^k \alpha_i x^{(i)}_{j-1} &= \sum_{i \in \{i_1, \ldots, i_r\}} \alpha_i x^{(i)}_j + \sum_{i \notin \{i_1, \ldots, i_r\}} \alpha_i x^{(i)}_j - \sum_{i \in \{i_1, \ldots, i_r\}} \alpha_i x^{(i)}_{j-1} - \sum_{i \notin \{i_1, \ldots, i_r\}} \alpha_i x^{(i)}_{j-1}\\
        &= \left(\sum_{i \in \{i_1, \ldots, i_r\}} \alpha_i (x_j^{(i)} - x_{j-1}^{(i)}) \right) + \left( \sum_{i \notin \{i_1, \ldots, i_r\}} \alpha_i (x_j^{(i)} - x_{j-1}^{(i)})\right).
    \end{align*}
    Note that in second term, we have $x^{(i)}_j \neq x^{(i)}_{j-1}$ which implies that if there exists a non-zero $\alpha_i$ in the second term, then $x_j - x_{j-1} > 0$. Thus, $j$-th piece is empty if and only if there exists no $i \in \{0, \ldots, k\} \setminus \{i_1, \ldots, i_r\}$ where $\alpha_i > 0$. Further,
    \begin{align*}
        u(\eqClass{\vec{x}}, X_j) = \sum_{i=0}^k \alpha_i \cdot u(\eqClass{\vec{x^{(i)}}}, X^{(i)}_j) &= \left(\sum_{i \in \{i_1, \ldots, i_r\}} \alpha_i \cdot u(\eqClass{\vec{x^{(i)}}}, X^{(i)}_j) \right) + \left( \sum_{i \notin \{i_1, \ldots, i_r\}} \alpha_i \cdot u(\eqClass{\vec{x^{(i)}}}, X^{(i)}_j) \right)\\
        & = \left( \sum_{i \notin \{i_1, \ldots, i_r\}} \alpha_i \cdot u(\eqClass{\vec{x^{(i)}}}, X^{(i)}_j) \right),
    \end{align*}
    where we have the last equality because $u(\eqClass{\vec{x^{(i)}}}, X^{(i)}_j) = 0$ for $i \in \{i_0, \ldots, i_r\}$. On the other hand, $u(\eqClass{\vec{x^{(i)}}}, X^{(i)}_j) > 0$  for $i \in \{i_1, \ldots, i_r\}$. Finally, since $j$-th piece is empty if and only if there exists no $i \in \{0, \ldots, k\} \setminus \{i_1, \ldots, i_r\}$ where $\alpha_i > 0$, we have the utility of piece $X_j$ is equal to zero if and only if it is empty.
\end{proof}

The way that we defined our utility functions enables us to show that if corners of a base simplex are colored with at most two colors, then none of the points inside this base simplex is a solution for the \mccs{} problem. As a corollary, if we can find a \mccs{}, the corresponding point in the simplex must be within a base simplex that has at least three colors.

\begin{lemma}\label{lem:envy-free-colors-base}
    Let $\vec{x} \in \Delta^k $ be a $k$-cut that is within a base simplex with corners $\vec{x^{(0)}}, \vec{x^{(1)}}, \ldots, \vec{x^{(k)}} \in \Delta^k_n$. If \begin{align*}
        \left|\{P(\vec{x^{(0)}}), P(\vec{x^{(1)}}), \ldots, P(\vec{x^{(k)}})\}\right| \leq 2,
    \end{align*} 
    then $\vec{x}$ is not a solution for \mccs{} problem when $\eps \leq 1/(10N)$.
\end{lemma}

\begin{proof}
    We can write $\vec{x} = \sum_{i=0}^k \alpha_i \vec{x^{(i)}}$ s.t. $\sum_{i=0}^k \alpha_i=1$. Let $\mathcal{P} = \{P(\vec{x^{(0)}}), P(\vec{x^{(1)}}), \ldots, P(\vec{x^{(k)}})\}$ and $\overline{\mathcal{P}} = \{0, 1, \ldots, k\} \setminus \mathcal{P}$. For $j' \in \overline{\mathcal{P}}$, we have $u(\eqClass{\vec{x^{(i)}}}, X^{(i)}_{j'}) \leq 1/(10k^2N)$ by the definition of our utility function. Thus, 
    \begin{align*}
        u(\eqClass{\vec{x}}, X_{j'}) = \sum_{i=0}^k \alpha_i \cdot u(\eqClass{\vec{x^{(i)}}}, X^{(i)}_{j'}) \leq \frac{1}{10k^2 N} \sum_{i=0}^k \alpha_i =  \frac{1}{10k^2 N}.
    \end{align*}
    On the other hand, since $|\mathcal{P}| \leq 2$, by pigeonhole principle there exists a $j \in \mathcal{P}$ such that $\sum_{i \text{ and } P(\vec{x^{(i)}}) = j} \alpha_i \geq 1/2$ which implies that 
    \begin{align*}
        u(\eqClass{\vec{x}}, X_{j}) = \sum_{i=0}^k \alpha_i \cdot u(\eqClass{\vec{x^{(i)}}}, X^{(i)}_{j}) \geq \sum_{i \text{ and } P(\vec{x^{(i)}}) = j} \alpha_i \cdot u(\eqClass{\vec{x^{(i)}}}, X^{(i)}_{j}) &=   \sum_{i \text{ and } P(\vec{x^{(i)}}) = j} \alpha_i \cdot u(\eqClass{\vec{x^{(i)}}}, X^{(i)}_{P(\vec{x^{(i)}})}) \\ 
        & = \frac{1}{2N} \cdot \sum_{i \text{ and } P(\vec{x^{(i)}}) = j} \alpha_i\\
        &\geq \frac{1}{4N},
    \end{align*}
    where the last equality follows by the fact that $u(\eqClass{\vec{x^{(i)}}}, X^i_{P(\vec{x^{(i)}})}) = 1/(2N)$. For any $j' \in \overline{\mathcal{P}}$, by combining the last two bounds for the specific $j$ in the last bound, we obtain
    \begin{align*}
        u(\eqClass{x}, X_{j'}) + \frac{1}{10N} \leq \frac{1}{10k^2 N} + \frac{1}{10N} < \frac{1}{4N} \leq u(\eqClass{\vec{x}}, X_j).
    \end{align*}
    Therefore, when $\eps \leq 1/(10N)$, there exists a piece in $\mathcal{P}$ whose utility is greater than that of pieces in $\overline{\mathcal{P}}$ by a margin larger than $\eps$. Consequently, as $|\mathcal{P}| \leq 2$, there are at most two pieces where agents do not envy each other if we ignore the $\eps$ gap in their utilities. This implies that $\vec{x}$ cannot be a solution for the \mccs{} problem when $\eps \leq 1/(10N)$.
\end{proof}

\subsubsection*{Properties of the Utility Function} 

Note that when constructing utilities and preferences, we must satisfy several properties to accurately model our cake-cutting problem. The first property is the {\em boundary property}, which means that for a point $\vec{x}$, if $X_i$ represents an empty piece, then the utility of that piece should be equal to zero. It is easy to see that this property is satisfied due to \Cref{clm:majority-nonnegativity-holds}.

\begin{property}[Boundary]\label{prt:boundry-condition}
    Let $\vec{x} \in \Delta^k$ be a $k$-cut and $X = (X_0, X_1, \ldots, X_k)$ be its corresponding pieces of cake. If $X_i$ is an empty piece of cake, then $u(\eqClass{\vec{x}}, X_i) = 0$ and $P(\vec{x}) \neq i$. 
\end{property}

The second property that needs to be satisfied is that the facets of the simplex should resemble one another under permutations of colors. This property is referred to as the {\em symmetry property}. To see why such a property is necessary, let us consider the case with 3 agents and 2 cuts. Consider two distinct vectors, $\vec{x} = (0, 0.5)$ and $\vec{y} = (0.5, 1)$, which correspond to two different sets of 2-cuts. It is important to note that the pieces obtained from $\vec{x}$ and $\vec{y}$ are exactly alike, but their order differs, i.e., $X_1 = Y_0$ and $X_2 = Y_1$. Consequently, it must hold that $u'(\vec{x}, X_1) = u'(\vec{y}, Y_0)$ and $u'(\vec{x}, X_2) = u'(\vec{y}, Y_1)$. From the preference viewpoint, if we have $P(\vec{x}) = 1$, then $P(\vec{y}) = 0$ due to this similarity. As a result of this property, for two $k$-cuts that are equivalent, the utility of a similar piece is equal.

\begin{property}[Symmetry]\label{prt:first-symmetry}
Let $\vec{x}, \vec{y} \in \Delta^k$ be two $k$-cuts and $X = (X_0, X_1, \ldots, X_k)$ and $Y = (Y_0, Y_1, \ldots, Y_k)$ be their corresponding pieces of cake. Suppose that there exists a permutation $\pi: \{0, \ldots k\} \rightarrow \{0, \ldots k\}$ such that $X_i = Y_{\pi(i)}$ for all $i \in \{0, \ldots, k\}$. Then, $u'(\vec{x}, X_i) = u'(\vec{y}, Y_{\pi(i)})$ for each $i$.
\end{property}

Any coloring of simplex that we use to construct our utility function and preference must satisfy the symmetry property (\Cref{prt:first-symmetry}). It is worth noting that when aiming to achieve $\eps$-approximate envy-free cake cutting for all agents, similar to \cite{DBLP:journals/ior/DengQS12}, this property loses its significance. This is because in the facets of the simplex, there exists no solution, and therefore satisfying this property cannot hurt the hard instance. On the other hand, in the \mccs, there indeed exists a solution on the facets of the simplex. 
Therefore, we must be careful in how we color the simplex to ensure it satisfies the symmetry property. In \Cref{lem:satisfying-the-first-property}, we show that our utility function that is obtained from the hard instance in \Cref{thm:main} satisfies the symmetry property.

\begin{claim}\label{clm:swapping-empty-peice}
    Let $\vec{x}, \vec{y} \in \Delta^k_n$ be two $k$-cuts in a $k$-dimensional simplex that satisfies the symmetric condition in \Cref{def:approx-symmetric-kd-sperner}. Let $i \in [k]$ and suppose that $X_j = Y_j$ for $j \in \{0, \ldots, k\} \setminus \{i-1, i\}$. Moreover, let $X_{i-1} = Y_i$, $X_i = Y_{i-1}$, and $X_{i-1}$ be an empty piece. Then, $P(\vec{y}) = P(\vec{x})$ if $P(\vec{x}) \neq i$, and $P(\vec{y})=i-1$, otherwise.
\end{claim}
\begin{proof}
    Note that $\vec{x}$ and $\vec{y}$ satisfy the symmetric condition in \Cref{def:approx-symmetric-kd-sperner} by the definition. This is because all pieces except the $i$-th and $(i-1)$-th pieces are exactly similar, with one of them being empty. Consequently, the claim holds true.
\end{proof}

\begin{lemma}\label{lem:satisfying-the-first-property}
The proposed utility function $u$ for the \mccs{} problem satisfies the symmetric condition.
\end{lemma}
\begin{proof}
Note that it suffices to show that the utility function satisfies the symmetry property solely for points in $\Delta^k_n$. Subsequently, as the pseudo-utilities are a convex combination of the corners of base simplices, it holds for all points in $\Delta^k$. Further, it is enough to show that if $P(\vec{x}) = i$, then $P(\vec{y}) = \pi(i)$. This is because the pseudo-utility of empty pieces is zero in both cuts, the pseudo-utility of the preferred piece is $1/(2N)$, and all other non-empty pieces have a pseudo-utility of $1/(10k^2 N)$ for points in $\Delta^k_n$.

Let $\vec{x}, \vec{y} \in \Delta^k_n$ where there exists a permutation $\pi: \{0, \ldots k\} \rightarrow \{0, \ldots k\}$ such that $X_i = Y_{\pi(i)}$ for all $i \in \{0, \ldots, k\}$. First, by the boundary property, the preference of agents cannot be an empty piece. Thus, we only need to show that for permutations that map the non-empty pieces from one cut to the other cut, the preference of agents does not change. Let $X_{i_0}, X_{i_1}, \ldots, X_{i_r}$ be the non-empty pieces of $X$ such that $i_0 < i_1 < \ldots < i_r$. Similarly, let $Y_{j_0}, Y_{j_1}, \ldots Y_{j_r}$ be the non-empty pieces of $Y$ such that $j_0 < j_1 < \ldots < j_r$. Note that for permutation $\pi$ such that $\pi(i_d) = j_d$ and arbitrary maps empty pieces to each other, we have $X_i = Y_{\pi(i)}$ for all $i \in \{0, \ldots, k\}$. Now, we need to show that if $P(\vec{x}) = i_d$, then $P(\vec{y}) = \pi(i_d)$ since empty pieces cannot be a preferred piece.

Let $\vec{z} \in \Delta^k_n$ to be a point in the simplex which has this property that $Z_{d}$ is empty if $d > r$ and otherwise, $Z_d = X_{i_d} = Y_{j_d}$. Basically, empty pieces are the rightmost pieces in $Z$. We claim that if $P(\vec{x}) = i_{d}$ for some $d \leq r$, then $P(\vec{z}) = d$. We prove this claim step by step. In each step, we swap two consecutive pieces where the left piece is empty and the right piece is not empty. Note that by doing this process a finite number of times, we will finally end up with $Z$ where all empty pieces are in the rightmost positions. Let $X^{(0)}, X^{(1)}, \ldots ,X^{(t)}$ be the different configurations of the pieces that we see in this process and $\vec{x^{(0)}}, \vec{x^{(1)}},\ldots, \vec{x^{(t)}}$ be their corresponding $k$-cuts vectors. So we have $X^{(0)} = X$ and $X^{(t)} = Z$. Note that $P(\vec{x^{(0)}}) = i_d$ and every time that we convert $X^{(i)}$ to $X^{{(i+1)}}$, if we swap $P(\vec{x^{(i)}})$ and an empty piece, then we have $P(\vec{x^{{(i+1)}}}) = P(\vec{x^{(i)}}) - 1$ by \Cref{clm:swapping-empty-peice}. We call this type of swap as {\em critical swap}. By the definition of $Z$, during this process, we have exactly $i_d - d$ critical swaps. Therefore, we have
\begin{align*}
    P(\vec{z}) = P(\vec{x^{(t)}}) = P(\vec{x^{(0)}}) - (i_d - d) = P(\vec{x}) - (i_d - d) = d,
\end{align*}
which completes the proof of the claim. With a similar argument, we can show that if $P(\vec{z}) = d$ for some $d \leq r$, then $P(\vec{y}) = j_d$. Now assume that $P(\vec{x}) = i_d$. By the above claim, we have $P(\vec{z}) = d$ and consequently, $P(\vec{y}) = j_d$. Since $\pi(i_d) = j_d$, we have $P(\vec{y}) = \pi(i_d)$ which conclude the proof.
\end{proof}

Now we are ready to prove our main theorem for the \mccs{} problem.

\begin{theorem}
    For any constant $\eps$ such that $k < \log^{1-\delta}(1/\eps)$ for some constant $\delta > 0$, the problem of finding the \mccs{} for $k$ cuts and $k + 1$ agents is \PPAD-complete. Further, it requires a query complexity of $(1/\eps)^{\Omega(1/k)} / \polylog(1/\eps)$.
\end{theorem}
\begin{proof}
    We assume that the utilities of all agents are as defined at the beginning of this subsection. By \Cref{clm:majority-lipschitz-holds} and \Cref{clm:majority-nonnegativity-holds}, the Lipschitz and nonnegativity conditions hold for utility functions. Further, by \Cref{clm:majority-nonnegativity-holds} and \Cref{lem:satisfying-the-first-property}, we have both the boundary and symmetry properties (\Cref{prt:boundry-condition} and \Cref{prt:first-symmetry}). Hence, the utility function is a valid utility function for the cake-cutting problem.

    We defined our utilities and preferences based on the circuit $C: \Delta^k_n \rightarrow \{0, \ldots, k\}$ that is constructed in \Cref{thm:main}. We know that $\eps$-approximate 3-out-of-$p$-envy-free cake cut problem is in \PPAD{} by \cite{DBLP:journals/ior/DengQS12}. If we can find an allocation that is a solution for the \mccs, we can convert the allocation to its corresponding point in the simplex. Further, by \Cref{lem:envy-free-colors-base}, if we can find a solution for \mccs{} when $\eps \leq 1/(10N)$, then the corresponding point in the simplex is in a trichromatic base simplex. Thus, we can find a solution for the \outofk Approximate Symmetric \Sperner problem. By \Cref{thm:main}, \outofk Approximate Symmetric \Sperner is \PPAD-hard, and therefore, the \mccs{} problem is \PPAD-hard.

    To get the query complexity, consider $N = 2^{4kn}$. Thus, we have $nk = \Theta(\log N)$ and $n = \Theta(\log N/ k)$. Plugging $k < \log^{1-\delta}(N)$, we have $n = \polylog(N)$ and $k = \poly(n)$. Also, we have $N = \Theta(1/\eps)$. Therefore, by \Cref{thm:main}, the query complexity is at least $(1/\eps)^{\Omega(1/k)}/\polylog(1/\eps)$.
\end{proof}

\subsection{Beyond Connected Pieces}\label{sec:three-agents-non-contiguous}

In this section, we consider both relaxations: (1) finding three agents who do not envy others, and (2) allowing agents to have more than one piece. More formally, we have $p$ agents where $p \leq k + 1$, and our objective is to find a bundling of $k + 1$ pieces into $p$ bundles.  The goal is to allocate each agent one of these bundles in such a way that the allocation is $\epsilon$-approximate envy-free for at least three agents.

\poutofqcakecut*

Our approach for this subsection is very similar to the previous one. We first extend the utility function for multiple pieces. Then, we demonstrate that if we can find a solution for the $\eps$-approximate 3-out-of-$p$-envy-free cake cut, the corresponding point representing the $k$-cut lies within a base simplex with at least three colors in the hard instance of \Cref{thm:main}, indicating the problem's PPAD-hardness.

\paragraph{Extending the utility function to multiple pieces:} We let $u'(\vec{x}, B)$ be equal to the sum of the utilities of the pieces that belong to bundle $B$. Since we defined the pseudo-utility function the same as the previous section, for two $k$-cuts that are equivalent, the pseudo-utility of the similar pieces is going to be equal with a similar argument (formalized in \Cref{prt:first-symmetric-2} and \Cref{clm:first-symmetric-hold-2}). Finally, we let $u(\eqClass{\vec{x}}, B)$ be equal to the sum of the utilities of the pieces that belong to bundle $B$. Since the utility function is defined the same way as the previous subsection, we have both the Lipschitz and nonnegativity conditions. We repeat the statements for the usage of this subsection.

\begin{claim}\label{clm:three-agent-lipschitz-holds}
    The defined utility function $u$ for the $\eps$-approximate 3-out-of-$p$-envy-free cake cut problem is Lipschitz.
\end{claim}

\begin{claim}\label{clm:three-agent-nonnegativity-holds}
    The defined utility function $u$ for the $\eps$-approximate 3-out-of-$p$-envy-free cake cut problem satisfies the nonnegativity condition.
\end{claim}

We now show that for the extended utility function, any cut that is an $\eps$-approximate 3-out-of-$p$-envy-free cake cut implies that the corresponding point in the simplex must lie within a trichromatic base simplex.

\begin{lemma}\label{lem:envy-free-colors-base-2}
    Let $\vec{x} \in \Delta^k $ be a $k$-cut that is within a base simplex with corners $\vec{x^{(0)}}, \vec{x^{(1)}}, \ldots, \vec{x^{(k)}} \in \Delta^k_n$. If \begin{align*}
        \left|\{P(\vec{x^{(0)}}), P(\vec{x^{(1)}}), \ldots, P(\vec{x^{(k)}})\}\right| \leq 2,
    \end{align*} 
    then, there exists no bundling for the point $\vec{x}$ that serves as a solution for the $\eps$-approximate 3-out-of-$p$-envy-free cake cut problem when $\epsilon \leq 1/(10N)$.
\end{lemma}
\begin{proof}
    We use almost the same approach as proof of \Cref{lem:envy-free-colors-base}. We can write $\vec{x} = \sum_{i=0}^k \alpha_i \vec{x^{(i)}}$ s.t. $\sum_{i=0}^k \alpha_i=1$.
    Let $\mathcal{P} = \{P(\vec{x^{(0)}}), P(\vec{x^{(1)}}), \ldots, P(\vec{x^{(k)}})\}$ and $\overline{\mathcal{P}} = \{0, 1, \ldots, k\} \setminus \mathcal{P}$. For $j' \in \overline{\mathcal{P}}$, with the same argument as the proof of \Cref{lem:envy-free-colors-base}, we have $u(\eqClass{\vec{x}}, X_{j'}) \leq 1/(10k^2 N)$. Therefore, if a bundle $B_{d'}$ only contains pieces that are in $\overline{\mathcal{P}}$, we have that $u(\eqClass{\vec{x}}, B_{d'}) \leq (k+1)/(10k^2N)$ since the bundle has at most $k+1$ of these pieces and each of them has a utility of at most $1/(10k^2 N)$.

    On the other hand, since $|\mathcal{P}| \leq 2$, by pigeonhole principle there exists a $j \in \mathcal{P}$ such that $\sum_{i \text{ and } P(\vec{x^{(i)}}) = j} \alpha_i \geq 1/2$ which implies that $u(\vec{x}, X_{j}) \geq 1/(4N)$ for the same $j$, again using the same argument as \Cref{lem:envy-free-colors-base}. Now consider bundle $B_d$ that contains this piece. We have that $u(\eqClass{\vec{x}}, B_d) \geq 1/(4N)$.

    As $|\mathcal{P}| \leq 2$, at most two bundles can have a utility larger than or equal to $1/(4N)$. By the above argument, there exists a bundle with a utility of at least $1/(4N)$. Additionally, there exist $p - 2$ bundles, each with a utility of at most $(k+1)/(10k^2N)$. Hence, when $\eps \leq 1/(10N)$, this implies that $\vec{x}$ cannot be a solution for the $\eps$-approximate 3-out-of-$p$-envy-free cake cut problem, as at least $p-2$ agents would envy others.
\end{proof}

We need to satisfy the boundary and symmetry properties when we have bundling similar to that in \Cref{sec:making-almost-everyone-happy}. We redefine these two properties for the $\eps$-approximate 3-out-of-$p$-envy-free cake cut problem. Regarding the boundary property, any bundle of a $k$-cut vector that is empty must have a utility equal to zero.

\begin{property}[Boundary]\label{prt:boundary-condition-2}
    Let $\vec{x}$ be a $k$-cut and $\mathcal{B} = \{B_0, B_1, \ldots, B_{p-1} \}$ be a bundling of its pieces into $p$ bundles. If $B_i$ is empty, then $u(\vec{x}, B_i) = 0$. 
\end{property}

\begin{claim}\label{clm:boundary-condition-hold-2}
    The proposed utility function $u$ 
    satisfies the boundary property.    
\end{claim}

\begin{proof}
    Let $B_i$ be an empty bundle consisting of pieces $ X_{i_1}, \ldots, X_{i_r}$. Since $B_i$ is empty, all pieces $X_{i_1}, \ldots, X_{i_r}$ are empty pieces. Thus,
    \begin{align*}
        u(\eqClass{\vec{x}}, B_i) = \sum_{j=1}^r u(\eqClass{\vec{x}}, X_{i_j}) = 0,
    \end{align*}
    where the last equality is followed by \Cref{clm:three-agent-nonnegativity-holds}.
\end{proof}

The symmetry property is also very similar to the symmetry property (\Cref{prt:first-symmetry}) in \Cref{sec:making-almost-everyone-happy}, with slight modifications to suit our application in this section.

\begin{property}[symmetry]\label{prt:first-symmetric-2}
    Let $\vec{x}$ and $\vec{y}$ be two $k$-cuts. Suppose that there exists a permutation $\pi: \{0, \ldots, k\} \rightarrow \{0, \ldots, k\}$ such that $X_i = Y_{\pi(i)}$ for all $i$. Let $\mathcal{B}^{(x)} = \{B^{(x)}_0, \ldots, B^{(x)}_{p-1} \}$ and $\mathcal{B}^{(y)} = \{B^{(y)}_0, \ldots, B^{(y)}_{p-1} \}$ be bundling of the pieces $X$ and $Y$, respectively. Suppose that $B^{(x)}_i = \{X_{i_1}, \ldots, X_{i_r}  \}$ and $B^{(y)}_j = \{Y_{\pi(i_1)}, \ldots, Y_{\pi(i_r)} \}$. Then, $u'(\vec{x}, B^{(x)}_i) = u'(\vec{y}, B^{(y)}_j)$.
\end{property}

\begin{claim}\label{clm:first-symmetric-hold-2}
    The proposed utility function $u$ for the $\eps$-approximate 3-out-of-$p$-envy-free cake cut problem satisfies the symmetry property.    
\end{claim}

\begin{proof}
    Let $\vec{x}$ and $\vec{y}$ be two $k$-cuts with the condition in the premise of \Cref{prt:first-symmetric-2}. Since we use the exact the same utility function as the \mccs{} problem, by \Cref{lem:satisfying-the-first-property}, we have $u'(\vec{x}, X_i) = u'(\vec{y}, Y_{\pi(i)})$. Therefore, we can conclude the proof since
    \begin{align*}
        u'(\vec{x}, B^{(x)}_i) = \sum_{d=1}^r u'(\vec{x}, X_{i_d}) = \sum_{d=1}^r u'(\vec{y}, Y_{\pi(i_d)}) = u'(\vec{y}, B^{(y)}_j). & \qedhere
    \end{align*}
\end{proof}

Now we are ready to prove our main theorem for the $\eps$-Approximate 3-out-of-$p$-envy-free cake cut problem.

\cakecuttingtheorem*

\begin{proof}

We assume that the utilities of all agents are as defined at the beginning of this subsection. By \Cref{clm:three-agent-lipschitz-holds} and \Cref{clm:three-agent-nonnegativity-holds}, the Lipschitz and non-negativity conditions for utility functions are satisfied. Moreover, according to \Cref{clm:boundary-condition-hold-2} and \Cref{clm:first-symmetric-hold-2}, we observe that both the boundary and symmetry properties (\Cref{prt:boundary-condition-2} and \Cref{prt:first-symmetric-2}) hold. Hence, the utility function is a valid utility function for the cake-cutting problem.

We have defined our utilities and preferences based on the circuit $C: \Delta^k_n \rightarrow \{0, \ldots, k\}$ constructed in \Cref{thm:main}. We know that $\eps$-approximate 3-out-of-$p$-envy-free cake cut problem is in \PPAD{} by \cite{DBLP:journals/ior/DengQS12}. Moreover, as stated in \Cref{lem:envy-free-colors-base-2}, if a solution for the $\eps$-approximate 3-out-of-$p$-envy-free cake cut problem is found when $\epsilon \leq 1/(10N)$, then the corresponding point in the simplex lies within a trichromatic base simplex. Thus, we can solve the \outofk Approximate Symmetric \Sperner problem. Since \outofk Approximate Symmetric \Sperner is proven to be \PPAD-hard by \Cref{thm:main}, it follows that the $\eps$-approximate 3-out-of-$p$-envy-free cake cut problem is also \PPAD-hard.

To get the query complexity, consider $N = 2^{4kn}$. Thus, we have $nk = \Theta(\log N)$ and $n = \Theta(\log N/ k)$. Plugging $k < \log^{1-\delta}(N)$, we have $n = \polylog(N)$ and $k = \poly(n)$. Also, we have $N = \Theta(1/\eps)$. Therefore, by \Cref{thm:main}, the query complexity is at least $(1/\eps)^{\Omega(1/k)}/\polylog(1/\eps)$.
\end{proof}

\section*{Acknowledgement}
This research was partially supported by NSF grants CCF-2112824, CCF-2209520, and CCF-2312156, ONR under award N000142212771, a Stanford Graduate Fellowship, Dantzig-Lieberman Operations Research Graduate Fellowship, and the David and Lucile Packard Foundation Fellowship.

\printbibliography

\appendix

\section{Deferred Proofs}
\label{app:omitted-new}
\subsection{Proof of \cref{lem:all-solutions-of-C-are-in-the-core}}
\label{proof:all-solutions-of-C-are-in-the-core}

\repi*

    Let $\eps=2^{-n}, \eps_0=1.6\eps$.
    Let $\mathcal{D}=\settxt{(x,y,z)\in \Delta^2: y\in 0.5\pm 0.1, z\in 0.2\pm 0.1}$ denote the core region.
    We are going to show that $\vec{x}^{(1)}, \vec{x}^{(2)}, \vec{x}^{(3)}\in \mathcal{D}$.
    We prove this fact by contradiction. 
    W.l.o.g., we suppose that $\vec{x}^{(1)}$ violates this condition.
    Since $\normtxt{\vec{x}^{(i)}-\vec{x}^{(j)}}{\infty}\leq \eps$, all of them are inside the following subset:
    \begin{align*}
        \overline{\mathcal{D}} = \set{\vec{y}\in \Delta^2: y_2\notin [0.4+\eps,0.6-\eps] \text{ or } y_3\notin [0.1+\eps,0.3-\eps]} ~.
    \end{align*}
    According to \cref{alg:base-instance} and \cref{def:rect-sperner}, we have 
    \begin{align*}
        \forall \vec{y}\in \overline{\mathcal{D}}, \quad C(\vec{y}) = \begin{cases}
            \hfill 1 \hfill & \text{if $y_2\leq 0.5$ and $y_3\leq 0.1$,} 
            \\
            \hfill 2 \hfill & \text{if $y_2> 0.5$ and $y_3\leq 0.1$,}\\
            \hfill 1 \hfill & \text{if $y_2\in [0.4,0.4+\eps_0)$ and $y_3\in [0.1, 0.1+\eps_0)$, }\\
            \hfill 2 \hfill & \text{if $y_2\in [0.6-\eps_0,0.6)$ and $y_3\in [0.1, 0.1+\eps_0)$, } \\
            1\text{ or }2 & \text{if $y_2\in [0.4+\eps_0, 0.6-\eps_0)$ and $y_3\in [0.1, 0.1+\eps]$,} \\
            \hfill 3 \hfill & \text{otherwise.}
        \end{cases}
    \end{align*}
    Next, we discuss several cases on where $\vec{x}^{(1)}\notin D$ is to establish the fact.
    \begin{enumerate}[itemsep=0em,topsep=0.5em]
        \item If $x^{(1)}_2\leq 0.4$, we have $C(\vec{x}^{(1)})\in \{1,3\}$. However, as any $\vec{y}\in \overline{\mathcal{D}}$ such that $C(\vec{y})=2$ has $y_2\geq 0.4+\eps_0>0.4+\eps$, the triangle formed by $\vec{x}^{(1)},\vec{x}^{(2)},\vec{x}^{(3)}$ cannot be trichromatic. 
        \item If $x^{(1)}_2\geq 0.6$, we have $C(\vec{x}^{(1)})\in \{2,3\}$. However, as any $\vec{y}\in \overline{\mathcal{D}}$ such that $C(\vec{y})=1$ has $y_2<0.6-\eps_0<0.6-\eps$, the triangle formed by $\vec{x}^{(1)},\vec{x}^{(2)},\vec{x}^{(3)}$ cannot be trichromatic. 
        \item If $x^{(1)}_3 \leq 0.1$, we have $C(\vec{x}^{(1)})\in \{1,2\}$. Note that any $\vec{y}\in \overline{\mathcal{D}}$ such that $C(\vec{y})=3$ and $y_3\leq 0.1+\eps$ has $y_2\notin [0.4,0.6]$. If the triangle formed by $\vec{x}^{(1)}, \vec{x}^{(2)}, \vec{x}^{(3)}$ is trichromatic, we should further have $x_2^{(1)}, x_2^{(2)}, x_2^{(3)}\leq 0.4+\eps$ or have $x_2^{(1)}, x_2^{(2)}, x_2^{(3)}\geq 0.6-\eps$. However, as discussed above, we cannot find $C(\vec{x}^{(j)})=2$ with the first condition and we cannot find $C(\vec{x}^{(j)})=1$ with the second condition. The triangle formed by $\vec{x}^{(1)},\vec{x}^{(2)},\vec{x}^{(3)}$ cannot be trichromatic. 
        \item If $x_3^{(1)}\geq 0.3$, we have $C(\vec{x}^{(1)})=3$. Note that there is no $\vec{y}\in \overline{\mathcal{D}}$ such that $y_3>0.2$. The triangle formed by $\vec{x}^{(1)},\vec{x}^{(2)},\vec{x}^{(3)}$ cannot be trichromatic. 
    \end{enumerate}

\subsection{Proof of \cref{lem:lipschitz-of-the-converter-on-boundaries}}
\label{proof:lipschitz-of-the-converter-on-boundaries}

\repii*

According to the characterization \cref{lem:symmetry-on-converted-coordinates}, for any $z\in [0,1]$, we have
\begin{align*}
    \rel(1-z,0,z),~\rel(0,1-z,z) ~\in~
    \begin{cases}
        \hfill \{0,1\} \hfill & \text{if $z\leq 0.1-2\eps^2$~;} \\
        \hfill \settxt{0} \hfill & \text{if $z\in (0.1-2\eps^2, 0.1-\eps^2]$~;} \\
        \hfill \set{0.5\eps^{-2}\cdot(z-0.1+\eps^2)} \hfill & \text{if $z\in 0.1\pm \eps^2$~;} \\
        \hfill \settxt{1} \hfill & \text{if $z\in [0.1+\eps^2, 0.1+2\eps^2)$~;} \\
        \hfill \{0,1\} \hfill & \text{if $z\geq 0.1+2\eps^2$~.}
    \end{cases}
\end{align*}
That is, when $z\in 0.1\pm 2\eps^2$, we have an exact characterization for it.
Note that the RHS of the characterization is continuous and $0.5\eps^{-2}\cdot (z-0.1+\eps^2)$ is $0.5\eps^{-2}$-Lipschitz. 
We have obtained this lemma for any $z,z'\in 0.1\pm 2\eps^{2}$.

On the other hand, w.l.o.g., suppose that $z\notin 0.1 \pm 2\eps^{2}$.
If $|z-z'|\geq \eps^2$, this lemma is simply trivial because we have $\eps^{-2}\cdot |z-z'|\geq 1$. 
Otherwise, according to the triangle inequality, $z'\notin 0.1\pm \eps^2$. 
Therefore, we have $\rel(\vec{x}), \rel(\vec{x'}) \in \{0,1\}$, which matches the second bullet of this lemma.

\subsection{Proof of \cref{fact:new-coordinate-converter-at-the-bottom-part}}
\label{proof:new-coordinate-converter-at-the-bottom-part}

\repiii*

Note that in our base instance $C$, we have 
\begin{align*}
    \forall \vec{x'}\in \Delta^2, \quad
    \begin{cases}
    C(\vec{y'}) = 1 & \quad \text{if } y'_3\leq 0.1, y'_2 \leq 0.5~,\\
    C(\vec{y'}) = 2 & \quad \text{if } y'_3\leq 0.1, y'_2 > 0.5~.
    \end{cases}
\end{align*}

If $y_2\leq 0.5$, we have $C(\vec{y})=1$. 
$C(\vec{y'})$ is different with $C(\vec{y})$ only if $y'_2>0.5$ or $y'_3>0.1$. 
Since $0.1-y_3>0.04\gg 2\eps^2$, $\vec{y}$ is hot or warm if and only if $\alpha(y_3) \cdot (0.5-y_2)<2\eps^2$, i.e., $g(\vec{y})\in (-2\eps^2,0)$. 
In this case, we have $d^{\alpha}(\vec{y}, \nna(\vec{y})) = \alpha(y_3)\cdot |0.5-y_2| = |g(\vec{y})|$, $C^{\alpha}_{\sf nn}(\vec{y})=2$, and 
\begin{align*}
    \rela(\vec{y}) = \left(0.5-0.5\eps^{-2}\cdot |g(\vec{y})|\right)_+ = \left(0.5+0.5\eps^{-2}\cdot g(\vec{y})\right)_{[0,1]}~.
\end{align*}

On the other hand, if $y_2> 0.5$, we have $C(\vec{y})=2$. 
$C(\vec{y'})$ is different with $C(\vec{y})$ only if $y'_2\leq 0.5$ or $y'_3>0.1$. 
Since $0.1-y_3>0.04\gg 2\eps^2$, $\vec{y}$ is hot or warm if and only if $\alpha(y_3) \cdot (y_2-0.5)<2\eps^2$, i.e., $g(\vec{y})\in (0,2\eps^2)$. 
In this case, we have $d^{\alpha}(\vec{y}, \nna(\vec{y})) = \alpha(y_3)\cdot |0.5-y_2| = |g(\vec{y})|$, $C^{\alpha}_{\sf nn}(\vec{y})=1$, and 
\begin{align*}
    \rela(\vec{y}) = \left(0.5+0.5\eps^{-2}\cdot |g(\vec{y})|\right)_- = \left(0.5+0.5\eps^{-2}\cdot g(\vec{y})\right)_{[0,1]}~.
\end{align*}

\subsection{Proof of \cref{lem:poly-time-recovery-symmetry}}
\label{proof:poly-time-recovery-symmetry}

\repiv*

We assume that $\vec{x}^{(1)}, \vec{x}^{(2)}, \vec{x}^{(3)}$ give a trichromatic triangle in the \outofk Approximate Symmetric \Sperner instance $\Csym{k}$ such that $\normtxt{\vec{x}^{(i)}-\vec{x}^{(j)}}\infty\leq 2^{-4kn}$, where $k\geq 3$.
As in the earlier proof of \cref{lem:poly-time-recovery}, 
we use $\vec{y}^{(i)}(\vec{x}) = (y^{(i)}_1(\vec{x}), y^{(i)}_2(\vec{x}), y^{(i)}_3(\vec{x}))$
to denote the {\em intermediate projections} and use
$\vec{c}^{(i)}(\vec{x}) = (c^{(i)}_1(\vec{x}), c^{(i)}_2(\vec{x}), c^{(i)}_3(\vec{x}))$ 
to denote the {\em intermediate palettes}.
In addition, we introduce a new notation $\tilde{y}_0(\vec{x})$ for the converted coordinate we compute on \cref{line:first-converted-coordinate}.
In the rest of this section, because the subscripts we will use for $\vec{c}^{(i)}$ can be very complicated, we will use $c_j^{(i)}(\vec{x})$ and $c^{(i)}(\vec{x}, j)$ interchangeably for better presentation.
For any vector $\vec{y}\in \Delta^2$, we use $i^*(\vec{y})$ to denote the first non-zero index of $\vec{y}$, i.e., 
\begin{align}
    \label{eqn:first-nonzero-index-symmetry}
    i^*(\vec{y}) = \begin{cases}
        1 & \text{if $y_1>0$,}\\
        2 & \text{if $y_1=0$ and $y_2>0$,}\\
        3 & \text{otherwise.}
    \end{cases}
\end{align}

Our recovery algorithm (formalized in \cref{alg:recover-2d-sol-symmetry}) simulates \cref{alg:approx-symmetric-sperner-kgeq3} for each $\vec{x}^{(j)}$.
At time $i\in [2,k-1]$ when we have computed $\vec{y}^{(i)}(\vec{x}^{(j)})$ for each $j\in [3]$, we examine whether one of them is inside a trichromatic region, i.e., whether $C$ has three different colors in $\mathcal{N}(\vec{y}^{(i)}(\vec{x}^{(j)}))$ for some $j\in [3]$.
According to our earlier discussions, this examination can be done in polynomial time.
After we finish the simulation, we simply output $\vec{y}^{(k)}(\vec{x}^{(1)}), \vec{y}^{(k)}(\vec{x}^{(2)}), \vec{y}^{(k)}(\vec{x}^{(3)})$ as a solution for $C$.

It is clear that any output during the simulation phase of this recovery algorithm gives a valid solution for the 2{\rm D}-\Sperner instance $C$. 
To prove \cref{lem:poly-time-recovery-symmetry}, we only need to show that the three intermediate projections after the simulation form a valid solution for $C$ if we output after the simulation phase. 
Or equivalently, if each $\vec{y}^{(i)}(\vec{x}^{(j)})$ does not lie in a trichromatic region, the final converted coordinates, $\vec{y}^{(k)}(\vec{x}^{(1)}), \vec{y}^{(k)}(\vec{x}^{(2)}), \vec{y}^{(k)}(\vec{x}^{(3)})$, give a solution for $C$.

As in \cref{lem:poly-time-recovery}, our proof will be considers two cases; the first is where
the given solution $\vec{x}^{(1)}, \vec{x}^{(2)}, \vec{x}^{(3)}$  for $\Csym{k}$ further enjoys the following property: 
\begin{align}
\label{eqn:assumption-for-simple-proof-of-recovery-symmetry}
\forall 2\leq i\leq k, \forall j\in [3], \quad P_{i+1}^{(k-i)}(\vec{x}^{(j)})\leq 0.9~.
\end{align}
One benefit of first considering this case is that we don't have too crazy projections in \cref{alg:recover-2d-sol-symmetry}, which can help us significantly simplify the proof.
The formal intermediate technical result is presented in \cref{cor:no-sols-in-sim-implies-final-ones-symmetry}.
Later, in the second step, we will explain how to prove \cref{lem:poly-time-recovery-symmetry} in the  case where this property is not satisfied.

\paragraph{Case 1: the output satisfies \cref{eqn:assumption-for-simple-proof-of-recovery-symmetry}.}
We can continue to use the Lipschitzness of the projection step (\cref{lem:lipschitz-of-projection}) to obtain Lipschitzness for the intermediate projections.

\begin{lemma}
\label{lem:intermediate-projections-are-close-to-each-other-symmetry}
    Consider any $k\geq 2$ and any $\vec{x}^{(1)}, \vec{x}^{(2)}\in \Delta^k$.
    Suppose that $P_{i+1}^{(k-i)}(\vec{x}^{(j)})\leq 0.9$ for any $2\leq i\leq k$ and any $j\in [2]$.
    Also, suppose that $\mathcal{N}(\vec{y}^{(i)}(\vec{x}^{(j)}))$ is at most bichromatic for any $2\leq i<k$ and any $j\in [2]$.
    Then, for any $2\leq i\leq k$, we have Lipschitzness for the third coordinate of the intermediate projections:
    \begin{align*}
        \abstxt{y^{(i)}_3(\vec{x}^{(1)})-y^{(i)}_3(\vec{x}^{(2)})}\leq 2^{O(k)}\cdot \normtxt{\vec{x}^{(1)}-\vec{x}^{(2)}}\infty~.
    \end{align*}
    Furthermore, at least one of the following is satisfied for the second coordinates of the intermediate projections:
    \begin{itemize}[itemsep=0.1em, topsep=0.4em]
        \item {\bf Lipschitz:} $\abstxt{y_2^{(i)}(\vec{x}^{(1)})-y_2^{(i)}(\vec{x}^{(2)})}\leq 2^{3in+O(k)}\cdot \normtxt{\vec{x}^{(1)}-\vec{x}^{(2)}}\infty$, or
        \item {\bf both are on the left/right boundaries:} for any $j\in [2]$, we have either $y_1^{(i)}(\vec{x}^{(j)})=0$ or $y_2^{(i)}(\vec{x}^{(j)})=0$.
    \end{itemize}
\end{lemma}
\begin{proof}
    We prove this lemma by induction on $i$. The base case is when $i=2$. 
    We have $y^{(2)}_3(\vec{x})=P_3^{(k-2)}(\vec{x})$ and thus we can obtain the Lipschitzness for $y^{(2)}_3(\cdot)$ by
    \begin{align*}
        \abs{y_3^{(i)}(\vec{x}^{(1)})-y_3^{(i)}(\vec{x}^{(2)})}
        &\leq 
        \normtxt{\vec{P}^{(k-i)}(\vec{x}^{(1)})-\vec{P}^{(k-i)}(\vec{x}^{(2)})}{\infty}
        \\
        &\leq 
        110\cdot \normtxt{\vec{P}^{(k-i-1)}(\vec{x}^{(1)})-\vec{P}^{(k-i-1)}(\vec{x}^{(2)})}{\infty}
        \\
        &\leq
        \cdots
        \\
        &\leq 
        110^{k-i}\cdot \normtxt{\vec{x}^{(1)}-\vec{x}^{(2)}}{\infty}
        \leq 2^{7k}\cdot \normtxt{\vec{x}^{(1)}-\vec{x}^{(2)}}{\infty}~,
    \end{align*}
    which works for any $i\geq 2$. 
    At the same time, we can similarly obtain $\abstxt{P_2^{(k-1)}(\vec{x}^{(1)})-P_2^{(k-1)}(\vec{x}^{(2)})}\leq 2^{7k}\cdot \normtxt{\vec{x}^{(1)}-\vec{x}^{(2)}}\infty$. 
    Because the function $g(y_0)=(0.5+0.5\eps^{-2}\cdot (y_0-0.1))_{[0,1]}$ is clearly $0.5\eps^{-2}$-Lipschitz, we have
    \begin{align*}
        \abstxt{\tilde{y}_0(\vec{x}^{(1)}) - \tilde{y}_0(\vec{x}^{(2)})} \leq 2^{2n} \cdot \abstxt{P_2^{(k-1)}(\vec{x}^{(1)})-P_2^{(k-1)}(\vec{x}^{(2)})} \leq 2^{2n+7k} \cdot \normtxt{\vec{x}^{(1)}-\vec{x}^{(2)}}\infty~.
    \end{align*}
    Hence, we can the first bullet of the second statement of this lemma. 
    \begin{align*}
        \abstxt{y_2^{(2)}(\vec{x}^{(1)}) - y_2^{(2)}(\vec{x}^{(2)})} \leq \abs{\tilde{y}_0(\vec{x}^{(1)})-\tilde{y}_0(\vec{x}^{(2)})} + \abs{y^{(2)}_3(\vec{x}^{(1)})-y^{(2)}_3(\vec{x}^{(2)})} \leq 2^{2n+7k+1}\cdot \normtxt{\vec{x}^{(1)}-\vec{x}^{(2)}}\infty~.
    \end{align*}

    Consider any $i_0\geq 3$. 
    Suppose that we have proved this lemma for $i<i_0$. 
    Next, we consider when $i=i_0$.
    Note that we have $y_3^{(i)}(\vec{x})=P_{i+1}^{(k-i)}(\vec{x})$ for any $\vec{x}$.
    By the premise that $P_{i+1}^{(k-i)}(\vec{x}^{(j)})\leq 0.9$, we can use \cref{lem:lipschitz-of-projection} to get
    the first statement of this lemma by the earlier inequalities. 

    Next, we establish the second statement of this lemma, in which we need to prove at least one of the bullets is satisfied.
    For the induction hypothesis, we will use the following more specific version of the first bullet 
    \[
        \abstxt{y_2^{(i)}(\vec{x}^{(1)})-y_2^{(i)}(\vec{x}^{(2)})}\leq 2^{3in+7k+\log  2i}\cdot \normtxt{\vec{x}^{(1)}-\vec{x}^{(2)}}\infty~.
    \]
    Note that $y_2^{(i)}(\vec{x})=(1-y_3^{(i)}(\vec{x}))\cdot \rela(\vec{y}^{(i-1)}(\vec{x}))$. 
    Then, it is easy to obtain that
    \begin{align*}
        \abs{y_2^{(i)}(\vec{x}^{(1)})-y_2^{(i)}(\vec{x}^{(2)})} 
        &= 
        \abs{(1-y_3^{(i)}(\vec{x}^{(1)}))\cdot \rela(\vec{y}^{(i-1)}(\vec{x}^{(1)})) - (1-y_3^{(i)}(\vec{x}^{(2)}))\cdot \rela(\vec{y}^{(i-1)}(\vec{x}^{(2)}))}
        \\
        &\leq
        \rela(\vec{y}^{(i-1)}(\vec{x}^{(1)}))\cdot \abs{y_3^{(i)}(\vec{x}^{(1)})-y_3^{(i)}(\vec{x}^{(2)})} 
        \\
        & \qquad \quad + (1-y_3^{(i)}(\vec{x}^{(2)})) \cdot \abs{\rela(\vec{y}^{(i-1)}(\vec{x}^{(1)})) - \rela(\vec{y}^{(i-1)}(\vec{x}^{(2)})) }
        \\
        &\leq
        \abs{y_3^{(i)}(\vec{x}^{(1)})-y_3^{(i)}(\vec{x}^{(2)})} + \abs{\rela(\vec{y}^{(i-1)}(\vec{x}^{(1)})) - \rela(\vec{y}^{(i-1)}(\vec{x}^{(2)})) }
        \\
        &\leq 
        2^{7k}\cdot \normtxt{\vec{x}^{(1)}-\vec{x}^{(2)}}\infty + \abs{\rela(\vec{y}^{(i-1)}(\vec{x}^{(1)})) - \rela(\vec{y}^{(i-1)}(\vec{x}^{(2)})) }~.
    \end{align*}

    Suppose that our induction hypothesis gives the first bullet for $i-1$.
    Combining the Lipschitzness on the third coordinate, we have 
    \begin{align*}
        \normtxt{\vec{y}^{(i-1)}(\vec{x}^{(1)})-\vec{y}^{(i-1)}(\vec{x}^{(2)})}\infty &\leq (2^{3(i-1)n+7 k + \log 2(i-1)}+2^{7k}) \cdot \normtxt{\vec{x}^{(1)}-\vec{x}^{(2)}}\infty~. 
    \end{align*}
    If $\rela(\vec{y}^{(i-1)}(\vec{x}^{(1)})), \rela(\vec{y}^{(i-1)}(\vec{x}^{(2)})) \in \{0,1\}$, we have the second bullet for $i=i_0$.
    Otherwise, according to \cref{lem:lipschitz-of-the-converter-for-symmetry}, we have 
    \begin{align*}
        \abs{y_2^{(i)}(\vec{x}^{(1)})-y_2^{(i)}(\vec{x}^{(2)})} 
        &\leq 
        2^{7k}\cdot \normtxt{\vec{x}^{(1)}-\vec{x}^{(2)}}\infty + 2^{3n} \cdot \abs{\vec{y}^{(i-1)}(\vec{x}^{(1)}) - \vec{y}^{(i-1)}(\vec{x}^{(2)})}
        \\
        &\leq
        2^{7k}\cdot \normtxt{\vec{x}^{(1)}-\vec{x}^{(2)}}{\infty} + 2^{3n}\cdot (2^{3(i-1)n+7k+\log 2(i-1)}+2^{7k}) \cdot \normtxt{\vec{x}^{(1)}-\vec{x}^{(2)}}{\infty}
        \\
        &\leq
        2^{7k}\cdot \normtxt{\vec{x}^{(1)}-\vec{x}^{(2)}}{\infty} + (2i-1)\cdot 2^{3in+7k}\cdot \normtxt{\vec{x}^{(1)}-\vec{x}^{(2)}}{\infty}
        \\
        &\leq 
        2^{3in+7k + \log 2i}\cdot \normtxt{\vec{x}^{(1)}-\vec{x}^{(2)}}{\infty} ~,
    \end{align*}
    which gives the first bullet for $i=i_0$.

    On the other hand, suppose that the induction hypothesis further gives us the second bullet for $i-1$. 
    That is, for any $j\in [2]$, we have 
    \[
        \vec{y}^{(i-1)}(\vec{x}^{(j)}) \in \set{\left(0,\ 1-y_3^{(i-1)}(\vec{x}^{(j)}),\ y_3^{(i-1)}(\vec{x}^{(j)})\right), ~~ \left(1-y_3^{(i-1)}(\vec{x}^{(j)}),\ 0,\ y_3^{(i-1)}(\vec{x}^{(j)})\right)}~.
    \]
    If $\rela(\vec{y}^{(i-1)}(\vec{x}^{(1)})), \rela(\vec{y}^{(i-1)}(\vec{x}^{(2)})) \in \{0,1\}$, we have the second bullet for $i=i_0$.
    Otherwise, according to \cref{lem:lipschitz-of-the-converter-on-boundaries-for-symmetry}, we have 
    \begin{align*}
        \abs{y_2^{(i)}(\vec{x}^{(1)})-y_2^{(i)}(\vec{x}^{(2)})} 
        &\leq 
        2^{7k}\cdot \normtxt{\vec{x}^{(1)}-\vec{x}^{(2)}}\infty + 2^{2n} \cdot \abs{y_3^{(i-1)}(\vec{x}^{(1)}) - y_3^{(i-1)}(\vec{x}^{(2)})}
        \\
        &\leq
        2^{7k}\cdot \normtxt{\vec{x}^{(1)}-\vec{x}^{(2)}}{\infty} + 2^{2n+7k} \cdot \normtxt{\vec{x}^{(1)}-\vec{x}^{(2)}}{\infty}
        \\
        &\leq 
        2^{2in+7k+1}\cdot \normtxt{\vec{x}^{(1)}-\vec{x}^{(2)}}{\infty} ~, 
    \end{align*}
    which gives the first bullet for $i=i_0$.
\end{proof}

Suppose that \cref{alg:recover-2d-sol-symmetry} fails to give us any solution during the simulation phase, i.e., for any $2\leq i\leq k-1$ and $j\in [3]$, $\mathcal{N}(\vec{y}^{(i)}(\vec{x}^{(j)}))$ is at most bichromatic. 
Since the solution of the \outofk Approximate Symmetric \Sperner instance satisfies $\normtxt{\vec{x}^{(j_1)}-\vec{x}^{(j_2)}}{\infty}\leq 2^{-4kn}$, we have $\abstxt{y_3^{(i)}(\vec{x}^{(j_1)})-y_3^{(i)}(\vec{x}^{(j_2)})} < 2^{-5n}=\eps^5$ for any $2\leq i\leq k$ and any $j_1,j_2\in [3]$, and further that
\begin{itemize}[itemsep=0.2em, topsep=0.5em]
    \item {\bf the intermediate projections are close to each other:} for any $j_1,j_2\in [3]$,$\abstxt{y_2^{(k)}(\vec{x}^{(j_1)})-y_2^{(k)}(\vec{x}^{(j_2)})}<2^{-2n}=\eps^2$ and $\abstxt{y_2^{(i)}(\vec{x}^{(j_1)})-y_2^{(i)}(\vec{x}^{(j_2)})}<2^{-5n}=\eps^5$ for $i<k$; or
    \item {\bf all intermediate projections are on the left/right boundaries:} for any $j\in [3]$, we have either $y_1^{(i)}(\vec{x}^{(j)})=0$ or $y_2^{(i)}(\vec{x}^{(j)})=0$.
\end{itemize}

Next, we prove the same characterization of the intermediate palettes used in \cref{alg:approx-symmetric-sperner-kgeq3} by \cref{lem:characterization-of-intermediate-colors-symmetry}.
The characterization gives equivalence between the set of relevant colors in the palette.
In the first case, where the intermediate projections of the three input vectors are close to each other (and at least one of them is not on the left/right boundaries), the palettes are exactly the same.
In the second case, where the intermediate projections of the input vectors are all on the left/right boundaries, the only two relevant colors for the points, $c^{(i)}(\vec{x}, i^*(\vec{y}^{(i)}(\vec{x})))$ and $c^{(i)}_3(\vec{x})$, are respectively equal. 

\begin{lemma}
    \label{lem:characterization-of-intermediate-colors-symmetry}
    Suppose that $P_{i+1}^{(k-i)}(\vec{x}^{(j)})\leq 0.9$ for any $2\leq i\leq k$ and any $j\in [3]$.
    Suppose that $\mathcal{N}(\vec{y}^{(i)}(\vec{x}^{(j)}))$ is not trichromatic for any $2\leq i<k$ and $j\in [3]$. 
    For any $2\leq i\leq k$, at least one of the following holds:
    \begin{itemize}
        \item The intermediate projections are close to each other, and we have that the corresponding palettes are the same: $\vec{c}^{(i)}(\vec{x}^{(j_1)})=\vec{c}^{(i)}(\vec{x}^{(j_2)})$ for any $j_1,j_2\in [3]$.
        \item All intermediate projections are on the left/right boundaries, and the palettes may be different on an irrelevant color, but we still have that both of the following hold: 
        \begin{itemize}
            \item The color of the first non-zero coordinate $c^{(i)}\left(\vec{x}^{(j)}, i^*(\vec{y}^{(i)}(\vec{x}^{(j)}))\right)$ is the same across all  $j\in [3]$; and
            \item the 3rd color is the same, $c^{(i)}_3(\vec{x}^{(j)})=i+1$ for any $j\in [3]$.
        \end{itemize} 
    \end{itemize}
\end{lemma}

Since the output $\vec{x}^{(1)}, \vec{x}^{(2)}, \vec{x}^{(3)}$ is trichromatic in $\Csym{k}$, we have 
\begin{align}
    \label{eqn:final-trichromatic-symmetry}
    \abs{\set{c^{(k)}(\vec{x}^{(j)},C(\vec{y}^{(k)}(\vec{x}^{(j)})))}_{j\in [3]}}=3~.
\end{align}
We should always have the first bullet of \cref{lem:characterization-of-intermediate-colors-symmetry} for $k$ when each $\mathcal{N}(\vec{y}^{(i)}(\vec{x}^{(j)}))$ is not trichromatic, because otherwise the second bullet of \cref{lem:characterization-of-intermediate-colors-symmetry} violates \cref{eqn:final-trichromatic-symmetry} as $C(\vec{y}^{(k)}(\vec{x}^{(j)}))\in \{i^*(\vec{y}^{(k)}(\vec{x}^{(j)})), 3\}$.
Note that
the first bullet of \cref{lem:characterization-of-intermediate-colors} and \cref{eqn:final-trichromatic} imply that the colors in the base instance $C(\vec{y}^{(k)}(\vec{x}^{(1)}))$, $C(\vec{y}^{(k)}(\vec{x}^{(2)}))$, and $C(\vec{y}^{(k)}(\vec{x}^{(3)}))$ should be distinct. 
We can always guarantee that the tuple $\vec{y}^{(k)}(\vec{x}^{(1)}), \vec{y}^{(k)}(\vec{x}^{(2)}), \vec{y}^{(k)}(\vec{x}^{(3)})$ gives a solution to $C$.
\begin{corollary}
    \label{cor:no-sols-in-sim-implies-final-ones-symmetry}
    Suppose that $P_{i+1}^{(k-i)}(\vec{x}^{(j)})\leq 0.9$ for any $2\leq i\leq k$ and any $j\in [3]$.
    Suppose that $\mathcal{N}(\vec{y}^{(i)}(\vec{x}^{(j)}))$ is not trichromatic for any $2\leq i<k$ and $j\in [3]$. 
    Then, we have $|\settxt{C(\vec{y}^{(k)}(\vec{x}^{(j)})): j\in [3]}|=3$.
\end{corollary}

Before giving the proof of \cref{lem:characterization-of-intermediate-colors-symmetry}, we give an useful lemma here. The lemma states that any point very close (e.g., $<\eps^3$) to a hot point should be either hot or warm. 
This could resolve the problem arisen from the fact that $d^{\alpha}(\cdot, \cdot)$ is only quasimetric and does not enjoy the triangle inequality. 
\begin{lemma}
\label{lem:close-to-hot-points}
    For any $\vec{x}$ that is hot and any $\vec{x'}$ such that $d(\vec{x}, \vec{x'})\leq \eps^3$, $\vec{x'}$ is either hot or warm. 
\end{lemma}
\begin{proof}
    If $\vec{x}, \vec{x'}$ has different colors, both of them are hot. 
    Next, we consider that $\vec{x}, \vec{x'}$ have the same color.
    Let $\vec{y}=\nna(\vec{x})$ denote the nearest neighbor of $\vec{x}$.
    If $x_3,x'_3\geq 0.05$, because of \cref{lem:same-def-for-converted-coordinate} and the fact $d(\cdot,\cdot)$ satisfies the triangle inequality, the lemma holds as 
    \begin{align*}
        d^\alpha(\vec{x'},\vec{y}) = d(\vec{x'},\vec{y}) \leq d(\vec{x}, \vec{x'}) + d(\vec{x}, \vec{y}) = d(\vec{x}, \vec{x'}) + d^{\alpha}(\vec{x}, \vec{y}) < \eps^3 + \eps^2 < 2\eps^2~.
    \end{align*}
    Further, we consider that $x_3,x'_3\leq 0.06$ to finish the proof (because we have $|x_3-x_3'|\leq \eps^3$). 
    Since $\vec{x}$ is hot, we have $d^{\alpha}(\vec{x}, \vec{y})<\eps^2$, and thus
    \begin{align*}
        \alpha(x_3) \cdot |x_2-y_2|, |x_3-y_3| < \eps^2~.
    \end{align*}
    Because $d(\vec{x}, \vec{x'})\leq \eps^3$, we have $|x_2-x'_2|, |x_3-x'_3|\leq \eps^3$, and thus 
    \begin{align*}
        \alpha(x'_3) \cdot |x'_2-y_2| & \leq |\alpha(x_3)-\alpha(x'_3)| + \alpha(x_3) \cdot |x'_2-y_2|\\
        &\leq O(n)\cdot |x_3-x'_3| + \alpha(x_3)\cdot (|x'_2-x_2| + |x_2-y_2|) \tag{\cref{lem:shrinking-factor-basics}}\\
        &\leq O(n)\cdot \eps^3 + \eps^3 + \eps^2 < 2\eps^2~,\\
        |x'_3-y_3| & \leq |x_3-x'_3| + |x_3-y_3| < \eps^3 +\eps^2 < 2\eps^2~.
    \end{align*}
    We have $d^{\alpha}(\vec{x'}, \vec{y})<2\eps^2$. Because $\vec{y}$ (or, an infinite sequence of points that converges to $\vec{y}$) has a different color with $\vec{x},\vec{x'}$, $\vec{x'}$ is hot or warm.
\end{proof}

\begin{proof}[Proof of \cref{lem:characterization-of-intermediate-colors-symmetry}]
    We prove this lemma by induction.
    The base case is when $i=2$. 
    At the beginning of \cref{alg:approx-symmetric-sperner-kgeq3}, we have $\normtxt{\vec{y}^{(2)}(\vec{x}^{(j_1)})-\vec{y}^{(2)}(\vec{x}^{(j_2)})}\infty\leq 2^{O(k)}\cdot 2^{-4kn}<\eps^5$ for any $j_1,j_2\in [3]$ according to the proof of \cref{lem:intermediate-projections-are-close-to-each-other-symmetry} and $\vec{c}^{(2)}(\vec{x}^{(j)})=(1,2,3)$ for any $j\in [3]$. 
    We have the first bullet satisfied for $i=2$. 

    Consider any $i_0\geq 3$. Assume that we have established this lemma for any $i=i_0-1$. 
    Next, we establish this lemma for $i=i_0$.

    First, consider that the first bullet holds for $i-1$. We have $\abstxt{y_2^{(i-1)}(\vec{x}^{(j_1)})-y_2^{(i-1)}(\vec{x}^{(j_2)})}< 2^{-5n}=\eps^5$ for any $j_1,j_2\in [3]$. 
    Hence, $d(\vec{y}^{(i-1)}(\vec{x}^{(j_1)}), \vec{y}^{(i-1)}(\vec{x}^{(j_2)}))\leq \eps^5$ for any $j_1,j_2\in [3]$. 
    We discuss two cases on whether there is a hot point in $\vec{y}^{(i-1)}(\vec{x}^{(1)}), \vec{y}^{(i-1)}(\vec{x}^{(2)})$ and $\vec{y}^{(i-1)}(\vec{x}^{(3)})$. 
    \begin{itemize}[itemsep=0.2em,topsep=0.5em]
        \item W.l.o.g., suppose that $\vec{y}^{(i-1)}(\vec{x}^{(1)})$ is a hot point. 
        According to \cref{lem:lipschitz-of-the-converter-for-symmetry} and \cref{line:y_1,line:y_2,line:y_3}, all the $i$-th intermediate projections are close to each other. 
        According to \cref{lem:close-to-hot-points}, $\vec{y}^{(i-1)}(\vec{x}^{(2)})$ and $\vec{y}^{(i-1)}(\vec{x}^{(2)})$ are also hot or warm.
        Since they are all in a region that is at most bichromatic, 
        the following color set $\{C(\vec{y}^{(i-1)}(\vec{x}^{(j)})), \hat{C}^\alpha_{\sf nn}(\vec{y}^{(i-1)}(\vec{x}^{(j)}))\}$ is then the same across all $j\in [3]$.
        Since $\vec{c}^{(i-1)}(\vec{x}^{(1)})=\vec{c}^{(i-1)}(\vec{x}^{(2)})=\vec{c}^{(i-1)}(\vec{x}^{(3)})$, we then have the same $\vec{c}^{(i)}(\vec{x}^{(j)})$ across all $j\in [3]$, which gives the first bullet of this lemma. 
        \item Otherwise, suppose that $\vec{y}^{(i-1)}(\vec{x}^{(j)})$ is warm or cold for each $j\in [3]$. We have $\rela(\vec{y}^{(i-1)}(\vec{x}^{(j)}))\in \{0,1\}$ for any $j\in [3]$ in this case.
        The $i$-th intermediate projections are on the left/right boundaries, i.e., we have $y_1^{(i)}(\vec{x}^{(j)})=0$ or $y_2^{(i)}(\vec{x}^{(j)})=0$ for each $j\in [3]$.
        Note that there is only one color in $\settxt{C(\vec{y}^{(i-1)}(\vec{x}^{(j)}))}_{j\in [3]}$, because otherwise we clearly have $d^\alpha(\vec{y}^{(i-1)}(\vec{x}^{(j)}),\nna(\vec{y}^{(i-1)}(\vec{x}^{(j)})))<\eps^2$ and all points are hot. 
        Note that \cref{alg:approx-symmetric-sperner-kgeq3} ensures in this scenario that 
        \begin{align*}
            i^*(\vec{y}^{(i)}(\vec{x}^{(j)}))=1 
            &\Leftrightarrow 
            \rela(\vec{y}^{(i-1)}(\vec{x}^{(j)}))=0\\
            &\Leftrightarrow
            C(\vec{y}^{(i-1)}(\vec{x}^{(j)}))<\hat{C}^\alpha_{\sf nn}(\vec{y}^{(i-1)}(\vec{x}^{(j)}))\\
            &\Leftrightarrow
            c^{(i)}(\vec{x}^{(j)},1) = c^{(i-1)}(\vec{x}^{(j)}, C(\vec{y}^{(i-1)}(\vec{x}^{(j)})))
        \end{align*}
        where the second {\em iff} is obtained by \cref{fact:modified-neighboring-color-for-symmetry}.
        We have $c^{(i)}(\vec{x}^{(j)}, i^*(\vec{y}^{(i)}(\vec{x}^{(j)}))) = c^{(i-1)}(\vec{x}^{(j)}, C(\vec{y}^{(i-1)}(\vec{x}^{(j)})))$ for each $j\in [3]$.
        Since $\vec{c}^{(i-1)}(\vec{x}^{(1)})=\vec{c}^{(i-1)}(\vec{x}^{(2)})=\vec{c}^{(i-1)}(\vec{x}^{(3)})$ and $C(\vec{y}^{(i-1)}(\vec{x}^{(1)}))=C(\vec{y}^{(i-1)}(\vec{x}^{(2)}))=C(\vec{y}^{(i-1)}(\vec{x}^{(3)}))$, 
        we have the second bullet for $i$.
    \end{itemize}

    Second, consider that the second bullet holds for $i-1$. 
    Because $|y_3^{(i-1)}(\vec{x}^{(j_1)})-y_3^{(i-1)}(\vec{x}^{(j_2)})|<\eps^{5}$ for any $j_1,j_2\in [3]$, according to the characterization of the temperature on the left/right boundaries (\cref{lem:symmetry-on-converted-coordinates-for-symmetric-sperner}), we have 
    \begin{itemize}[itemsep=0.2em,topsep=0.5em]
        \item for each $j\in [3]$, $\vec{y}^{(i-1)}(\vec{x}^{(j)})$ is hot or warm; or 
        \item for each $j\in [3]$, $\vec{y}^{(i-1)}(\vec{x}^{(j)})$ is warm or cold. 
    \end{itemize}
    Since we have $C(1-z,0,z)=3=C(0,1-z,z)$ or $C(1-z,0,z)=1, C(0,1-z,z)=2$ for any $z\in [0,1]$ (\cref{lem:base-instance-lr-boundaries}), we have the same $c^{(i-1)}(\vec{x}^{(j)},C(\vec{y}^{(i-1)}(\vec{x}^{(j)})))$ across all $j\in [3]$.
    Next, we discuss the above two cases to finish the proof. 
    \begin{itemize}[itemsep=0.2em,topsep=0.5em]
        \item Consider when each $\vec{y}^{(i-1)}(\vec{x}^{(j)})$ is hot or warm. We have $\hat{C}^\alpha_{\sf nn}(\vec{y}^{(i-1)}(\vec{x}^{(j)}))=C^\alpha_{\sf nn}(\vec{y}^{(i-1)}(\vec{x}^{(j)}))$ for each $j\in [3]$. Because we have $C^\alpha_{\sf nn}(1-z,0,z)=1, C^\alpha_{\sf nn}(0,1-z,z)=2$, or $C^\alpha_{\sf nn}(1-z,0,z)=3=C^\alpha_{\sf nn}(0,1-z,z)$ for each $z\in 0.1\pm 2\eps^2$ (\cref{lem:symmetry-on-converted-coordinates-for-symmetric-sperner}), we have the same $c^{(i-1)}(\vec{x}^{(j)},\hat{C}^\alpha_{\sf nn}(\vec{y}^{(i-1)}(\vec{x}^{(j)})))$ across all $j\in [3]$.
        Therefore, the palette $\vec{c}^{(i)}(\vec{x}^{(j)})$ is the same across all $j\in [3]$. 
        Because of \cref{lem:intermediate-projections-are-close-to-each-other-symmetry} and the Lipschitzness of the coordinate converter on the left/right boundaries \cref{lem:lipschitz-of-the-converter-on-boundaries-for-symmetry}, the $i$-th intermediate projections, $\vec{y}^{(i)}(\vec{x}^{(1)}), \vec{y}^{(i)}(\vec{x}^{(2)})$ and $\vec{y}^{(i)}(\vec{x}^{(3)})$, are close to each other. 
        Hence, we prove the first bullet of this lemma for $i$. 
        \item Consider when each $\vec{y}^{(i-1)}(\vec{x}^{(j)})$ is warm or cold. We have $\rela(\vec{y}^{(i-1)}(\vec{x}^{(j)}))\in \{0,1\}$ and all $i$-th intermediate projections are on the left/right boundaries. 
        Note that there is only one color in $\settxt{C(\vec{y}^{(i-1)}(\vec{x}^{(j)}))}_{j\in [3]}$, because otherwise we clearly have $d^\alpha(\vec{y}^{(i-1)}(\vec{x}^{(j)}),\nn(\vec{y}^{(i-1)}(\vec{x}^{(j)})))<\eps^2$ and all points are hot. 
        According to our earlier discussions, \cref{alg:approx-symmetric-sperner-kgeq3} ensures in this scenario that $c^{(i)}(\vec{x}^{(j)}, i^*(\vec{y}^{(i)}(\vec{x}^{(j)}))) = c^{(i-1)}(\vec{x}^{(j)}, C(\vec{y}^{(i-1)}(\vec{x}^{(j)})))$.
        Therefore, we have the same $c^{(i)}(\vec{x}^{(j)}, i^*(\vec{y}^{(i)}(\vec{x}^{(j)})))$ across all $j\in [3]$, and thus the second bullet of this lemma holds. \qedhere
    \end{itemize}
\end{proof}

\paragraph{Case 2: the output does not satisfy~\cref{eqn:assumption-for-simple-proof-of-recovery-symmetry}.}
Next, we complete our second step by showing how to prove \cref{lem:poly-time-recovery-symmetry} without the property (\cref{eqn:assumption-for-simple-proof-of-recovery-symmetry}). 
Suppose $\theta^*$ is the minimum threshold $\theta\geq 2$ such that $P_{i+1}^{(k-i)}(\vec{x}^{(j)})\leq 0.9$ for any $\theta\leq i\leq k$ and $j\in[3]$. 
Such threshold always exists because otherwise we have $x^{(1)}_{k+1}, x^{(2)}_{k+1}, x^{(3)}_{k+1}\geq 0.8$. 
This implies $y_3^{(k)}(\vec{x}^{(j)})>0.8$ for any $j\in [3]$, and according to our construction of the base instance (\cref{alg:base-instance}), we have $C^{(k)}(\vec{x}^{(j)})=c^{(k)}_3(\vec{x}^{(j)}) = k+1$ for any $j\in[3]$, violating the assumption that $\vec{x}^{(1)}, \vec{x}^{(2)}, \vec{x}^{(3)}$ form a solution for the \outofk Approximate Symmetric \Sperner problem. 
When $\theta^*=2$, it is equivalent with the special case satisfying \cref{eqn:assumption-for-simple-proof-of-recovery-symmetry} and we have proved \cref{lem:poly-time-recovery-symmetry} for this case.
On the other hand, if $\theta^*>2$, we have $P^{(k-\theta^*+1)}_{\theta^*}(\vec{x}^{(j)})\geq 0.8$ for any $j\in [3]$ because of the Lipschitzness of the projection step \cref{lem:lipschitz-of-projection} and that $\normtxt{\vec{x}^{(i)}-\vec{x}^{(j)}}{\infty}\leq 2^{-4kn}$ for any $i,j\in [3]$. 
This means that we have $C(\vec{y}^{(\theta^*-1)}(\vec{x}^{(j)}))=3$.
And because of \cref{lem:trivial-converter-for-symmetry}, we have $\rela(\vec{y}^{(\theta^*-1)}(\vec{x}^{(j)}))=1$ and $y_1^{(\theta^*)}(\vec{x}^{(j)})=0$ for any $j\in [3]$. 
Further, we have for each $j\in [3]$,
\begin{align}
    \label{eqn:y-theta-star}
    \vec{y}^{(\theta^*)}(\vec{x}^{(j)})=\left(0,\ 1-P^{(k-\theta^*)}_{\theta^*+1}(\vec{x}^{(j)}),\ P^{(k-\theta^*)}_{\theta^*+1}(\vec{x}^{(j)})\right)~.
\end{align}
Therefore, running \cref{alg:approx-symmetric-sperner-kgeq3} on instance $C^{(k)}$ for $\vec{x}^{(1)},\vec{x}^{(2)}, \vec{x}^{(3)}$ is equivalent to running \cref{alg:approx-symmetric-sperner-kgeq3} on instance $C^{(k')}$ with $k'=k-\theta^*+2$ for $\vec{\hat{x}}^{(1)},\vec{\hat{x}}^{(2)},\vec{\hat{x}}^{(3)}\in \Delta^{k'}$ such that
\begin{align*}
    \forall i\in [k'+1], j\in [3], \quad \hat{x}_i^{(j)} = \begin{cases}
        0 & \text{if $i=1$,}\\
        1-\sum_{i'=\theta^*+1}^{k+1} x^{(j)}_{i'} & \text{if $i=2$,}\\
        x^{(j)}_{i+\theta^*-2} & \text{if $i>2$.}
    \end{cases}
\end{align*}
This is because $\vec{y}^{(2)}(\vec{\hat{x}}^{(j)})$ equals the RHS of \cref{eqn:y-theta-star}.
Because this new instance has a smaller number of dimensions and satisfies the condition of our special cases (\cref{eqn:assumption-for-simple-proof-of-recovery-symmetry}), we complete the proof for \cref{lem:poly-time-recovery-symmetry}.

\end{document}